\documentclass[runningheads]{llncs}

\usepackage{balance,tikz}
\usetikzlibrary{arrows, automata,shapes}
\usetikzlibrary{decorations.pathreplacing}
\usepackage{booktabs}
\usepackage{colortbl}
\usepackage[ruled]{algorithm2e}
\usepackage[mathscr]{eucal}
\usetikzlibrary{decorations.pathreplacing}
\usepackage{hyperref,multirow,xcolor}
\usepackage{comment,amsmath,amssymb,amstext,amsfonts,graphics,enumerate,boxedminipage}
\usepackage{color,graphicx,latexsym,amsfonts,epsfig,eufrak,array,boxedminipage,graphicx,bm}

\newcommand{\BibTeX}{\rm B\kern-.05em{\sc i\kern-.025em b}\kern-.08em\TeX}
\newcommand{\NP}{{\sf NP}}
\newcommand{\coNP}{{\sf coNP}}
\newcommand{\XP}{{\sf XP}}
\newcommand{\sP}{{\sf P}}
\newcommand{\R}{\mathbb{R}}
\DeclareMathOperator{\frc}{frac}
\newcommand{\minitab}[2][l]{\begin{tabular}{#1}#2\end{tabular}}

\usepackage{boxedminipage}

\newcommand{\problemdef}[3]{
        \begin{center}
                \begin{boxedminipage}{.96\textwidth}
                        \textsc{{#1}}\\[2pt]
                        \begin{tabular}{ r p{0.77\textwidth}}
                                \textit{~~~~Instance:} & {#2}\\
                                \textit{Question:} & {#3}
                        \end{tabular}
                \end{boxedminipage}
        \end{center}
}

\newtheorem{open}{Open Problem}

\begin{document}

\title{Computing Balanced Solutions for Large International Kidney Exchange Schemes When Cycle Length Is Unbounded\thanks{A 3-page extended abstract of this paper appeared in the proceedings of AAMAS 2024~\cite{BBCJPX24}.}}
\titlerunning{Computing Balanced Solutions for Large International Kidney Exchange Schemes}

\author{M\'arton Benedek\inst{1,2} \and
P\'eter Bir\'o\inst{2,1} \and
Gergely Cs\'aji\inst{2,3}\and\\
Matthew Johnson\inst{4}\and
Dani\"el Paulusma\inst{4} \and
Xin Ye\inst{4}}

\authorrunning{M. Benedek et al.}
\institute{Corvinus University of Budapest, Budapest, Hungary \email{marton.benedek@uni-corvinus.hu} \and
KRTK, Institute of Economics, Budapest, Hungary \email{\{peter.biro,csaji.gergely\}@krtk.hu}\and 
Eötvös Loránd University, Budapest, Hungary\and
Department of Computer Science, Durham University, Durham, UK \email{\{matthew.johnson2,daniel.paulusma,xin.ye\}@durham.ac.uk}}

\maketitle
\begin{abstract}
{In kidney exchange programmes, patients with incompatible donors obtain kidneys via cycles of transplants. Countries may merge their national patient-donor pools to form international programmes. To ensure fairness, a credit-based system is used: a cooperative game-theoretic solution concept prescribes a ``fair'' initial allocation, which is adjusted with accumulated credits to form a target allocation. The objective is to maximize the number of transplants while staying close to the target allocation.
When only $2$-cycles are permitted, a solution that lexicographically minimizes deviations from the target can be found in polynomial time. \textcolor{black}{However, e}ven the problem of maximizing the number of transplants is \NP-hard for larger upper bounds on cycle length. This latter problem is tractable when cycle lengths are not bounded. We formalize this setting via a new class of cooperative games called \emph{partitioned permutation games}, and prove that computing an optimal solution that is lexicographically closest to the target allocation is \NP-hard. We give a randomized XP-time algorithm for solve this problem exactly. We present an experimental study, simulating programmes with up to 10 countries.  Allowing unbounded cycle lengths increases the number of transplants by up to 46\% compared to $2$-cycles. Using credits and selecting lexicographically closest solutions yields low total relative deviation (below 2\% for all fairness notions). Among the seven fairness notions tested, a modified Banzhaf value performs best in balancing fairness and efficiency, achieving average deviations below 0.65\%. Lexicographic minimization from the target allocation leads to significantly ($36$-$56$\%) smaller average deviations than minimizing
maximum difference only.}
 
\end{abstract}

\section{Introduction}\label{s-intro}

We introduce a new class of cooperative games called {\it partitioned permutation games}, which are closely related to the known classes of permutation games~\cite{TPPR84} and partitioned matching games~\cite{BBKPP}. Partitioned matching games and partitioned permutation games have immediate applications in international kidney exchange. Before defining these games, we first explain this application area.

\smallskip
\noindent
{\bf Kidney Exchange.}
The most effective treatment for kidney failure is transplanting a kidney from a deceased or living donor, with better long-term outcomes in the latter case. However, a kidney from a family member or friend might be medically incompatible and could easily be rejected by the patient's body. Therefore, many countries run national Kidney Exchange Programmes (KEPs); {see} 
\cite{Bi_etal2019,Bi_etal2021} {and
the recent survey 
\cite{BCDMP25}}. 
In a KEP, all patient-donor pairs are placed together in one pool.
If for two patient-donor pairs $(p,d)$ and $(p',d')$, it holds that $d$ and $p$ are incompatible as well as $d'$ with $p'$, {but $d$ and $p'$ are compatible as well as $d'$ with $p$,}
then $d$ could donate a kidney to $p'$, and $d'$ could donate a kidney to $p$. This is a {\it $2$-way exchange}.

We {can} generalize a $2$-way exchange {as follows}. We model a pool of patient-donor pairs as a directed graph $G=(V,A)$ (the {\it compatibility graph}) in which
each vertex of $V$ is a patient-donor pair,
and $A$ consists of every arc $(u,v)$ such that the donor of~$u$ is compatible with the patient of~$v$. In a directed cycle $C=u_1u_2\ldots u_ku_1$, for some $k\geq 2$, the kidney of the donor of~$u_i$ could be given to the patient of~$u_{i+1}$ for every $i\in \{1,\ldots,{k}\}$, with $u_{k+1}:=u_1$. This is a {\it $k$-way exchange} using the {\it exchange cycle}~$C$.

To prevent exchange cycles from breaking (and a patient from losing their willing donor), hospitals perform the $k$ transplants in a $k$-way exchange simultaneously. Hence,
KEPs impose a bound $\ell$, called the {\it exchange bound}, on the maximum {\it length} (number of 
{arcs}) 
of an exchange cycle, typically $2\leq \ell\leq 5$.  An  {\it $\ell$-cycle packing} of $G$ is
a set~${\mathcal C}$ of directed cycles, each of length at most~$\ell$, that are pairwise vertex-disjoint; if $\ell=\infty$, we also say that ${\mathcal C}$ is a {\it cycle packing}.
The {\it size} of ${\mathcal C}$ is 
is the sum of the lengths of the cycles of ${\cal C}$.

KEPs operate in rounds.
A solution for round~$r$ is an $\ell$-cycle packing in the corresponding compatibility graph $G^r$.  The goal is to help as many patients as possible in each round. Hence, to maximize the number of transplants in round~$r$, we seek an {\it optimal} solution, that is, a {\it maximum} $\ell$-cycle packing of $G^r$, i.e., one that has maximum size.
After round~$r$, some patients will have received a kidney or died 
{(or have left the pool for some other reason)}, while other patient-donor pairs may have arrived. This results in a new compatibility graph $G^{r+1}$ for round~$r+1$.
{We have now finished the description of the basic KEP model, as used in this paper.
As discussed by Bir\'o et al.~\cite{Bi_etal2019},
there exist many variants and extensions, some of which we will discuss in more detail later.}

{An important}
computational issue for KEPs is how to find an optimal solution in each round. If $\ell=2$, we can {apply a well-known transformation (see e.g.~\cite{RSU05}) that only keeps the ``double'' arcs of a compatibility graph $G$. Namely} we transform $G$ into an undirected graph 
${\hat{G}}=(V,E)$
as follows. For every $u,v\in V$, we have $uv\in E$ if and only if $(u,v)\in A$ and $(v,u)\in A$. It then remains to compute a maximum matching in {$\hat{G}$}, 
which can be done in polynomial time~\cite{Ed65a}.
If we set $\ell=\infty$, {another} well-known trick works (see e.g.~\cite{ABS07}). We transform $G$ into a bipartite graph $H$ with partition classes $V$ and $V'$, where $V'$ is a copy of $V$. For each $u\in V$ and its copy $u'\in V'$, we add the edge~$uu'$ with weight~$0$. For each $(u,v)\in A$, we add the edge~$uv$ with weight~$1$. Now it remains to find in polynomial time a maximum weight perfect matching in~$H$. However, for any constant $\ell\geq 3$, the 
complexity changes, as shown by  Abraham, Blum and Sandholm~\cite{ABS07}.

\begin{theorem}[\cite{ABS07}]\label{t-hard}
If $\ell=2$ or $\ell=\infty$, we can find an optimal solution for a KEP round in polynomial time; else this is \NP-hard.
\end{theorem}

\noindent
{\bf International Kidney Exchange.}
As merging pools of national KEPs leads to better outcomes, 
many countries today work together with the aim of forming
 an {\it international} KEP (IKEP), e.g. Austria and the Czech Republic~\cite{Bo_etal17}; Denmark, Norway and Sweden; and Italy, Portugal and Spain~\cite{Va_etal19}.  Apart from ethical, legal and logistical issues (all beyond our scope), there is now a new and highly non-trivial issue that needs to be addressed: {\it How can we ensure long-term stability of an IKEP?} If countries are not treated {\it fairly}, they may leave the IKEP. In a worst-case scenario it could even happen that in the end each country 
returns to
their own national KEP~again.

\medskip
\noindent
{\it Example {1}.} Let $G$ be the compatibility graph from Figure~\ref{f-thirdleft}. Then a total of five kidney transplants is possible (only) if we choose the solution  of size~$5$ that consists of the single exchange cycle $C=abdeca$.
We have that patient-donor pairs $a,b,c$ belong to country~$1$, whereas patient-donor pairs $d$ and $e$ belong to country~$2$, and patient-donor pair~$f$ belongs to country~$3$. So using this solution, countries $1$, $2$ and $3$ receive three, two and zero kidney transplants, respectively. Note that without cooperation only country~$2$ can receive kidney transplants in this example.\qed

\begin{figure}[t]
\vspace*{-0.3cm}
\centering
{\begin{tikzpicture}[scale=0.8,rotate=0]
\draw
(0, 1) node[circle, black, scale=0.8,draw](a){\small$a$}
(1.5, 2) node[circle, black,scale=0.8, draw](b){\small$b$}
(1.5, 0) node[circle, black, scale=0.8,draw](c){\small$c$}
(3, 3) node[circle, black, scale=0.8,draw](d){\small$d$}
(3, 1) node[circle, black, scale=0.8,draw](e){\small$e$}
(3, -1) node[circle, black, scale=0.8,draw](f){\small$f$}
(-0.5, 2) node[scale=0.8](h){\small$V_{1}$}
(3.7, 3) node[scale=0.8](i){\small$V_{2}$}
(3.7, -1) node[scale=0.8](j){\small$V_{3}$};
\draw[->,line width=1.5pt] (a) -> (b);
line width=1.5pt\draw[->,line width=1.5pt] (d) ->  (e);
\draw[->] (e) to [out=30,in=-30]  (d);
\draw[->] (b) ->  (e);
\draw[->,line width=1.5pt] (c) ->  (a);
\draw[->, line width=1.5pt](b) -> (d);
\draw[->,line width=1.5pt] (e) -> (c);
\draw[->] (c) -> (f);
\draw[->] (e) ->  (f);
\draw[dashed] (0.97,1) ellipse (1.3 and 1.5);
\draw[dashed] (3,2) ellipse (0.6 and 1.4);
\draw[dashed] (3,-1) ellipse (0.35 and 0.35);
\end{tikzpicture}}
\caption{A {partitioned} permutation game $(N,v)$ of width $c=3$ defined on a graph $G=(V,A)$. Note that $N=\{1,2,3\}$ and $V=V_1\cup V_2\cup V_3$ with $V_1=\{a,b,c\}$, $V_2=\{d,e\}$ and $V_3=\{f\}$, as indicated by the dotted circles. {Recall that $G$ has exactly one maximum cycle packing ${\cal C}$, which consists of the single exchange cycle $C=abdeca$, as indicated by the thick edges.
Selecting ${\cal C}$ yields the allocation  $x=(3,2,0)$.
We note that
$v(N)=v(\{1,2\})=5$, while $v(\{2\})=2$ and $v(\{1\})=v(\{3\})=v(\{1,3\})=v(\emptyset)=0$ and $v(\{2,3\})=2$. It can be readily checked that for all $S\subseteq N$, $x(S)\geq v(S)$. Hence, $x$ belongs to the core of $(N,v)$.}}\label{f-thirdleft}
\vspace*{-0.5cm}
\end{figure}
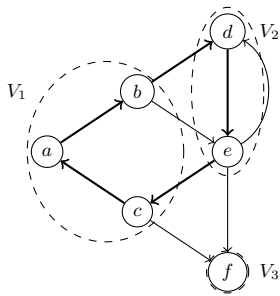

\medskip
\noindent
{\bf Cooperative Game Theory}
considers
 {\it fair} distributions of joint profit if all parties involved collaborate
 (the notion of fairness depends on context).
 Before describing its role in our setting, we first give some
 {basic} terminology {on cooperative game theory};
{we refer to~\cite{Peters2008} for more details}.
 
 A \emph{(cooperative) game} is a pair $(N,v)$, where $N$ is a set of $n$ \emph{players} and $v: 2^N\to \R$
 is a \emph{value function} with $v(\emptyset) = 0$.  A subset $S\subseteq N$ is a {\it coalition}. If for every possible partition $(S_1,\ldots,S_r)$ of $N$ it holds that $v(N)\geq v(S_1)+\cdots +v(S_r)$, then players will benefit most by forming the {\it grand coalition} $N$. The problem is then how to fairly distribute $v(N)$ amongst the players of $N$. An {\it allocation}  is a vector $x \in \R^N$ with $x(N) = v(N)$ (we write $x(S)=\sum_{p\in S}x_p$ for $S\subseteq N$). 
{In this context, the notion of fairness is defined by}
 a {\it solution concept}, {which} prescribes a set of 
fair allocations for a game $(N,v)$.
{Each solution concept has its own fairness properties; we discuss these properties in detail in Section~\ref{s-inin}.}

The solution concepts 
considered in this
paper are the Shapley value,
Banzhaf value (two variants)
nucleolus, 
 tau value,
benefit value,
contribution value (all of which are defined in Section~{\ref{s-inin}}) 
and the core. They all prescribe a unique allocation except for the {\it core}, which consists of all allocations $x \in \R^N$ with $x(S)\geq v(S)$ for every $S\subseteq N$.
 Core allocations ensure $N$ is stable {in the sense that} no subset $S$ will benefit from forming their own coalition. {However, the} core {of a cooperative game} may be~empty.

We now define some relevant games.  For a directed graph $G=(V,A)$ and subset $S\subseteq V$, we let $G[S]=(S,\{(u,v)\in A\; |\; u,v\in S\})$ be the subgraph of $G$ {\it induced by $S$}. An {\it $\ell$-permutation game} on a  directed graph~$G=(V,A)$ is the game $(N,v)$, where $N=V$ and for $S\subseteq N$, the value~$v(S)$ is the maximum size of an $\ell$-cycle packing of~$G[S]$. Two special cases are well studied. We obtain a {\it matching game} if $\ell=2$, which may have an empty core (e.g. when $G$ is a triangle on vertices $a,b,c$ with arcs $(a,b)$, $(b,a)$, $(a,c)$, $(c,a)$, $(b,c)$, $(c,b)${, see also Figure~\ref{empty-core}}) and a {\it permutation game} if $\ell=\infty$, whose core is always nonempty~\cite{TPPR84}.

In the remainder, we differentiate between the sets $N$ (of players in the game) and $V$ (of vertices in the underlying graph $G$). That is, we associate each player~$i\in N$ with a distinct subset~$V_i$ of~$V$.
 A {\it partitioned $\ell$-permutation game} on a directed graph $G=(V,A)$ with a partition $(V_1,\ldots,V_n)$ of $V$ is the game $(N,v)$, where $N=\{1,\ldots,n\}$, and for $S\subseteq N$, the value $v(S)$ is the maximum size of an $\ell$-cycle packing of $G[\bigcup_{i\in S}V_i]$.
We obtain a {\it partitioned matching game}~\cite{BBJPX23,BBKPP} if $\ell=2$, and a {\it partitioned permutation game} if $\ell=\infty$. The {\it width} of $(N,v)$ is the {\it width} of $(V_1,\ldots,V_n)$, which is defined as $c=\max\{|V_i| \; |\; 1\leq i\leq n\}$. 
{We refer again to Figure~\ref{f-thirdleft} for an example.}

\medskip
\noindent
{\bf The Model.} For a round of an IKEP with exchange bound~$\ell$, let $(N,v)$ be the partitioned $\ell$-permutation game defined on the compatibility graph $G=(V,A)$, where $N=\{1,\ldots,n\}$ is the set of countries in the IKEP, and $V$ is partitioned into subsets $V_1,\ldots,V_n$ such that for every $p\in N$, $V_p$ consists of the patient-donor pairs of country~$p$. We say that $(N,v)$ is the {\it associated} game for $G$.
We can now make use of a solution concept ${\mathcal S}$ for $(N,v)$ to obtain a fair {\it initial allocation}~$y$, where $y_p$ prescribes the initial number of kidney transplants country~$p$ should receive in this round (possibly, $y_p$ is not an integer, but as we shall see this is not relevant).

To ensure IKEP stability, we use the model of Klimentova et al.~\cite{KNPV20}, which is a {\it credit-based} system. For round $r\geq 1$, let $G^r$ be the compatibility graph with associated game $(N,v^r)$; let $y^r$ be the initial allocation (as prescribed by some solution concept ${\mathcal S}$); and let $c^r:N\to \R$ be a {\it credit function}, which satisfies $\sum_{p\in N}c^r_p=0$; if $r=1$, we set $c^r\equiv 0$.
 For $p\in N$, we set $x^r_p:=y^r_p+c^r_p$ to obtain the {\it target allocation} $x^r$  for round~$r$; note that $x^r$ is indeed an allocation, as $y^r$ is an allocation and $\sum_{p\in N}c_p^r=0$). We choose some maximum $\ell$-cycle packing ${\mathcal C}$ of $G^r$ as optimal solution for round~$r$ (out of possibly exponentially many optimal solutions).
{For round~$r$}, let $s_p({\mathcal C})^{{r}}$ be the number of kidney transplants for patients in country~$p$ (with donors both from~$p$ and other countries). For $p\in N$, we set $c^{r+1}_p:=x^r_p-s_p({\mathcal C}^{{r}})$ to get the credit function $c^{r+1}$ for round $r+1$ (note that $\sum_{p\in N}c_p^{r+1}=0$). For round~$r+1$,
a new initial allocation~$y^{r+1}$ is prescribed by~${\mathcal S}$ for the associated game $(N,v^{r+1})$. For every $p\in N$, we set $x_p^{r+1}:=y_p^{r+1}+c_p^{r+1}$, and we repeat the process.

Apart from specifying the solution concept~${\mathcal S}$, we must also determine how to choose in each round a maximum $\ell$-cycle packing~${\mathcal C}$ (optimal solution) of the corresponding compatibility graph~$G$. We will choose~${\mathcal C}$, such that the vector $s({\mathcal C})$, with {number of kidney transplants} entries $s_p({\mathcal C})$,  is {\it closest} to the target allocation~$x$ for the round under consideration.
{To ensure \emph{(long-term) stability} of an IKEP, we aim to make an IKEP {\it balanced}, that is, in each round
the goal is to find optimal cycle packings that are close to the
target allocations, yielding
consistently low (ideally zero) credit values.}

To explain our distance measures, let  $|x_p-s_p({\mathcal C})|$ be the {\it deviation} of country~$p$ from its target $x_p$ if ${\mathcal C}$ is chosen out of all optimal~solutions. We order the deviations $|x_p-s_p({\mathcal C})|$ non-increasingly as a vector $$d({\mathcal C})= (|x_{p_1}-s_{p_1}({\mathcal C})|, \dots, |x_{p_n}-s_{p_n}({\mathcal C})|).$$
We say {that} ${\mathcal C}$ is {\it strongly close} to~$x$ if $d({\mathcal C})$ is lexicographically minimal over all optimal solutions. If we only minimize $d_1({\mathcal C})=\max_{p\in N} \{|x_p-s_p({\mathcal C})|\}$ over all optimal solutions, we obtain a {\it weakly close} optimal solution. If an optimal solution is strongly close, 
{then} it is weakly close, but the reverse might not be true. Both measures have been used {in the literature}, {as we discuss} below 
{after first giving an example.}

\medskip
\noindent
{\it{Example 2.}}
{Let $\ell=2$ and consider the partitioned matching game $(N,v)$
from Figure~\ref{empty-core}, which has width $c=1$. Note that $N=\{1,2,3\}$ and $V=V_1\cup V_2\cup V_3$ with $V_1=\{a\}$, $V_2=\{b\}$ and $V_3=\{c\}$. Moreover, $v(N)=v(\{1,2,3\}=v(\{1,2\})=v(\{1,3\})=v(\{2,3\})=2$, while $v(\{1\})=v(\{2\})=v(\{3\})=v(\emptyset)=0$. Say we are in round~1 and use the Shapley value (see  Section~\ref{s-simul}), which yields the initial allocation $x^1=\{\frac{2}{3}, \frac{2}{3},\frac{2}{3}\}$.
As $c^1=(0,0,0)$, we have $y^1=x^1$ as target allocation.
Note that $y^1$ is not in the core, because the core of $(N,v)$ is empty, as we observed before.
There are three optimal solutions, which each correspond to a cycle packing consisting of exactly one $2$-vertex cycle. By symmetry, each of them is strongly close to $x$. Suppose we select ${\cal C}=\{\{aba\}\}$. Note that $s_1({\cal C})=s_2({\cal C})=1$ and $s_3({\cal C})=0$. Hence, $d({\cal C})=(\frac{1}{3},\textcolor{black}{\frac{1}{3},\frac{2}{3}})$, and the credit vector $c^2$ for round~$2$ is 
$(-\frac{1}{3},-\frac{1}{3},\frac{2}{3})$. \qed}

\begin{figure}[t]
\vspace*{-0.3cm}
\centering
{\begin{tikzpicture}[scale=0.8,rotate=0]
\draw
(0, 1.5) node[circle, black, scale=0.8,draw](a){\small$a$}
(2.5, 3) node[circle, black,scale=0.8, draw](b){\small$b$}
(2.5, 0) node[circle, black, scale=0.8,draw](c){\small$c$}
(-0.5, 2) node[scale=0.8](h){\small$V_{1}$}
(3.2, 3) node[scale=0.8](i){\small$V_{2}$}
(3.2, 0) node[scale=0.8](j){\small$V_{3}$};
\draw[->] (a) -> (b);
\draw[->] (b) to [out=-180,in=60] (a);
\draw[->] (b) -> (c);
\draw[->] (c) to [out=60,in=-60] (b);
\draw[->,thick] (c) ->  (a);
\draw[->,thick] (a) to [out=-60,in=180] (c);
\draw[dashed] (0,1.5) ellipse (0.4 and 0.4);
\draw[dashed] (2.5,3) ellipse (0.4 and 0.4);
\draw[dashed] (2.5,0) ellipse (0.4 and 0.4);
\end{tikzpicture}}
\caption{{A partitioned matching game $(N,v)$ of width $c=1$ on the graph $G=(V,A)$ with $N=\{1,2,3\}$ and $V=V_1\cup V_2\cup V_3$ for $V_1=\{a\}$, $V_2=\{b\}$ and $V_3=\{c\}$. Note that  $v(\{1\})=v(\{2\})=v(\{3\})=v(\emptyset)=0$, whereas $v(N)=v(\{1,2,3\}=v(\{1,2\})=v(\{1,3\})=v(\{2,3\})=2$, implying the core of $(N,v)$ is empty.}}
\label{empty-core}
\vspace*{-0.5cm}
\end{figure}

\subsection*{Related Work}

Benedek et al.~\cite{BBKPP} proved the following theorem for $\ell=2$ (the ``matching'' case).

\begin{theorem}[\cite{BBKPP}]\label{t-easy}
For partitioned matching games, the problem of finding an optimal solution that is strongly close to a given target allocation~$x$ is polynomial-time solvable.
\end{theorem}
{In contrast to Theorem~\ref{t-easy}, 
Benedek et al.~\cite{BBKPP} also showed that the problem becomes \NP-hard for {\it edge-weighted} graphs, in which the edges $e\in A$ have an associated weight $w(e)$ reflecting
the expected utility of the corresponding kidney transplant.
Namely, if maximum weight matchings are taken as optimal solutions for edge-weighted graphs, then it is \NP-hard to find a weakly close optimal solution even for $n=2$.}

Benedek et al.~\cite{BBKP22} used the algorithm of Theorem~\ref{t-easy} to perform simulations for up to 
15 countries for $\ell=2$.  As initial allocations, they used the Shapley value, nucleolus, benefit value and contribution value, with the Shapley value yielding the best results. 
Afterwards, Benedek et al.~\cite{BBPY24} extended the experimental results from~\cite{BBKP22} for $\ell=2$ by also considering two variants of the Banzhaf value and the tau value. It turned out that the Banzhaf value variants behaved slightly better than the Shapley value.

The good performance of the Shapley value in~\cite{BBKP22,BBPY24} is in line with the simulation results of Klimentova et al.~\cite{KNPV20} and Bir\'o et al.~\cite{BGKPPV20} for $\ell=3$. 
Due to Theorem~\ref{t-hard}, the simulations in~\cite{BGKPPV20,KNPV20} are for up to four countries and use weakly close optimal solutions.
{We note that the simulations in~\cite{BGKPPV20,KNPV20} model IKEPs that allow for {\it non-directed} donors. 
These are altruistic donors that are added to the pool without a corresponding patient. By allowing non-directed donors, optimal solutions for compatibility graphs may include {\it exchange chains}, which are directed paths that start from a non-directed donor. This may result in more kidney transplants.}

For $\ell=2$, we refer to~\cite{STW21} for an alternative model based on so-called selection ratios using lower and upper target numbers. IKEPs have also been modelled as non-cooperative games in the {\it consecutive matching} setting. In this setting, each round consists of two phases: national pools in phase~1 and a merged pool for unmatched patient-donor pairs in phase~2; 
see~\cite{CL23,CLPV17,SBS} 
for some results in this setting.

{Finally,} fairness (versus optimality) issues are also studied for national KEPs, in particular in the US. 
The US {setting} is different from {the} Europe{an setting, which is} the setting we consider{, for the following reasons}.
{As we explained above, in Europe the treatment of patient-donor pairs is regulated centrally implying that all patient-donor pairs are registered at an (I)KEP, leading to large and diverse
patient-donor pair pools with many exchange options and matching runs that are conducted typically every three months.
In contrast, in the US the transplant centers (hospitals) are independent. They mostly treat incoming patient-donor pairs in-house, and need to be convinced to register patient-donor pairs to one of the three nationwide KEPs (UNOS, APD, and NKR) which compete against each other (see~\cite{AAAFK18}). 
Therefore, in the US, kidney transplants usually take place 
immediately
after the registration of a new pair. 
How to incentivize KEPs in the US setting  
to register all their patient-donors pairs (instead of only the hard-to-match pairs) to a joint pool is a challenging problem that has been widely studied; see e.g.~\cite{AFKP15,AR12,AR14,BCHPPV17,HDHSS15,TP15} for a number of different approaches, which were compared to each other in
the recent survey~\cite{BCDMP25}. In particular, Hajaj at al.~\cite{HDHSS15} suggested a credit-based approach, namely to give the nationwide KEPs positive credits for registering easy-to-match pairs and negative credits for registering hard-to-match pairs. Their approach was further analysed by Agarwal et al.~\cite{AAAFK18} and has been adapted by the NKR but is very different in nature from our credit-based approach.} 

\subsection*{Our Results} 

We consider partitioned permutation games and IKEPs, so we assume $\ell=\infty$. {The latter assumption} is not realistic in kidney exchange. {However, the assumption} has also been made for national KEPs {in order} to obtain general results of theoretical
interest~\cite{ACDM21,BKKV23,BMR09,RSU04}. 
{Moreover, results of studies with $\ell=\infty$} may have wider applications, {for instance}, in portfolio compression in financial clearing~\cite{ER2021}, 
{time exchange~\cite{ACsEE2021}, and
shift-reallocation~\cite{MW2021}}.
{In particular,} we {wanted} to research how {the} stability {of an IKEP} and the total number of kidney transplants~are affected when moving from one extreme ($\ell=2$) to the other ($\ell=\infty$). 

{As mentioned, the KEP model we use is the most basic version, which has been used in the literature for conducting
computer simulations on generated instance. In particular, this allows us to make a clean comparison with a previous experimental study for $\ell=2$~\cite{BBPY24}.
That is, we ignore the diverse sets of complex, hierarchical optimization criteria that are used in European kidney 
exchange~\cite{Bi_etal2019,Bi_etal2021} and only maximize the number of transplants in each matching run. We do note that this objective is always the first objective in every used hierarchy apart from the hierarchy used in the UK, where it is the second objective.
We also ignore the presence of non-directed donors
and the possibility of having multiple donors registering for one patient. To further justify our choice of not considering non-directed donors, we note that such donors are illegal in some 
European countries, such as France and Hungary.}

{Our} paper consists of a theoretical part and an experimental part. We start with our theoretical results (Section~\ref{s-theory}).
Permutation games, i.e. partitioned permutation games of width~$1$, have a nonempty core~\cite{TPPR84}, and a core allocation can be found in polynomial
time~\cite{CT86}.
We generalize these two results to partitioned permutation games of any width~$c$, and also show a {complexity} dichotomy for testing core membership, which is in contrast with the 
{complexity}
dichotomy for
partitioned matching games, where the complexity jump happens at $c=3$~\cite{BBKPP}.

\begin{theorem}\label{t-core}
The core of every partitioned permutation game is nonempty, and it is possible to find a core allocation in polynomial time.
Moreover, for partitioned permutation games of fixed width~$c$, the problem of deciding if an allocation is in the core
is polynomial-time solvable if $c=1$ and co\NP-complete if $c\geq 2$.
\end{theorem}

\noindent
Due to Theorem~\ref{t-hard}, we cannot hope to generalize Theorem~\ref{t-easy} to hold for any constant $\ell\geq 3$. 
Nevertheless, Theorem~\ref{t-hard} leaves open the question of whether there is an analogue of Theorem~\ref{t-easy} for $\ell=\infty$, the ``cycle packing'' case; Theorem~\ref{t-easy} is concerned only with $\ell=2$, the ``matching'' case. 
We show that the answer to this question is no (assuming $\sP\neq \NP$).

\begin{theorem}\label{t-ppg}
For partitioned permutation games even of width~$2$, the problem of finding an optimal solution that is weakly or strongly close to a given target allocation~$x$ is \NP-hard.
\end{theorem}

\noindent
{We refer to Table~\ref{dichotomies} for an overview of all the known and new results that we discussed above.}

Our last theoretical result is a randomized \XP\ algorithm with parameter~$n$. As we shall prove, derandomizing it requires solving the notorious {\sc Exact Perfect Matching} problem in polynomial time.
The complexity status of the latter problem is still open since its introduction by Papadimitriou and Yannakakis~\cite{PY82} in 1982.

\begin{theorem}\label{t-interval}
For a partitioned permutation game $(N,v)$ on a directed graph $G=(V,A)$, the problem of finding an optimal solution that is weakly or strongly close to a given target allocation~$x$ can be solved by a randomized algorithm in $|A|^{O(n)}$ time.
\end{theorem}

\noindent
Our theoretical results highlighted severe computational limitations, and we now turn to our simulations. We perform a large-scale experimental study, in which our simulations are strongly guided by our theoretical results.
Namely, we note that the algorithm in Theorem~\ref{t-interval} is not practical for instances of realistic size (which we aim for).
Moreover, it is a randomized algorithm. Hence, it it not acceptable either in the setting of kidney exchange, as policy makers would only use solution yielding the maximum number of kidney transplants.
Therefore and also due to Theorem~\ref{t-ppg},
 we formulate the problems of computing a weakly or strongly close optimal solution as integer linear programs, as described in Section~\ref{a-ilp}. This enables us to use an ILP solver. We still exploit the fact that for $\ell=\infty$ (Theorem~\ref{t-hard}) we can find optimal solutions and values~$v(S)$ in polynomial time. In this way we can still perform, in Section~\ref{s-simul},  simulations for IKEPs up to {\it ten} countries, so more than the four countries in the simulations for $\ell=3$~\cite{BGKPPV20,KNPV20}, but less than the 
15 countries in the simulations for $\ell=2$~\cite{BBPY24}.

For the initial allocations we use two easy-to-compute solution concepts: the benefit value and contribution value and four hard-to-compute solution concepts: the Banzhaf value, Shapley value, nucleolus and tau-value.
{In particular, it follows from a well-known alternative definition of the nucleolus~\cite{MPS79} that the nucleolus is a core allocation for every game with a nonempty core. As by Theorem~\ref{t-core}, every partitioned permutation game has a nonempty core, this means that we do not need to consider the core as a separate solution concept in our simulations.}
{Finally,} we use the modified Banzhaf value called Banzhaf* value from~\cite{BBPY24}, where the credits are incorporated directly into the coalitional values. As in~\cite{BBPY24}, we research both the case where all countries have the same size and the case where countries can be either of small, medium or large size.

Our simulations show, as in~\cite{BBPY24},  that a credit system using strongly close optimal solutions instead of weakly optimal solutions makes an IKEP up to  
5{6}\% 
more balanced\footnote{{As discussed in Section~\ref{s-simul}, using the Banzhaf value with $n=10$ and equal country sizes yields an average relative deviation of 1.97\% when selecting weakly close optimal solutions and an average relative deviation of 0.86\% when selecting strongly close optimal solutions. So, the improvement in balancedness is $(1.97-0.86)/1.97=56\%$. For the other solution concepts, the improvements are between 36\% and~54\%.}} without decreasing the overall number of transplants. 
Our simulations {also} indicate that {out of the seven solution concepts, using} the Banzhaf* value 
{leads to the smallest deviations},
namely on average, a deviation of at most 0.90\%  from the 
initial allocation (but the differences between the Banzhaf* value {and} the Banzhaf value {or the} Shapley value are small).
Moreover, moving from $\ell=2$ to $\ell=\infty$ yields on average 46\% more kidney transplants
(using the same simulation instances generated by the data generator~\cite{PT21}). 

{For our simulations, we also kept track of the length of the exchange cycles. It turned out that these may be very large, in particular in the starting round. Previously, the relation between increase in transplants versus increase in cycle length was also researched by Bir\'o~\cite{BMR09}, namely for $\ell=2$, $\ell=3$, and $\ell=\infty$. However, their simulations were done in the context of the UK KEP only.}

\begin{table}[t]
\centering
\begin{tabular}{lllll}
\toprule
                                        & $\ell=2$    & & $\ell\in \{3,4,\ldots\}\;\;$     & $\ell=\infty$    \\
\midrule
finding an optimal solution                        & poly~\cite{ABS07}&         & \NP-h~\cite{ABS07}      & poly~\cite{ABS07}             \\[4pt]
testing core nonemptyiness                                    & poly if $c\leq 2$& & \NP-h & poly~\cite{TPPR84}
\\
& co\NP-hard if $c\geq 3$~\cite{BBKPP} 
\\[4pt]
finding a core allocation                                    & poly if $c\leq 2$& & \NP-h & poly~\cite{CT86}
\\
& co\NP-h if $c\geq 3$~\cite{BBKPP} 
\\[4pt]
deciding core membership                                   & poly if $c\leq 2$& & \NP-h & \textcolor{purple}{poly if $c=1$}\\
& co\NP-c if $c\geq 3$~\cite{BBKPP} 
&&&\textcolor{purple}{co\NP-c if $c\geq 2$}
\\[4pt]
finding a weakly close solution to target allocation   & poly~\cite{BBKPP}&         & \NP-h      & \textcolor{purple}{\NP-h}          \\[4pt]
finding a strongly close solution to target allocation & poly~\cite{BBKPP}&         & \NP-h      & 
\textcolor{purple}{\NP-h }         \\
\bottomrule
\end{tabular}
\smallskip
\caption{{Survey of complexity results relevant for IKEPs, for fixed cycle length $\ell=2$, $\ell \in \{3,4,\ldots\}$ and 
$\ell=\infty$, respectively, where $c$ is the width of the associated partitioned $\ell$-permutation game $(N,v)$. Here, poly stands for polynomial-time solvable; co\NP-c for co\NP-complete; \NP-h for \NP-hard; and \coNP-h means \coNP-hard. Results in \textcolor{purple}{purple} are new results shown in this paper, whereas non-referenced results follow directly from the referenced results.}} \label{dichotomies}
\end{table}

\section{{The Proofs of Our} Theoretical Results}~\label{s-theory}
We start with Theorem~\ref{t-core}.
Recall that the width~$c$ of a partitioned permutation game $(N,v)$ defined on a directed graph $G=(V,A)$ with vertex partition $(V_1,\ldots,V_n)$  is the maximum size of a set  $V_i$.

\medskip
\noindent
{\bf Theorem~\ref{t-core} (restated).}
{\it The core of every partitioned permutation game is nonempty, and it is possible to find a core allocation in polynomial time.
Moreover, for partitioned permutation games of fixed width~$c$, the problem of deciding if an allocation is in the core
is polynomial-time solvable if $c=1$ and co\NP-complete if $c\geq 2$.}

\begin{proof}
We first show that finding a core allocation of a partitioned permutation game can be reduced to finding a core allocation of a permutation game.    As the latter can be done in polynomial-time~\cite{CT86} (and such a core allocation always exists~\cite{TPPR84}),  the same holds for the former.

    Let $(N,v)$ be a partitioned permutation game on a graph $G=(V,A)$ with partition $(V_1,\ldots,V_n)$ of $V$. We create a permutation game $(N',v')$ by splitting each $V_i$ into sets of size~$1$, i.e., every vertex becomes a player in $N'$. Let $x'$ be a core allocation of $(N',v')$. For each {player} $i\in N$, we set ${x_i}=\sum_{v\in V_i} {x'_v}$, {that is, we sum the $x'$ values of the players in $N'$ corresponding to the vertices that $i$ controls in $G$.} It holds that $x(N)=v(N)$, as $(N,v)$ and $(N',v')$ are defined on the same graph~$G$; hence, the weight of a maximum weight cycle packing is unchanged.

    Suppose there is a {\it blocking} coalition $S\subset N$, that is, $v(S)>x(S)$ holds. By the construction of~$x$, it holds that the sum of the $x'$ values over all vertices in $\cup_{i\in S}V_i$ is less than $v(S)=v'(\{ u\mid u\in \cup_{i\in S}V_i \} )$. Hence, the players in $N'$ corresponding to these vertices would form a blocking coalition to $x'$ for $(N',v')$, a contradiction. As $x'$ can be found in polynomial-time, so can $x$.

Now we show that deciding whether an allocation is in the core can be solved in polynomial time for partitioned permutation games with width~$1$,
that is, for permutation games.
       Let $x\in \mathbb{R}^N$ be an allocation. 
       We create a weight function $w_x$ over the arcs by setting $w_x((u,v))={x_u}-1$.
    We claim that if there exists a blocking coalition, then there is a blocking coalition that consists of only vertices along a cycle.
    In order to see this, let $S$ be a blocking coalition, so $x(S)<v(S)$. By definition, $v(S)$ is the maximum size of a cycle packing ${\mathcal C}= \{C_1, C_2,..,C_k\}$ in $G[S]$. For $i=1,\ldots,k$, let $S_i$ be the set of vertices in $C_i$. From $$x(S_1)+x(S_2)+ \ldots x(S_k) \leq x(S) < v(S) = |S_1|+|S_2|+\ldots |S_k|,$$ we find that $x(S_i)<|S_i|$ for at least one set $S_i$. Hence, such an $S_i$ is also blocking.

So we just need to check whether there is a cycle {$C_i$ with vertex set $S_i$ such that $x(S_i)<|S_i|$}. 
{Note that $\sum_{(u,v)\in C_i}w_x(u,v)=x(S_i)-|S_i|$. Hence}, 
such a cycle exists if and only if   
 $w_x$ is not {\it conservative}  (where conservative means no negative weighted cycles exist). The latter can be decided in polynomial-time, for example with the Bellman-Ford algorithm.
    
Finally, we show that deciding if an allocation $x$ is in the core of a partitioned permutation game is co\NP-complete, even if each $|V_i|$ has size~$2$ (so $c=2$).
Containment in \coNP\ holds, as we can check in polynomial time if a coalition blocks an allocation. To prove hardness, we reduce from a special case of the \NP-complete problem~\textsc{Exact $3$-Cover}~\cite{GJWZ96,HDRB08}.

\problemdef{{\sc Exact $3$-Cover}}{A family of 3-element subsets of $[3n]$, $\mathcal{S}=\{ S_1,\dots, S_{3n}\}$, where each element belongs to exactly three sets}{Is there an \textit{exact $3$-cover} for $\mathcal{S}$, that is, a subset $\mathcal{S'}\subset \mathcal{S}$ such that each element appears in exactly one of the sets of $\mathcal{S'}$?}

\noindent	
Given an instance $I$ of {\sc Exact $3$-Cover},
we construct a partitioned permutation game $(N,v)$ as follows (see Figure~\ref{fig:coreverif} for an illustration).

\begin{itemize}
    \item [--] For each element $i\in [3n]$, there is a vertex $a_i$ and a vertex~$b_i$,
    \item [--] for each set $S_j\in \mathcal{S}$, there are vertices $s_j^1,s_j^2,s_j^3,t_j^1,t_j^2,t_j^3$,
    \item  [--] there are a further $12n$ vertices $x_1,\dots,x_{6n}$ and $y_1,\dots,y_{6n}$.
\end{itemize}

\noindent
Define the arcs as follows:
\begin{itemize}
\item [--]
for each $k\in [6n]$, an arc $(x_k,x_{k+1})$,
where $x_{6n+1} := x_1$.
\item [--]
for each $k\in [3n]$, an arc $(b_k,b_{k+1})$,
where $b_{3n+1} := b_1$;
 \item [--] for each $j\in [3n]$, the arcs $(t_j^1,t_j^2)$, $(t_j^2,t_j^3)$, $(t_j^3,t_j^1)$;
\item [--] for each $k\in [6n],j\in [3n],l\in [3]$, the arcs $(y_k,s_j^l)$ and $(s_j^l,y_k)$; and
\item [--] for each set $S_j=\{ j_1,j_2,j_3\}$, $j_1<j_2<j_3$, the arcs $(s_j^1,a_{j_1})$, $(a_{j_1},s_j^1)$, $(s_j^2,a_{j_2})$, $(a_{j_2},s_j^2)$, $(s_j^3,a_{j_3})$, $(a_{j_3},s_j^3)$.
\end{itemize}

\noindent
This gives a directed graph
$G=(V,A)$.
The players with their corresponding partition of the vertices are:
\begin{itemize}
\item [--] for each $i\in [3n]$, we have a player $A_i=\{ a_i,b_i\}$,
\item [--] for each $k\in [6n]$, we have a player $X_k=\{ x_k,y_k\}$, and
\item [--] for each $j\in [3n],l\in [3]$,  we have a player $T_j^l=\{ s_j^l,t_j^l\}$.
\end{itemize}

\noindent
Finally, we define the allocation $x{\in \mathbb{R}^N}$, as follows:
\begin{itemize}
\item [--] ${x_{A_i}}=3-\frac{n+1}{9n^2}$ for each $i\in [3n]$,
\item [--] ${x_{X_k}}=3-\frac{2n-1}{18n^2}$ for
each $k\in [6n]$, and
\item [--] ${x_{T_j^l}}=1+\frac{1}{9n}$ for each $j\in [3n],l\in [3]$.
\end{itemize}

\noindent
The size of the maximum cycle packing of $G$ is $$v(N)=6n+6n+9n+9n+3n+3n=36n,$$ as every vertex can be covered. This is realized by adding the $x_k$-cycle, the $b_i$-cycle, the $t_j^l$-cycles and then for each $a_i$, a cycle $\{(a_i,s_j^l),(s_j^l,a_i)\}$ (this can be done, because each element appears exactly three times in the sets, so there is a perfect matching covering the vertex of each element in the bipartite graph induced by the incidence relation between the sets and the elements). We can cover the remaining $6n$ $s_j^l$ vertices by two cycles $\{(y_k,s_j^l),(s_j^l,y_k)\}$ arbitrarily.

If we sum up all allocation values we find that
\[\begin{array}{lcl}
x(N) &= &6n \cdot (3-\frac{2n-1}{18n^2}) + 9n\cdot (1+\frac{1}{9n})+3n\cdot (3-\frac{n+1}{9n^2})\\
&=&18n -\frac{2n-1}{3n}+9n +1 +9n - \frac{n+1}{3n}\\
&=&36n\\
&= &v(N),
\end{array}\]
so $x$ is an allocation for $(N,v)$.

\begin{figure}[t]
  \centering
  \includegraphics[width=0.45\linewidth]{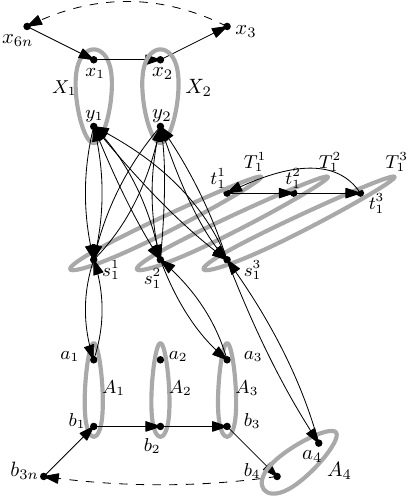}
  \caption{An illustration for Theorem \ref{t-core}. The figure shows the part of the construction for a set $S_1 =\{ 1,3,4\}$. Dashed arcs denote that there is a cycle through the given vertex whose vertices are not included in the picture. Dark grey bold ellipses denote the players.}\label{fig:coreverif}
\end{figure}

We claim $I$ has an exact $3$-cover if and only if $x$ is not in the core.

\medskip
\noindent
``$\Rightarrow$'' First suppose $\{ S_{k_1},\dots,S_{k_n} \}$ is an exact $3$-cover in $I$. We claim that $\mathcal{P}=\{ A_i \mid i\in [3n]\} \cup \{ T_{k_i}^l\mid i\in [n],l\in [3] \}$ is a blocking coalition. We first show that $v(\mathcal{P})=12n$,
which can be seen as follows. First, the $s_j^l$ and $a_i$ vertices in the coalition can be covered by $2$-cycles, as the corresponding sets form an exact $3$-cover. Moreover, the $b_i$ vertices can be covered by the $b_i$-cycle, as each of the $A_i$ players is in $\mathcal{P}$, and finally, the $t_j^l$ vertices can be covered too, as for each $j\in \{ k_1,\dots,k_n\}$ all of $T_j^1,T_j^2,T_j^3$ belong to $\mathcal{P}$. Then, $x$ is not in the core, as
\[\begin{array}{lcl}
    x(\mathcal{P}) &= &3n\cdot (3-\frac{n+1}{9n^2})+3n\cdot (1+\frac{1}{9n})\\ &= &9n-\frac{n+1}{3n}+3n +\frac{n}{3n}\\ &< &12n\\ &= &v({\mathcal P}).\end{array}\]

\medskip
\noindent
``$\Leftarrow$'' Suppose $x$ is not in the core. Then there is a coalition $\mathcal{P}$ with $v(\mathcal{P})>x(\mathcal{P})$.
We write $\mathcal{A}=\bigcup_{i\in [3n]}A_i$ and $\mathcal{X}=\bigcup_{k\in {6n}}X_k$. We claim that $\mathcal{A}\subset \mathcal{P}$ or $\mathcal{X}\subset \mathcal{P}$. For a contradiction, suppose that neither $\mathcal{A}\subset \mathcal{P}$ nor $\mathcal{X}\subset \mathcal{P}$ holds. Clearly, $\mathcal{P}\cap (\mathcal{A}\cup \mathcal{X})\ne \emptyset$, because $m$ of the $T_j^l$ players can only create a cycle packing of size $h$, but $h$ of them have together
	an allocation of $h(1+\frac{1}{9n})>h$.
	So suppose that $|\mathcal{P}\cap (\mathcal{A}\cup \mathcal{X})|=m\le 9n$. By our assumption, none of the vertices $b_i$ or $x_k$ vertices can be covered. Hence, if there are $h \ge 1$ participating $T_j^l$ players besides them (there must be at least one to have any cycles), then the size of the maximum cycle packing they can obtain is $h +2 \min \{ m,h \}\le 2m+h$, as at best all $t_j^l$ vertices can be covered, but the other vertices can only be covered with cycles of length~$2$ by pairing the $h$ $s_j^l$ vertices to $m$ $a_i$ or $x_i$ vertices. But, their assigned allocation in $x$ is at least $$m\cdot \min \{(3-\frac{n+1}{9n^2}),(3-\frac{2n-1}{18n^2})\} + h \cdot (1+\frac{1}{9n})>  2m+h,$$ a contradiction. Hence, $\mathcal{A}\subseteq \mathcal{P}$ or $\mathcal{X}\subseteq \mathcal{P}$ holds.

First suppose that $\mathcal{A}\cup \mathcal{X}\subseteq \mathcal{P}$. Let the number of participating $T_j^l$ players be $h$. Then, we have that $$v(\mathcal{P})\le 3n + 6n + 2h +h,\; \mbox{and}$$ $$v(\mathcal{P})>x(\mathcal{P})=18n-\frac{2n-1}{3n}+9n-\frac{n+1}{3n}+h +\frac{h}{9n}.$$ Hence, we find that $2h >18n-1+\frac{h}{9n}>18n -1$, so $h >9n-1$, but it also cannot be $9n$, as $18n = 18n -1 +\frac{9n}{9n}$, a contradiction (as there are only $9n$ $T_j^l$ players).

    Suppose next
     that $\mathcal{X}\subseteq \mathcal{P}$. Let $0\le m=|\mathcal{P}\cap \mathcal{A}|<3n$. Then, if the number of $T_j^l$ players in $\mathcal{P}$ is $h >0$, then $v(\mathcal{P})\le 6n+h +2\cdot h$. We can suppose that $h \le 6n+m$, because if there are more $T_j^l$ players, then at most $6n+m$ of their $s_j^l$ vertices can be covered, hence the remaining players bring strictly more $x$ value than what they can increase the maximum cycle packing size with. However,
\[\begin{array}{lcl}
  x(\mathcal{P}) &\ge &(18n-\frac{2n-1}{3n})+(h +\frac{h}{9n})+m\cdot (3-\frac{n+1}{9n^2})\\ &> &18n+h +3m -1.\end{array}\] Hence, in order for $\mathcal{P}$ to block, it must hold that $2h >12n+3m-1$, so $h >6n + 1.5m-0.5$, contradicting $h \le 6n+m$, if $m\ge 1$. In the case, when $m=0$, we get that $6n+3h > 18n -\frac{2n-1}{3n}+h + \frac{h}{9n}$, so $2h > 12n-\frac{2n-1}{3n}+\frac{h}{9n}>12n-1$. From this and $h \le 6n + 0$, we get that $h$ must be $6n$. However, then $12n>12n-\frac{2n-1}{3n}+\frac{6n}{9n}>12n$, which is a contradiction again.

Therefore, suppose that $\mathcal{A}\subseteq \mathcal{P}$, but $\mathcal{X}$ is not included in $\mathcal{P}$.  Let $0\le m=|\mathcal{P}\cap \mathcal{X}|<6n$. Now, if the number of $T_j^l$ players in $\mathcal{P}$ is $h$, then $v(\mathcal{P})\le 3n+h +2\cdot h$, similarly as before. Again, we can suppose that $h \le 3n+m$ for similar reasons. Furthermore, \[\begin{array}{lcl} x(\mathcal{P})&\ge &9n-\frac{n+1}{3n} +h +\frac{h}{9n}+m\cdot (3-\frac{2n-1}{9n^2})\\ &> &9n +h + 3m -1.\end{array}\] If $m\ge 1$, then this implies that $2h >6n+3m-1$, so $h > 3n+1.5m -0.5$, a contradiction. We conclude that $m=0$. Therefore, $h \le 3n$ and $2h > 6n -\frac{n+1}{3n}+\frac{h}{9n}>6n-1$, so $h =3n$.

 To sum up, we showed that $\mathcal{P}$ must contain $\mathcal{A}$, must be disjoint from $\mathcal{X}$ and there must be exactly $3n$ $T_j^l$ players inside $\mathcal{P}$.

We claim that for each $j\in [3n]$, if $T_j^l\in \mathcal{P}$ for some $l\in [3n]$, then $T_j^l$ is inside $\mathcal{P}$ for all $l\in [3]$: if not then there must be at least one $t_j^l$ vertex that cannot be covered, hence $v(\mathcal{P})\le 12n-1$, but $x(\mathcal{P})=9n-\frac{n+1}{3n}+3n+\frac{3n}{9n}>12n-1$, a contradiction. Therefore, for each set $S_j$, if $T_j^l\in \mathcal{P}$ for some $l\in [3]$, then $T_j^l\in \mathcal{P}$ for all $l\in [3]$, so there are exactly $n$ sets $S_j$, such that $T_j^l\in \mathcal{P}$.

Finally, it remains to show that the $n$ sets corresponding to those
value of $j$ such that $T_j^l$ is in $\mathcal{P}$ for $l\in [3]$ must be the indices of an exact $3$-cover.
Suppose that there is an element $i$ that cannot be covered by them. Then, $a_i$ cannot be covered by a cycle packing by $\mathcal{P}$, so $v(\mathcal{P})\le 12n-1$, which leads to the same contradiction. \qed
\end{proof}

\noindent
Before we prove Theorem~\ref{t-ppg}, we need some definitions and a lemma. Let $G=(V,A)$ be a directed graph with a partition $(V_1,\ldots,V_n)$ of $V$ for some $n\geq 1$. Recall that for a maximum cycle packing ${\mathcal C}$ of~$G$, we let $s_p({\mathcal C})$ denote the number of arcs~$(u,v)$ with $v\in V_p$ that belong to some directed cycle of ${\mathcal C}$. We say that ${\mathcal C}$ {\it satisfies} a set of intervals $\{I_1,\ldots,I_n\}$ if $s_p({\mathcal C})\in I_p$ for every $p\in \{1,\ldots,n\}$.

\begin{lemma}\label{l-interval2}
For instances $(G,{\mathcal V},{\mathcal I})$, where $G=(V,A)$ is a directed graph, ${\mathcal V}=(V_1,\ldots,V_n)$ is a partition of $V$ with fixed width~$c$, and ${\mathcal I}=\{I_1, \dots, I_n\}$ is a set of intervals, the problem of finding a maximum cycle packing of $G$ satisfying ${\mathcal I}$ is polynomially solvable~if $c=1$, and \NP-complete if $c\geq 2$ even if $I_p=[1,\infty]$ for every $p\in \{1,\ldots,n\}$ or $I_p=[1,1]$ for every $p\in \{1,\ldots,n\}$.
\end{lemma}

\begin{proof}
First suppose $c=1$.
Let $v_p$ be the unique vertex in~$V_p$.  We can assume that each $I_p$ contains either $0$ or $1$, else no cycle packing satisfying ${\mathcal I}$ exists.
If $1 \notin I_p$, we can delete $v_p$ and redefine $G$ as the graph that remains.
If this decreases the size of the maximum cycle packing, then we conclude that no desired maximum cycle packing  exists.  Let $U$ be the set of vertices~$v_p$ for which $0 \notin I_p$.

The problem reduces to finding a maximum cycle packing such that each vertex in $U$ is covered.
For this, we transform $G=(V,A)$ into a bipartite graph $H$ with partition classes $V$ and $V'$, where $V'$ is a copy of $V$. For each $v\in V\setminus U$ and its copy $v'\in V'\setminus U'$, we add the edge~$uu'$ with weight~$0$ (we do not add these for the vertices of $U$). For each $(u,v)\in A$, we add the edge $uv$ with weight~$1$. It remains to find in polynomial time a maximum weight perfect matching in $H$, if there is any and check whether its weight is the same as the size of a maximum cycle packing in the original directed graph. If there is a perfect matching with that weight, then in the maximum cycle packing it corresponds to, each vertex in $U$ must be covered with a cycle. In the other direction, if there is a maximum cycle packing covering each vertex in $U$, then the perfect matching it corresponds to has the desired weight and we need no nonexistent $uu'$ edge for any $u\in U$ indeed.

Now suppose $c \geq 2$. Containment in \NP\ is trivial, as we can check in polynomial time if an arc set consists of vertex disjoint cycles or not, and for each player we can compute the number of incoming arcs.

To prove completeness, as in the proof of Theorem~\ref{t-core}, we reduce from the \NP-complete problem \textsc{Exact $3$-Cover}.
Given an instance $I$ of {\sc Exact $3$-Cover},
we construct an instance $I'$ of our problem as follows (see also Figure~\ref{fig:enter-label}):
\begin{itemize}
    \item [--] For each element $i\in [3n]$, we create a vertex $a_i$,
    \item [--] for each set $S_j\in \mathcal{S}$, we create vertices $s_j^1,s_j^2,s_j^3,t_j^1,t_j^2,t_j^3$,
    \item [--] we create $2n$ source vertices $x_1,\dots,x_{2n}$ and $2n$ sink vertices $y_1,\dots,y_{2n}$.
\end{itemize}

\noindent
Define the arcs as follows:
\begin{itemize}
\item [--] for each $k\in [2n]$, an arc $(y_k,x_k)$,
\item [--] for each $k\in [2n],j\in [3n]$, the arcs $(x_k,t_j^1)$ and $(t_j^3,y_k)$,

\item [--] for each $j\in [3n]$, the arcs $(t_j^1,t_j^2)$ and $(t_j^2,t_j^3)$, and
\item [--] for each set $S_j=\{j_1,j_2,j_3\}$, $j_1<j_2<j_3$, the arcs $(s_j^1,a_{j_1})$, $(a_{j_1},s_j^1)$, $(s_j^2,a_{j_2})$, $(a_{j_2},s_j^2)$, $(s_j^3,a_{j_3})$ and $(a_{j_3},s_j^3)$.
\end{itemize}

\noindent
This gives a directed graph $G=(V,A)$. We partition $V$ into players as follows:
\begin{itemize}
\item [--] for each $i\in [3n]$, we have a player $A_i=\{ a_i\}$,
\item [--] for each $k\in [2n]$, we have a player $X_k=\{ x_k\}, Y_k=\{ y_k\}$, and
\item [--] for each $j\in [3n],l\in [3]$, we have a player $T_j^l=\{ s_j^l,t_j^l\}$ .
\end{itemize}

\noindent
The maximum cycle packing of $G$ has size $16n$. This is because the $x_i,y_i$ vertices allow $2n$ cycles of length~$5$ through $t_j^1,t_j^2,t_j^3$ triples covering $10n $ vertices. The rest of the $t_j^l$ vertices cannot be covered. Also, for the other $s_j^l$ and $a_i$ vertices, they span a directed bipartite graph, so at most $3n+3n$ vertices can be covered, as we have only $3n$ $a_i$ vertices. And $6n$ can be covered indeed, as we can just choose an arbitrary $s_j^l$ neighbour for each $a_i$ and pair them with a 2-cycle.

The interval for each player is $[1,\infty]$. Since the size of the maximum cycle packing of $G$ is $16n$, which is the same as the number of players, if there is a solution that satisfies these intervals, then it also must satisfy the intervals $[1,1]$ for each player. Hence the last two statements of the Lemma are equivalent in this instance.

\begin{figure}[t]
    \centering
    \includegraphics[width=180pt]{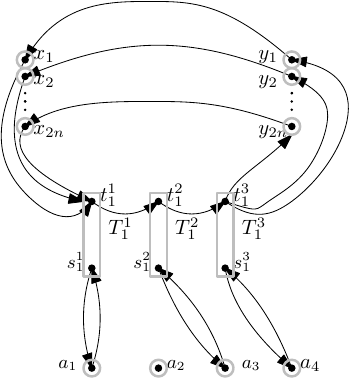}
    \caption{Illustration for Lemma \ref{l-interval2}, showing the construction for a set $S_1=\{1,2,4\}$. Grey bold lines mark the players.}
    \label{fig:enter-label}
\end{figure}

As a maximum cycle packing in $G$ has size~$16n$, which is equal to the sum of the lower bounds, $G$ has a cycle packing satisfying every interval if and only if $G$ has a maximum cycle packing satisfying every interval. We claim $I$ admits an exact $3$-cover if and only if  $G$ admits a
cycle packing satisfying every interval.

\medskip
\noindent
``$\Rightarrow$''  First suppose $I$ has an exact $3$-cover $S_{l_1},\dots,S_{l_n}$. We create a cycle packing $\mathcal{C}$ of $G$. For each $j\in \{ l_1,\dots,l_n\}$, we add the cycles  $\{ (s_j^1,a_{j_1}),(a_{j_1},s_j^1)\} ,\{ (s_j^2,a_{j_2}),(a_{j_2},s_j^2)\} ,\{ (s_j^3,a_{j_3}),(a_{j_3},s_j^3)\} $. For $j\notin \{ l_1,\dots,l_n\}$ we add the arcs $(t_j^1,t_j^2),(t_j^2,t_j^3)$. Finally, for each $i\in [2n]$, we add the arcs $(y_i,x_i),(x_i,t_{j_i}^1),(t_{j_i}^3,y_i)$, where $j_i$ is the $i$-th smallest index among the indices $[3n]\setminus \{ l_1,\dots,l_n\}$.

    Clearly, ${\mathcal C}$ is a cycle packing. Each $A_i$  has an incoming arc, as $S_{l_1},\dots,S_{l_n}$ was an exact $3$-cover. As there are exactly $2n$ sets not in the set cover, all of the corresponding players $T_j^l$  have one incoming arc in a cycle of the form $\{ (y_i,x_i),(x_i,t_{j_i}^1),(t_{j_i}^1,t_{j_i}^2),((t_{j_i}^3,y_i)\}$, and so did each $X_i$
and each $Y_i$.
Hence, all lower bounds are satisfied.

\medskip
\noindent
``$\Leftarrow$''  Now suppose $G$ admits a
cycle packing satisfying every interval. Then, as $X_i$ has an incoming arc for  all $i\in [2n]$, all $(x_i,y_i)$ arcs are included in the cycle packing. This means that there are $2n$ such $j\in [3n]$, such that the arcs $(t_j^1,t_j^2),(t_j^2,t_j^3)$ are included in a cycle $\{ (y_i,x_i),(x_i,t_j^1),(t_j^1,t_j^2),$ $(t_j^2,t_j^3),(t_j^3,y_i)\} $ of $\mathcal{C}$.

    From the above, we have that there are $n$ indices $j$ from $[3n]$, such that none of the players $T_j^1,T_j^2,T_j^3$ have incoming arcs of this form. Hence, all these players can only have incoming arcs from a player~$A_i$. As each such $T_j^i$  must have one incoming arc, it follows that for all these $j$, the cycles $\{ (a_{j_1},s_j^1),(s_j^1,a_{j_1})\} $, $\{ (a_{j_2},s_j^2),(s_j^2,a_{j_2})\}$, $\{ (a_{j_3},s_j^3),(s_j^3,a_{j_3})\} $ are included in $\mathcal{C}$, so they are vertex disjoint. Hence, the corresponding sets must form an exact $3$-cover. \qed
\end{proof}

\medskip
\noindent
{\bf Theorem~\ref{t-ppg} (restated).}
 \emph{For partitioned permutation games even of width~$2$, the problem of finding an optimal solution that is weakly or strongly close to a given target allocation~$x$ is \NP-hard.}

\begin{proof}
Recall that $x_p$ denotes the target for the number of arcs~$(u,v)$ with $v\in V_p$ that belong to some directed cycle of ${\mathcal C}$.  Letting each $I_p = [x_p, x_p]$ and applying Lemma~\ref{l-interval2}, we see that finding a cycle packing where each $s_p({\mathcal C})$ is equal $x_p$ (so differs by at most $0$) is \NP-complete.  Thus it is \NP-hard
to find the maximum cycle packing that minimizes $d_1({\mathcal C})=\max_{p\in N} \{|x_p-s_p({\mathcal C})|\}$; that is, to find a solution that is weakly close to a given target and similarly it is also hard to find a strongly close solution. \qed
\end{proof}

\noindent
In the remainder of our paper, the following problem plays an important role:

\problemdef{$q$-\textsc{Exact Perfect Matching}}{An undirected bipartite graph $B=(U,W;E)$, where each edge is coloured with one of $\{1,\dots, q\}$, and numbers $k_1,\ldots,k_q$.}{Is there a perfect matching in $B$ consisting of $k_q$ edges of each colour $q$?}

\noindent
For $q=2$, this problem is also known as {\sc Exact Perfect Matching}, which, as mentioned, was introduced by Papadimitriou and Yannakakis~\cite{PY82} and whose complexity status is open for more than 40 years. In the remainder of this section, we will give both a reduction to this problem and a reduction from this problem. We start with the former in the proof of our next result (Theorem~\ref{t-general}), from which
Theorem~\ref{t-interval}  immediately follows.

Let $x$ be an allocation for a partitioned permutation game $(N,v)$ on a graph $G=(V,A)$. For a maximum cycle packing ${\mathcal C}$,
 $d'({\mathcal C})=(|x_{p_1}-s_{p_1}({\mathcal C})|, \dots, |x_{p_n}-s_{p_n}({\mathcal C})|)$ is the {\it unordered deviation vector} of ${\mathcal C}$.

\begin{theorem}\label{t-general}
For a partitioned permutation game $(N,v)$ on a directed graph $G=(V,A)$ and a target allocation~$x$, it is possible to generate the set of unordered 
deviation vectors
in $|A|^{O(n)}$ time by a randomized algorithm.
\end{theorem}

\begin{proof}
Let $(N,v)$ be a partitioned permutation game with $n$ players, defined on a directed graph $G=(V,A)$ with vertex partition $V_1, \ldots, V_n$.
As mentioned, we reduce to $q$-\textsc{Exact Perfect Matching} for an appropriate value of $q$.
From $(N,v)$ and a vector $d'=(d_1',\ldots,d_n')$ with $d_p'\geq 0$ for every $p\in N$, we define an undirected bipartite graph $B=(U,W;E)$ with coloured edges: for each vertex $v\in V$, there is a vertex $v^{in}\in U$ and a vertex $v^{out}\in W$ and an edge $v^{in}v^{out}\in E$ that has colour $n+1$; for each arc $(u,v)\in A$, there is an edge $u^{out}v^{in}$, which will be coloured $p$ if $v\in V_p$.  Let $k_{n+1}=|V|-v(N)$ and, for $p\in\{1,\dots,n\}$, let $k_p=d_p'$.

We observe that $G$ has a maximum cycle packing ${\mathcal C}$ with $s_p({\mathcal C})=d_p'$  if and only if $B$ has a perfect matching with $k_p$ edges of each colour~$p\in \{1,\ldots,n+1\}$.

As each $k_i$ can only have a value between $0$ and $|E|=|A|+n=O(|A|)$, the above reduction implies that the set of unordered deviation vectors has size $|A|^{O(n)}$  for any allocation $x$ for $(N,v)$. We can find each of these vectors in $|A|^{O(n)}$ time by a randomized algorithm, as $q$-\textsc{Exact Perfect Matching} is solvable in $|E|^{O(q)}$ time with $q$ colours with a randomized algorithm \cite{GKMST2012}. \qed
\end{proof}

\noindent

We cannot hope to derandomize the algorithm from Theorem~\ref{t-general} without first solving {$2$-\textsc{Exact Perfect Matching}} problem in polynomial time, as we now show. 

\begin{theorem}\label{thm:part-matching-harder-than}
   Every instance of $2$-\textsc{Exact Perfect Matching} can be reduced in polynomial time to checking whether a partitioned permutation game $(N,v)$ with only 2 players has a solution with no deviation from a target allocation $x$. 
\end{theorem}

\begin{proof}
Take an instance $I=(B,k_1,k_2)$ of {$2$-\textsc{Exact Perfect Matching}}.
Let $B=(U,W;E)$ be the bipartite graph in $I$ with $|U|=|W|=n$ for some integer~$n$ whose edges are coloured either red (colour~$1$) or blue (colour~$2$). We may suppose that $n=k_1+k_2$, otherwise $I$ is clearly a no-instance. We construct a digraph $G$ from $B$ by replacing every edge $e=uw$ with a directed $3$-cycle on arcs $(u,v_e)$, $(v_e,w)$, $(w,u)$, where $w_e$ is a new vertex that has only $u$ and $v$ as its neighbours in $G$. 

Let the first player's set $V_1$ consist of all vertices $v_e$, for which $e$ is a red edge in $E$ and let the other player's vertex set be $V_2=V\setminus V_1$. This defines a partitioned permutation game $(N,v)$. Let the target allocation $x$ be given by $x=(k_1,3n-k_1)$.

We claim that $(B,k_1,k_2)$ is a yes-instance of {\sc $2$-Exact Perfect Matching} if and only if $G$ has a cycle packing ${\mathcal C}$ with $s_1({\mathcal C})=k_1$ and $s_2({\mathcal C})=3n-k_1$. Indeed, from a perfect matching $M$ of $B$ containing exactly $k_1$ red edges, we can subsitute each edge $uw\in M$ by the corresponding $3$-cycle $(u,v_e,w)$, such that exactly $k_1$ vertices from $V_1$ and $3n-k_1$ from $V_2$ are covered. It is also clear that this is a maximum size cycle packing, as at most $n$ of the $v_e$ vertices can be covered in any cycle packing. In the other direction, suppose we have a maximum size cycle packing with $s_1({\mathcal C})=k_1$ and $s_2({\mathcal C})=3n-k_1$. Then, it covers $k_1$ of the vertices $v_e$ corresponding to red edges and $n-k_1=k_2$ of the vertices $v_e$ corresponding to the blue edges. Also, it must be that for each such covered vertex $v_e$ for $e=uw$, we have the arcs $(u,v_e)$ and $(v_e,w)$ in ${\mathcal C}$. Therefore, these edges induce a perfect matching in $B$ with exactly $k_1$ red and $k_2$ blue edges.  
\end{proof}

\section{ILP Formulation}\label{a-ilp}

In this section, we show how to find an optimal solution of a partitioned permutation game that is strongly close to a given target allocation~$x$ by solving a sequence of Integer Linear Programs (ILPs).
Let $(N,v)$ be a partitioned permutation game defined on a directed graph $G=(V,A)$ with vertex partition $V_1, \ldots, V_n$ for some $n\geq 1$.
Recall that for a maximum cycle packing ${\mathcal C}$ of $G$, we let $s_p({\mathcal C})$ denote the number of arcs~$(u,v)$ with $v\in V_p$ that belong to some directed cycle of ${\mathcal C}$. Recall also that $|x_p-s_p({\mathcal C})|$ is the deviation of country~$p\in N$ from its target $x_p$ if ${\mathcal C}$ is chosen as optimal solution.
Moreover, in the vector $d({\mathcal C})= (|x_{p_1}-s_{p_1}({\mathcal C})|, \dots, |x_{p_n}-s_{p_n}({\mathcal C})|)$, the deviations $|x_p-s_p({\mathcal C})|$ are ordered non-increasingly. Finally, we recall that
${\mathcal C}$ is strongly close to~$x$ if $d({\mathcal C})$ is lexicographically minimal over all optimal solutions for $(N,v)$.

In the kidney exchange literature the following ILP is called the \emph{edge-formulation} (see e.g. \cite{ABS07}).
For each ${ij}\in A$ in the graph $G$, let $e_{{ij}}\in \{0,1\}$ be a (binary) edge-variable; this yields a vector~$e$ with entries $e_{{ij}}$.

\begin{equation*}\tag{edge-formulation}\label{edge-formulation}
\begin{array}{rl}
\mu^*:= \displaystyle \max_e \sum_{ij\in A}e_{ij} & \text{s.t.}\\[-5pt]
\end{array}
\end{equation*}
\begin{alignat}{5}
\displaystyle \sum_{j: ji\in A} e_{ji} & = && \displaystyle \sum_{j: ij\in A} e_{ij} & \quad \forall i\in V  \label{const-1}\\
\displaystyle \sum_{j: ji\in A} e_{ji} & \leq && \quad 1 & \quad \forall i\in V  \label{const-2}
\end{alignat}

\noindent
{Constraint~\eqref{const-1}} represents the well-known Kirchoff law{, whereas} {\eqref{const-2}} ensures that every node is covered by at most one cycle. The objective function provides a maximum cycle packing{, which has} size $\mu^*$.

This ILP has $|A|$ binary variables and $2|V|$ constraints. In the following we are going to sequentially find largest country deviations $d_t^*$ ($t \geq 1$) and the corresponding minimal number $n_t^*$ of countries receiving that deviation.
We achieve this by solving 
an 
ILP of similar size for each $d_t^*$ and $n_t^*$, so two ILPs per iteration $t$. By similar size, we mean that in each iteration we are going to add $|N|$ binary variables
and a single additional constraint,
while $|N| \leq |V|$ holds by definition
and typically $|N|$ is much smaller than $|A|$. Meanwhile, since at every iteration we are going to fix the deviation of at least one additional country (we will not necessarily know \emph{which} country, we are only going to keep track of the number of countries with fixed deviation), the number of iterations are at most $|N|$ (as $t \leq |N|$).
Hence, we will solve no more than $2|N|$ ILPs, among which the largest has $O(|A|+|N|^2)$ binary variables and $O(|V|+|N|)=O(|V|)$ constraints.

Once we have $\mu^*$ we solve the following ILP to find $d_1^*$:
\begin{equation*}\tag{$\text{ILP}_{d_1}$}\label{ILPd1}
\begin{array}{rl}
d_1^* := \displaystyle \min_{e,d_1} d_1 & \text{s.t.}\\
\end{array}
\end{equation*}
\begin{alignat}{5}
\displaystyle \sum_{j: ji\in A} e_{ji} & = && \displaystyle \sum_{j: ij\in A} e_{ij} & \quad \forall i\in V  \tag{1} \label{d1-const-1}\\
\displaystyle \sum_{j: ji\in A} e_{ji} & \leq && \quad 1 & \quad \forall i\in V  \tag{2} \\
\displaystyle \sum_{ij\in A} e_{ij} & = && \quad \mu^* \label{const-3}& \\
\displaystyle \sum_{j\in V_p} {\sum_{i: ij \in A}} e_{ij}-x_p & \leq && \quad d_1 & \quad \forall p\in N \label{const-4}\\
\displaystyle x_p-\sum_{j\in V_p} {\sum_{i: ij \in A}} e_{ij} & \leq && \quad d_1 & \quad \forall p\in N \label{const-5}
\end{alignat}

\medskip
\noindent
{Constraints~\eqref{d1-const-1}-\eqref{const-3}} guarantee that all solutions are, in fact, maximum cycle packings.
These constraints will be part of the 
formulation{s} 
throughout the entire ILP-series. {Constraints~\eqref{const-4} and \eqref{const-5}}, together with the objective function, guarantee that we minimize the largest country deviation. Note that for each country~$p$, exactly one of $\sum_{j\in V_p}{\sum_{i: ij \in A}} e_{ij}-x_p$ and $x_p - \sum_{j\in V_p} {\sum_{i: ij \in A}}e_{ij}$ is positive and exactly one 
is negative, unless both 
are zero. However, as soon as we reach $d_t^* \leq 1/2$ we have found a
strongly close maximum cycle packing. Hence,
in the remainder, we assume $d_t^* > 1/2$ for each $t$ such that the series continues with $t+1$.

\eqref{ILPd1} has one additional continuous variable ($d_1$) and $2|N|+1$ additional constraints. For every country $p \in N$, we have that $|\sum_{j\in V_p}{\sum_{i: ij \in A}} e_{ij}-x_p| \leq d_1^*$. However, there exists
a smallest subset $N_1 \subseteq N$ (which may not necessarily be unique)
such that
\begin{equation*}
\begin{array}{ll}
\displaystyle \left|\sum_{j\in V_p}{\sum_{i: ij \in A}} e_{ij}-x_p\right| = d_1^* & \forall p\in N_1 \\
&\\
\displaystyle \left|\sum_{j\in V_p}{\sum_{i: ij \in A}} e_{ij}-x_p\right| \leq d_2^* < d_1^* & \forall p\in N\setminus N_1
\end{array}
\end{equation*}

\medskip
\noindent
In a solution of \eqref{ILPd1}, let $n_1$
be the number of countries with deviation $d_1^*$.
We need to determine if there is another solution of \eqref{ILPd1}
with fewer than $n_1$, possibly $n_1^*$, countries having 
deviation~$d_1^*$.
For this purpose we must be able to distinguish between countries unable to have less than $d_1^*$ deviation and countries for which the deviation is at most $d_2^*$, the latter value unknown at this stage.
In order to make this distinction,
we determine a lower bound on $d_1^* - d_2^*$ by examining the target allocation $x$.
We will then set $\varepsilon$ in the next ILP to be strictly smaller than this lower bound.

The number of vertices for a country covered by any cycle packing is an integer. Hence, the number of possible country deviations is at most $2|N|$ and depends only on $x$.
The fractional part of a country deviation is either $\frc(x_p)$
or $1-\frc(x_p)$. Therefore, to find the minimal positive difference in between the deviations of any two countries $p$ and $q$, we have to compare the values $\frc(x_p)$ and $1-\frc(x_p)$,  with $\frc(x_q)$ or $1-\frc(x_q)$ and take the minimum of those four possible differences.
{We l}et $\varepsilon$ be a small positive constant that is {weakly} smaller than the minimum possible positive difference between any two countries {except in the unique case when $\frc(x_p)=0.5$ for every $p \in N$. This is because in that case, $\frc(x_p)=1-\frc(x_p)=\frc(x_q)=1-\frc(x_q)=0.5$, meaning that all four possible differences are zero. However, should this happen, then the minimum possible positive difference of any two countries' deviation can be trivially found as~$1$, and then we can simply take 
any~$\varepsilon$ with $0 < \varepsilon \leq 1$.}

We will distinguish between countries having minimal deviation of $d_1^*$ and others through additional binary variables. Since later in the ILP series we will need to distinguish between countries fixed at different deviation levels, let us introduce 
binary variables $z^t_p$, 
where $z^t_p=1$ indicates that $p \in N_t$.
\vspace{2cm}
\begin{equation*}\tag{$\text{ILP}_{N_1}$}\label{ILPN1}
\begin{array}{rl}
\displaystyle \min_{z^1,e}  \displaystyle \sum_{p \in N}z^1_p & \text{s.t.} \\
\end{array}
\end{equation*}
\begin{alignat}{5}
\displaystyle \sum_{j: ji\in A} e_{ji} & = && \displaystyle \sum_{j: ij\in A} e_{ij} & \quad \forall i\in V  \tag{1} \\
\displaystyle \sum_{j: ji\in A} e_{ji} & \leq && \quad 1 & \quad \forall i\in V  \tag{2} \\
\displaystyle \sum_{ij\in A} e_{ij} & = && \quad \mu^* & \tag{3} \\
\displaystyle \sum_{j\in V_p} {\sum_{i: ij \in A}}  e_{ij}-x_p & \leq && \quad d_1^* - \varepsilon(1-z^1_p) & \quad \forall p\in N \label{const-6}\\
\displaystyle x_p-\sum_{j\in V_p} {\sum_{i: ij \in A}} e_{ij} & \leq && \quad  d_1^* - \varepsilon(1-z^1_p) & \quad \forall p\in N \label{const-7}
\end{alignat}

\noindent
As discussed, for each country $p$, the left hand side of either {\eqref{const-6} or \eqref{const-7}} is negative (i.e., would be satisfied even with $z^1_p=0$).  For those countries whose deviation cannot be lower than $d_1^*$, however,  the (positive) left hand side of either {\eqref{const-6} or \eqref{const-7}} will require $z^1_p=1$.
Thus, given an optimal solution~$z^{1*}$ of \eqref{ILPN1}, let $n_1^*:=\sum_{p \in N} z_p^{1*}$ be the minimal number of countries receiving the largest country deviations. It is guaranteed that the non-increasingly ordered country deviations at a strongly close maximum cycle packing starts with exactly $n_1^*$ many $d_1^*$ values, followed by some $d_2^* < d_1^*$. \eqref{ILPN1} has $|A|+|N|$ binary variables and $2|V|+2|N|+1$ constraints.
Now, to find $d_2^*$, we solve the following ILP:
\begin{equation*}\tag{$\text{ILP}_{d_2}$}\label{ILPd2}
\begin{array}{rl}
\displaystyle \min_{d_2,e,z^1} d_2 & \text{s.t.}\\
\end{array}
\end{equation*}
\begin{alignat}{5}
\displaystyle \sum_{j: ji\in A} e_{ji} & = && \displaystyle \sum_{j: ij\in A} e_{ij} & \quad \forall i\in V  \tag{1} \\
\displaystyle \sum_{j: ji\in A} e_{ji} & \leq && \quad 1 & \quad \forall i\in V  \tag{2} \\
\displaystyle \sum_{ij\in A} e_{ij} & = && \quad \mu^* & \tag{3} \\
\displaystyle \sum_{j\in V_p} {\sum_{i: ij \in A}} e_{ij}-x_p & \leq && \quad d_1^* & \quad \forall p\in N \\
\displaystyle x_p-\sum_{j\in V_p} {\sum_{i: ij \in A}} e_{ij} & \leq && \quad d_1^* & \quad \forall p\in N \\
\displaystyle \sum_{j\in V_p} {\sum_{i: ij \in A}} e_{ij}-x_p & \leq && \quad d_2 + z^1_pd_1^* & \quad \forall p\in N \\
\displaystyle x_p-\sum_{j\in V_p} {\sum_{i: ij \in A}} e_{ij} & \leq && \quad d_2 + z^1_pd_1^* & \quad \forall p\in N\\
\displaystyle \sum_{p \in N}z^1_p & = && n_1^* &
\end{alignat}

\noindent
\eqref{ILPd2} has $|A|+|N|$ binary variables and one continuous variable ($d_2$) with $2|V|+4|N|+2$ constraints, and guarantees that we find the minimal second-largest country deviation $d_2^*$ while exactly $n_1^*$ countries deviation is kept at $d_1^*$. Finding $n_2^*$ follows a similar approach, where $L$ is a large constant satisfying $L \geq 2(d_1^* - d_2^*)$:

\begin{equation*}\tag{$\text{ILP}_{N_2}$}\label{ILPN2}
\begin{array}{rl}
\displaystyle \min_{z^1,z^2,e} \displaystyle \sum_{p \in N}z^2_p & \text{s.t.} \\
\end{array}
\end{equation*}
\vspace{-\baselineskip}
\begin{alignat}{5}
\displaystyle \sum_{j: ji\in A} e_{ji} & = && \displaystyle \sum_{j: ij\in A} e_{ij} & \quad \forall i\in V  \tag{1} \\
\displaystyle \sum_{j: ji\in A} e_{ji} & \leq && \quad 1 & \quad \forall i\in V  \tag{2} \\
\displaystyle \sum_{ij\in A} e_{ij} & = && \quad \mu^* & \tag{3} \\
\displaystyle \sum_{j\in V_p} {\sum_{i: ij \in A}} e_{ij}-x_p & \leq && \quad d_2^* - \varepsilon(1-z^2_p) + z_p^1L & \quad \forall p\in N \\
\displaystyle x_p-\sum_{j\in V_p} {\sum_{i: ij \in A}} e_{ij} & \leq && \quad d_2^* - \varepsilon(1-z^2_p) + z_p^1L & \quad \forall p\in N \\
\displaystyle \sum_{j\in V_p} {\sum_{i: ij \in A}} e_{ij}-x_p & \leq && \quad d_1^* & \quad \forall p\in N \tag{8}\\
\displaystyle x_p-\sum_{j\in V_p} {\sum_{i: ij \in A}} e_{ij} & \leq && \quad d_1^* & \quad \forall p\in N  \tag{9} \\
\displaystyle \sum_{p \in N}z^1_p & = && \quad n_1^* & \tag{12} \\
z_p^1 + z_p^2 & \leq & 1 & \quad \forall p \in N
\end{alignat}

\medskip
\noindent
Subsequently we follow a similar approach for all $t \geq 3$, until either $|N|=n_1^*+n_2^*+\dots+n_t^*$ or we
terminate because $d_t^* \leq 1/2$.
Until reaching one of these conditions we iteratively solve the following two ILPs, introducing additional $|N|$ binary variables and an additional constraint to both. Let $L$ be a large constant satisfying $L \geq d_t^*$, e.g. $L = d_{t-1}^*$.

\begin{equation*}\tag{$\text{ILP}_{d_t}$}\label{ILPdt}
\begin{array}{rl}
\displaystyle \min_{d_t,e,(z^i)_{i=1}^{t-1}} d_t & \text{s.t.}\\
\end{array}
\end{equation*}
\begin{alignat}{5}
\displaystyle \sum_{j: ji\in A} e_{ji} & = && \displaystyle \sum_{j: ij\in A} e_{ij} & \quad \forall i\in V  \tag{1} \\
\displaystyle \sum_{j: ji\in A} e_{ji} & \leq && \quad 1 & \quad \forall i\in V  \tag{2} \\
\displaystyle \sum_{ij\in A} e_{ij} & = && \quad \mu^* & \tag{3} \\
\displaystyle \sum_{i=1}^{t-1}z^i_p & \leq && \quad 1 & \quad \forall p \in N \\
\displaystyle \sum_{j\in V_p} {\sum_{i: ij \in A}} e_{ij}-x_p & \leq && \quad d_t + \displaystyle \sum_{i=1}^{t-1}z^i_pd_i^* & \quad \forall p\in N \\
\displaystyle x_p-\sum_{j\in V_p} {\sum_{i: ij \in A}} e_{ij} & \leq && \quad d_t + \displaystyle \sum_{i=1}^{t-1}z^i_pd_i^* & \quad \forall p\in N\\
\displaystyle \sum_{j\in V_p} {\sum_{i: ij \in A}} e_{ij}-x_p & \leq && \quad \displaystyle \sum_{i=1}^{t-1}z^i_pd_i^* + {\left(1-\displaystyle \sum_{i=1}^{t-1}z^i_p\right)}L & \quad \forall p\in N \\
\displaystyle x_p-\sum_{j\in V_p} {\sum_{i: ij \in A}} e_{ij} & \leq && \quad \displaystyle \sum_{i=1}^{t-1}z^i_pd_i^* + {\left(1-\displaystyle \sum_{i=1}^{t-1}z^i_p\right)}L & \quad \forall p\in N \\
\displaystyle \sum_{p \in N}z^i_p & = && \quad n_i^* & \quad \forall {i \in \{1,\dots,t-1\}}
\end{alignat}

\medskip
\noindent
In the following 
{formulation $L'$} is a large constant satisfying 
${L'} \geq d_t^* - \varepsilon$.
\begin{equation*}\tag{$\text{ILP}_{N_t}$}\label{ILPNt}
\begin{array}{rl}
\displaystyle \min_{(z)_{i=1}^t,e} \displaystyle \sum_{p \in N}z^t_p &  \text{s.t.}\\
\end{array}
\end{equation*}
\begin{alignat}{5}
\displaystyle \sum_{j: ji\in A} e_{ji} & = && \displaystyle \sum_{j: ij\in A} e_{ij} & \quad \forall i\in V  \tag{1} \\
\displaystyle \sum_{j: ji\in A} e_{ji} & \leq && \quad 1 & \quad \forall i\in V  \tag{2} \\
\displaystyle \sum_{ij\in A} e_{ij} & = && \quad \mu^* & \tag{3} \\
\displaystyle \sum_{i=1}^{t-1}z^i_p & \leq && \quad 1 & \quad \forall p \in N \tag{16} \\
\displaystyle \sum_{j\in V_p} e_{ij}-x_p & \leq && \quad d_t^* - \varepsilon(1-z^t_p) + \displaystyle\sum_{i=1}^{t-1}z_p^id_i^* & \quad \forall p\in N \\
x_p - \displaystyle \sum_{j\in V_p} e_{ij} & \leq && \quad d_t^* - \varepsilon(1-z^t_p) + \displaystyle\sum_{i=1}^{t-1}z_p^id_i^* & \quad \forall p\in N \\
\displaystyle \sum_{j\in V_p} e_{ij}-x_p & \leq && \quad \displaystyle\sum_{i=1}^{t}z_p^id_i^* + {\left(1 - \displaystyle\sum_{i=1}^{t}z_p^i\right)}{L'} & \quad  \forall p \in N \\
\displaystyle x_p - \sum_{j\in V_p} e_{ij} & \leq && \quad \displaystyle\sum_{i=1}^{t}z_p^id_i^* + {\left(1 - \displaystyle\sum_{i=1}^{t}z_p^i\right)}{L'} & \quad \forall p \in N \\
\displaystyle \sum_{p \in N}z^i_p & = && \quad n_i^* & \quad \forall {i \in \{1,\dots,t-1\}} \tag{21}
\end{alignat}

\medskip
\noindent
From the above, we conclude that the following theorem holds, 
{(in which the bounds on the number of ILPs, variables and constraints are not necessarily tight)}.
\begin{theorem}
For a partitioned permutation {game} $(N,v)$ defined on a graph $G=(V,A)$, it is possible to find an optimal solution that is strongly close to a given target allocation~$x$
by solving a series of at most $2|N|$ ILPs, each having $O(|A|+|N|^2)$ binary variables and $O(|V|)$ constraints.
\end{theorem}

\noindent
Note that if we just want to find a weakly close optimal solution we can stop after solving the first ILP.

\section{Simulations}~\label{s-simul}
In this section we describe our simulations for $\ell=\infty$ in detail. Our goals are

\begin{enumerate}
\item to examine the benefits of strongly close optimal solutions over weakly close optimal solutions or arbitrarily chosen optimal solutions;
\item to examine the benefits of using credits;
\item to examine the exchange cycle length distribution; and
\item to compare all these results with those for the other extreme case where $\ell=2$~\cite{BBPY24}.
\end{enumerate}

\noindent
 For our fourth aim, we are especially interested in the increase in the total number of patients treated when we compare the $\ell=2$ and $\ell=\infty$ cases.

\subsection{The Set Up}\label{s-setup}

To allow a fair comparison we use the same set up as in~\cite{BBPY24}, and in addition, we also use the data from~\cite{Be21} that was used in~\cite{BBPY24}.
That is, for our simulations, we take the same 100 compatibility graphs $G_1,\ldots, G_{100}$, each with roughly\footnote{The graphs $G_i$ might not have exactly 2000 vertices: depending on the specific number of countries that we consider, we may need {to} do some rounding in order for the countries to have the right sizes.}
2000 vertices from~\cite{Be21,BBPY24}.
As real medical data is unavailable to us for privacy reasons, the data from~\cite{Be21} was obtained using the data generator from~\cite{PT21}. This data generator was used in many papers and is the most realistic synthetic data generator available; see also~\cite{DGGKMPT22}.

We use every $G_i$ ($i\in \{1,\ldots,100\}$) as the basis to perform simulations for $n$ countries, where we let $n=4,\ldots,10$. In order to do this, we first consider the case where all countries have equal size. That is, we partition the vertices $V_i$ of the graph $G_i$ into the same $n$ sets $V_{i,1},\ldots, V_{i,n}$ as in~\cite{BBPY24} which are of equal size ${2000}/{n}$ (subject to rounding), so $V_{i,p}$ is the set of patient-donor pairs of country~$p$. For round~1, we construct a compatibility graph~$G_i^1(n)$ as a subgraph of $G_i$ of size roughly~$500$. We add the remaining patient-donor pairs of $G_i$  as vertices by a uniform distribution between the remaining rounds $2,\ldots,24$.
So, starting with $G_i^1(n)$, we run an IKEP of 24 rounds in total. This gives us 24 compatibility graphs $G_i^1(n),\ldots, G_i^{24}(n)$. As in~\cite{BBPY24}, any patient-donor pair whose patient is not helped within four rounds will automatically be deleted from the pool.

We also consider a situation where countries have different sizes. We note that there are numerous ways one could allow for differently sized countries. However, we follow the approach of~\cite{BBPY24}, as this allows us to compare the results with those for $\ell=2$. That is, we consider the same 100 compatibility graphs {$G_1,\ldots,G_{100}$} as before; the only difference is that we now partition each $V_i$ into the same $n$ sets $V_{i,1},\ldots, V_{i,n}$ as in~\cite{BBPY24}, such that 

\begin{itemize}
\item approximately $n/3$ are small, that is, have size roughly  $1000/n$  (subject to rounding); 
\item approximately $n/3$ are medium, that is, have size roughly $2000/n$; and
\item approximately $n/3$ are large, that is, have size roughly  $3000/n$. 
\end{itemize}

\noindent
A {\it (24-round) simulation instance} consists of 

\begin{itemize}
\item [(i)] the data needed to generate a graph $G_i^1(n)$ and its successors $G_i^2(n),\ldots,G_i^{24}(n)$
\item [(ii)] an indication of whether the countries must be of the same size or can have different sizes (in the way explained above); 
\item [(iii)] a choice of solution concept for computing the initial allocation (see Section~\ref{s-inin}); and
\item [(iv)] a choice of type of optimal solution (see Section~\ref{s-opop}). 
\end{itemize}

\noindent
Our code for obtaining the simulation instances can be found in a GitHub repository~\cite{Ye23}, along with the compatibility graphs data and the seeds for the randomization.

We now discuss our choice for the initial allocations and optimal solutions.

\subsection{The Initial Allocations}\label{s-inin}

For the initial allocations~$y$ we use seven known solution concepts, which we all define below:

\begin{itemize}
\item the Shapley value,
\item the Banzhaf value, 
\item the nucleolus, 
\item the tau value,
\item the benefit value
\item the contribution value, and 
\item the Banzhaf* value (which, we recall, is a conceptually different solution concept).
\end{itemize}

\noindent
We define each of these solution concepts below. Each of them prescribes exactly one allocation for the partitioned permutation games  associated with the {compatibility} graphs. For computing these allocations we need the $v$-values of the partitioned permutation games. 

We can compute a single value $v(S)$  in polynomial time by Theorem~\ref{t-hard} through a transformation into a bipartite graph and 
 solving a maximum weight perfect matching problem (see Section~\ref{s-intro}). For this we use the package LEMON version 1.3.1 of~\cite{DJK11}.

Let $(N,v)$ be a cooperative game. The {\it Shapley value} $\phi(N,v)$~\cite{Sh53} is defined by setting for $p\in N$,
$$\phi_p(N,v) = \sum_{S \subseteq N\backslash \{p\}}
 \frac{|S|!(n-|S|-1)!}{n!}(v(S\cup \{p\})-v(S)).$$
To define the next solution concept, we first introduce {\it unnormalized Banzhaf value} $\psi(N,v)$~\cite{Ba64} defined by setting for $p\in N$,
$$\psi_p(N,v):=\sum_{S \subseteq N\backslash \{p\}} \frac{1}{2^{n-1}}(v(S\cup \{p\})-v(S)).$$
As $\psi_p$ may not be an allocation, the {\it (normalized) Banzhaf value} $\overline{\psi}(N,v)$ of a game $(N,v)$ was introduced and defined by setting for $p\in N$,
$$\overline{\psi}_p(N,v):=\frac{\psi_p(N,v)}{\sum_{q \in N}\psi_q(N,v)} \cdot v(N).$$
 Whenever we mention the Banzhaf value, we will mean $\overline{\psi}(N,v)$.
To compute the Shapley value and Banzhaf value in our simulations for even up to $n=10$ countries, we were still able to implement a brute force approach relying on the above definitions.

We now define the nucleolus.
The {\it excess} for an allocation~$x$ of a game $(N,v)$ and a nonempty coalition $S \subsetneq N$ is defined as $e(S,x) := x(S)-v(S)$.
Ordering the $2^n-2$ excesses in a non-decreasing sequence yields {\it excess vector} $e(x) \in  \R^{2^n-2}$.
The {\it nucleolus}
{$\eta$} 
of a game ($N,v$) is the unique allocation~\cite{Sc69} that lexicographically maximizes $e(x)$ over the set of allocations $x$ with $x_i\geq v(\{i\})$ (assuming this set is nonempty, as is the case for partitioned permutation games). 
To compute the nucleolus, we use the \emph{Lexicographical Descent} method of~\cite{BFN21}, which is the state-of-the-art method in nucleolus computation.

Next, we define the tau value, introduced by Tijs~\cite{Ti81}. First, let $b_p = v(N) - v(N \setminus p)$ be the \emph{utopia payoff} of $p \in N$, leading to a vector of utopia payoffs $b \in \mathbb{R}^N$. Now, for $p \in S \subseteq N$, we let $$R(S,p) = v(S) - \sum_{q \in S \setminus p} b_q$$ be the \emph{remainder} for player~$p$ in $S$, which is what would remain if  all players apart from~$p$ would leave a coalition $S$ with their utopia payoff. We now set for $p\in N$,
$$a_p = \max_{S \ni p} R(S,p),$$ leading to the vector $a \in \mathbb{R}^N$  of \emph{minimal rights}. A game $(N,v)$ is \emph{quasibalanced} if $a \leq b$ and $a(N) \leq v(N) \leq b(N)$. For a quasibalanced game, the tau value $\tau$ of a game $(N,v)$ is defined by setting for $p\in N$, $$\tau_p = \gamma a_p + (1-\gamma)b_p,$$ where $\gamma \in [0,1]$ is determined by $\tau(N) = v(N)$. Note that $\gamma$ is unique, unless $a=b$ (in which case we can simply take $\tau=b$).

The tau value is only defined for quasibalanced games. 
However, all partitioned permutation games have a nonempty core by Theorem~\ref{t-core}, and thus they are quasibalanced. 
As for the Shapley value and Banzhaf value, we were still able to compute the tau value by using a brute force approach based on its definition even for $n=10$ countries.

The {\it surplus} of a game $(N,v)$ is defined as $\mbox{surp}=v(N) - \sum_{p \in N} v(\{p\})$.  As long as $\sum_{p\in N}\alpha_p=1$, we can
allocate $v(\{p\})+\alpha_p\cdot\mbox{surp}$ to each player $p\in N$. In this way we obtain the {\it benefit value}~{\cite{KNPV20}} by setting for each $p\in N$,
$$\alpha_p=\frac{v(N) - v(N \setminus \{p\})-v(\{p\})}{\sum_{{q}\in N}(v(N) - v(N \setminus \{{q}\})-v(\{{q}\}))},$$
and we obtain the {\it contribution value}~{\cite{BBKP22}} by setting for each $p\in N$,
$$\alpha_p=\frac{v(N) - v(N \setminus \{p\})}{\sum_{{q}\in N}(v(N) - v(N \setminus \{{q}\}))}.$$
Both the benefit value and contribution value are easy to compute, as we only need to compute $2n-1$ $v$-values. We also note that that the benefit value and contribution value
do not exist if the denominator is zero. However, this did not happen in our simulations.

Finally, we consider a recent variant of the Banzhaf value, which was introduced in~\cite{BBPY24}. 
We first define for some round~$h$ in an IKEP the {\it credit-adjusted game} $(N^h,\overline{v}^h)$, in which the credits are directly incorporated, that is, $\overline{v}^h(S) = v^h(S) + \sum_{p \in S}c^h_p$ for every $S\subseteq N$. As explained in~\cite{BBPY24}, this alternative way of processing credits does not make any difference for any of the solution concepts that we consider except for one: the (normalized) Banzhaf value. We denote the adjusted Banzhaf value as the {\it Banzhaf* value}. In our simulations, we compute it by using brute force as we do for the Banzhaf value.

{The solution concepts defined above have several fairness properties. For example, they may prescribe an allocation $x_i$ for a game $(N,v)$ that is {\it individual rational}, that is, $x_p\geq v(\{p\})$ for every $p\in N$. This property holds for the nucleolus, benefit value and contribution value by definition. For the tau value, which is only defined on quasibalanced games, it follows from the observation that for every $p\in N$, $\tau_p \geq a_p \geq v(\{p\})$.
The Shapley value and unnormalized Banzhaf value are not individual rational in general. However, they are readily seen to be individual rational for {\it superadditive games}, i.e., games $(N,v)$ for which $v(S \cup T) \geq v(S) + v(T)$ for all pairs $S,T \subset N$ with $S \cap T = \emptyset$, which include the partitioned permutation games. 
In contrast, the normalized Banzhaf value is not individual rational for partitioned permutation games, as the following example shows.}

\medskip
\noindent
{{\it Example 3.}
Let $G=(\{a,b,c\},\{(a,b),(b,a),(b,c),(c,b)\}$ and $V_1=\{a\}$, $V_2=\{b\}$ and $V_3=\{c\}$.
Now, $\eta=(0,2,0)$ is the only core allocation, and we also note that $\phi_1,\psi_1,\overline{\psi}_1 > 0$. If
we extend $G$ by adding $V_4=\{d,e\}$ with arcs $(d,e)$, $(e,d)$, then
$\psi_4 = v(\{4\}) = 2 > \overline{\psi}_{4}$.
\qed}

\medskip
\noindent
{As discussed in Section~\ref{s-intro}, the nucleolus $\eta$ is not only individual rational but even belongs to the core (so $\eta(S)\geq v(S)$ for {\it every} $S\subseteq N$).}
{For a more detailed description on these solution concepts and their fairness properties, we refer to \cite{Peters2008,Young1985}. Here, we only highlight one other fairness property to illustrate that policy makers must make a choice on what fairness properties they find most important (in addition to any computational considerations).
A solution concept for a game $(N,v)$ is {\it coalitional monotone} if for every $i\in N$ the following holds: if the value of each of the coalitions player $i$ is a member of weakly increases, while the values of all the other coalitions remain the same, then the payoff of player $i$ must also weakly increase. 
The Shapley value and the Banzhaf value have this fairness property, whilst the nucleolus does not~\cite{Young1985}. However, whereas the nucleolus always belongs to the core of the game, Example 3 shows that 
the Shapley and Banzhaf values of permutation games may not be core allocations. In fact, Young~\cite{Young1985} even proved that no coalition monotone solution concept can guarantee to prescribe a core allocation.}

{Fairness properties, such as individual rationality, core membership and coalition monotonicity, are highly relevant for decision and policy making in international kidney exchange as well. As we have seen, such properties may be in conflict with each other. Hence, IKEPs should decide which fairness properties they prefer when selecting a certain solution concept. We also note that (fractional) target allocations are typically not possible to achieve exactly. This is, for instance, illustrated by Example~3, where for the original compatibility graph~$G$,
the nucleolus is $(0,2,0)$ while the only two optimal solutions are $(1,1,0)$ and $(0,1,1)$.
However, the credit system offers a way to compensate for any deviations from target allocations 
in the long term.} 

\subsection{The Optimal Solutions}\label{s-opop}

As optimal solutions we compute the following for each {compatibility} graph:

\begin{itemize}
\item [1.] an arbitrary optimal solution;
\item [2.] a weakly close optimal solution; and
\item [3.] a strongly close optimal solution.
\end{itemize}

\noindent
In this way we can measure the effect of using weakly close optimal solutions over arbitrarily chosen ones as well as the effect of using strongly close optimal solutions over weakly close ones.
We compute the weakly and strongly close optimal solutions by solving a sequence of ILPs, as described in Section~\ref{a-ilp}. 
We use Gurobi Optimizer version 10.0.02(linux64) to create environments, build and solve ILPs (via the Gurobi C++ interface on a cloud system).
As not every ILP might be solved within a reasonable time, we give each ILP a time limit of one hour. {As a minor comment, w}e recall that the values 
{$v(N)$} of the partitioned permutation games used in the ILPs have already been computed for the initial allocations, and we simply reuse {these to save some computation time}.

\subsection{Computational Environment and Scale}

In our large-scale experimental study, we run our simulations both without and with (``$+c$'')  the credit system, and as mentioned for the settings, where an arbitrary optimal solution (``\emph{arbitrary}''), weakly optimal solution (``$d_1$'') or  
strongly close optimal solution (``{\it lexmin}'') is chosen. This leads to the following five scenarios for each of the six selected solution concepts, the Shapley value, Banzhaf value, nucleolus, tau value, benefit value and contribution value: 
\begin{itemize}
\item \emph{arbitary},
\item \emph{d1},
\item \emph{d1+c}, 
\item \emph{lexmin},  
\item \emph{lexmin+c}; 
\end{itemize}
\noindent
plus two scenarios for the Banzhaf* value:
\begin{itemize}
\item \emph{d1+c}, 
\item \emph{lexmin+c},
\end{itemize}
as without credits the Banzhaf* value coincides with the Banzhaf value. 

Note that for the \emph{arbitrary} scenario, the use of credits is irrelevant. Hence, in total, we run the same set of simulations for $5\times 6 + 2 = 32$ different combinations of scenarios and solution concepts. We consider both the situation where all countries have the same size and a situation where they have different sizes. Moreover, we 
{vary the number of countries $n$ from $4$ to $10$ (so seven different numbers)} and {have} 100 compatibility graphs $G_i$. Hence, the scale of our experimental study is very large. Namely, the total number of 24-round simulation instances is 
$$32 \times 2 \times 7 \times 100 = 44800.$$

\noindent
All simulations were run on a dual socket server with AMD EPYC 7702 64-Core Processor with 2.00 GHz base speed and 256GB of RAM, where each simulation was given eight cores and 1G temporary disk space.
As we will compare our results for $\ell=\infty$ with the known results for $\ell=2$ from~\cite{BBPY24}, let us point out that
the simulations in~\cite{BBPY24} were run on a desktop PC with AMD Ryzen 9 5950X 3.4 GHz CPU and 128~GB of RAM, running on Windows 10 OS and C++ implementation in Visual Studio (the code is available in GitHub repository~\cite{Be21}).

\subsection{Evaluation Measures}

Let $y^*$ be the {\it total initial allocation} of a single simulation instance, that is,
$y^*$ is obtained by taking the sum of the 24 initial allocations of each of the 24 rounds of that instance. Let ${\mathcal C}^*$ be the union of the chosen maximum cycle packings in each of the 24 rounds. We use the {\it total relative deviation} defined as
$$\frac{\sum_{p \in N} |y_p^* - s_p({\mathcal C^*})|}{|{\mathcal C^*}|}.$$ For each choice of solution concept, choice of scenario and choice of number of countries,
 we run 100 simulation instances. We take the average of the 100 total relative deviations to obtain the \emph{average total relative deviation}.
Taking the {\it maximum relative deviation} $$\frac{\max_{p \in N} |y_p^* - s_p({\mathcal C^*})|}{|{\mathcal C^*}|}$$ gives us
the {\it average maximum relative deviation} as our second evaluation measure. As we shall see, both evaluation measures lead to the same conclusions.

\subsection{Simulation Results for ${\mathbf{\ell=\infty}}$ and Comparison with ${\mathbf{\ell=2}}$}

{In this section, we present our simulation results for $\ell=\infty$ and compare these results with the simulation results for $\ell=2$~\cite{BBPY24}.}
{We start with  Figure~\ref{fig:absolutedeviations}. In this figure}
we display our main results for {\it equal country sizes}, that is, when all countries are of the same size (see also Section~\ref{s-setup}), {and} we compare {the seven} different solution concepts under {the five} different scenarios for $\ell=\infty$ {with each other}. Figure~\ref{fig:absolutedeviations} also shows the effects of weakly and strongly close solutions and the credit system. As solution concepts have different complexities, we believe such a comparison might be helpful for policy makers in choosing a specific solution concept and scenario.
{Figure~\ref{fig:absolutedeviations}} shows that  
using an arbitrary maximum cycle packing in each round {indeed} makes the kidney exchange scheme significantly more unbalanced, with average total relative deviations over 4\% for all initial allocations~$y$.
{Figure~\ref{fig:absolutedeviations} also shows that t}he effect of both selecting a strongly close solution (to ensure being close to a target allocation) and using a credit function (for fairness, to keep deviations small) is significant.
{Finally, we observe that for all solution concepts and all scenarios, the deviations are slowly increasing when the number of countries is increasing.}
{However, the rate of increase is 
slower for strongly close optimal solutions than for weakly close optimal solutions, while the ``starting'' point of the line is lower when credits are employed.}

\begin{figure}
  \centering
  \hspace*{-1cm}
  \includegraphics[scale=0.42]{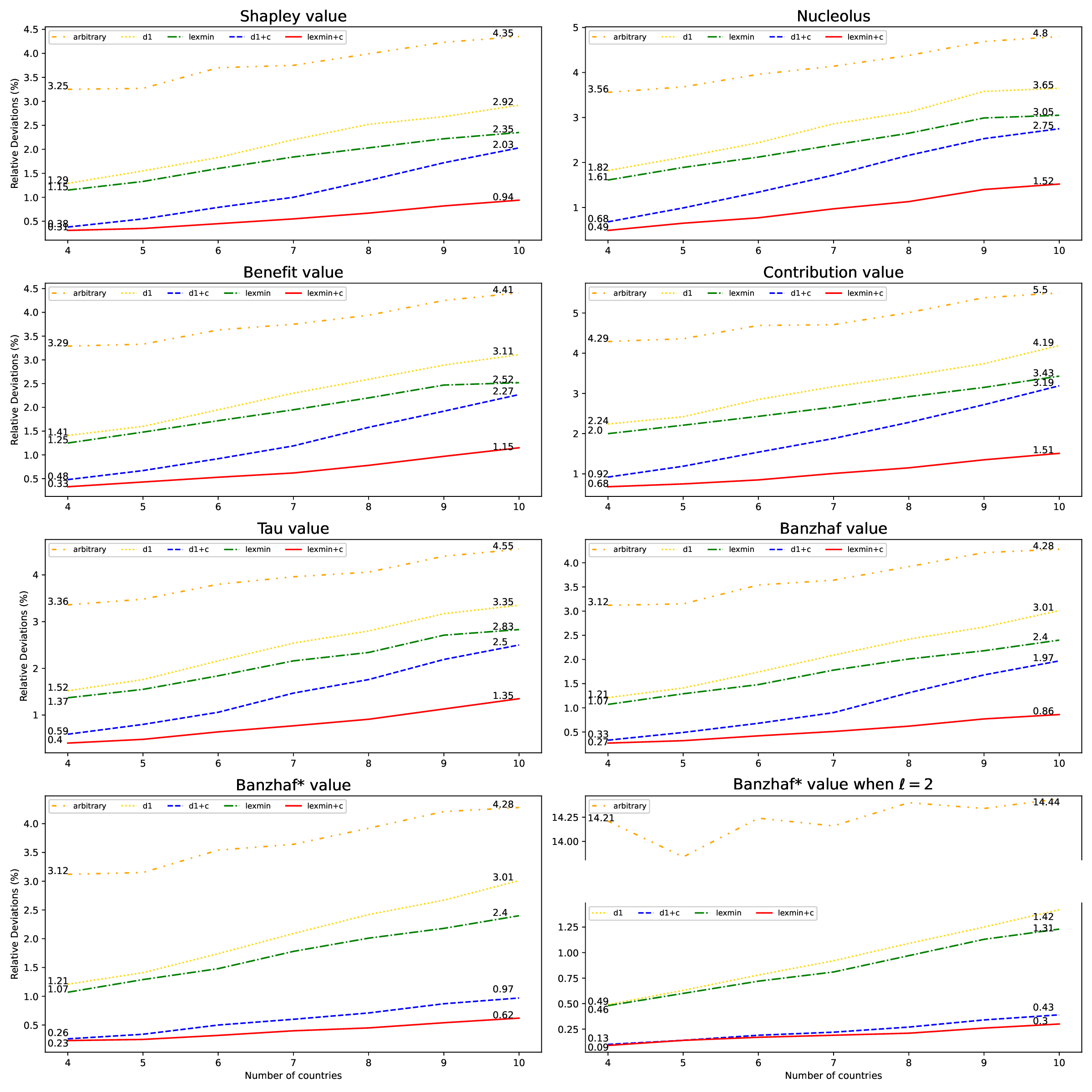}
  \caption{Average total relative deviations for each of the seven solution concepts under the five different scenarios for {\bf equal} country sizes, where the number of countries $n$ is ranging from $4$ to $10$. For comparison, the lower right figure displays a result from~\cite{BBPY24} for $\ell=2$, namely for the Banzhaf* value, which behaved best for $\ell=2$. We recall that the Banzhaf value and Banzhaf* value coincide when credits are not incorporated, and this is also reflected in the two corresponding figures.}\label{fig:absolutedeviations}
\end{figure}

The above 
{conclusions} 
are in line with the results under the setting where $\ell=2$~\cite{BBPY24}. 
However, for $\ell=2$, the effect of using arbitrary optimal solutions is much worse, while deviations are smaller than for $\ell=\infty$ when weakly close or strongly close optimal solutions are chosen.
{In Figure~\ref{fig:absolutedeviations} we illustrated this for the Banzhaf* value.}
From Figure~\ref{fig:absolutedeviations} we see that the Banzhaf* value in the \emph{lexmin+c} scenario provide{s} the smallest deviations from the target allocations (as when $\ell=2$~\cite{BBPY24}). However, all solution concepts are within 1.52\% (for \emph{lexmin+c}) and, as mentioned, which solution concept to select should be decided by the policy makers of the IKEP. 

\begin{figure}
  \centering
    \hspace*{-1cm}
  \includegraphics[scale=0.42]{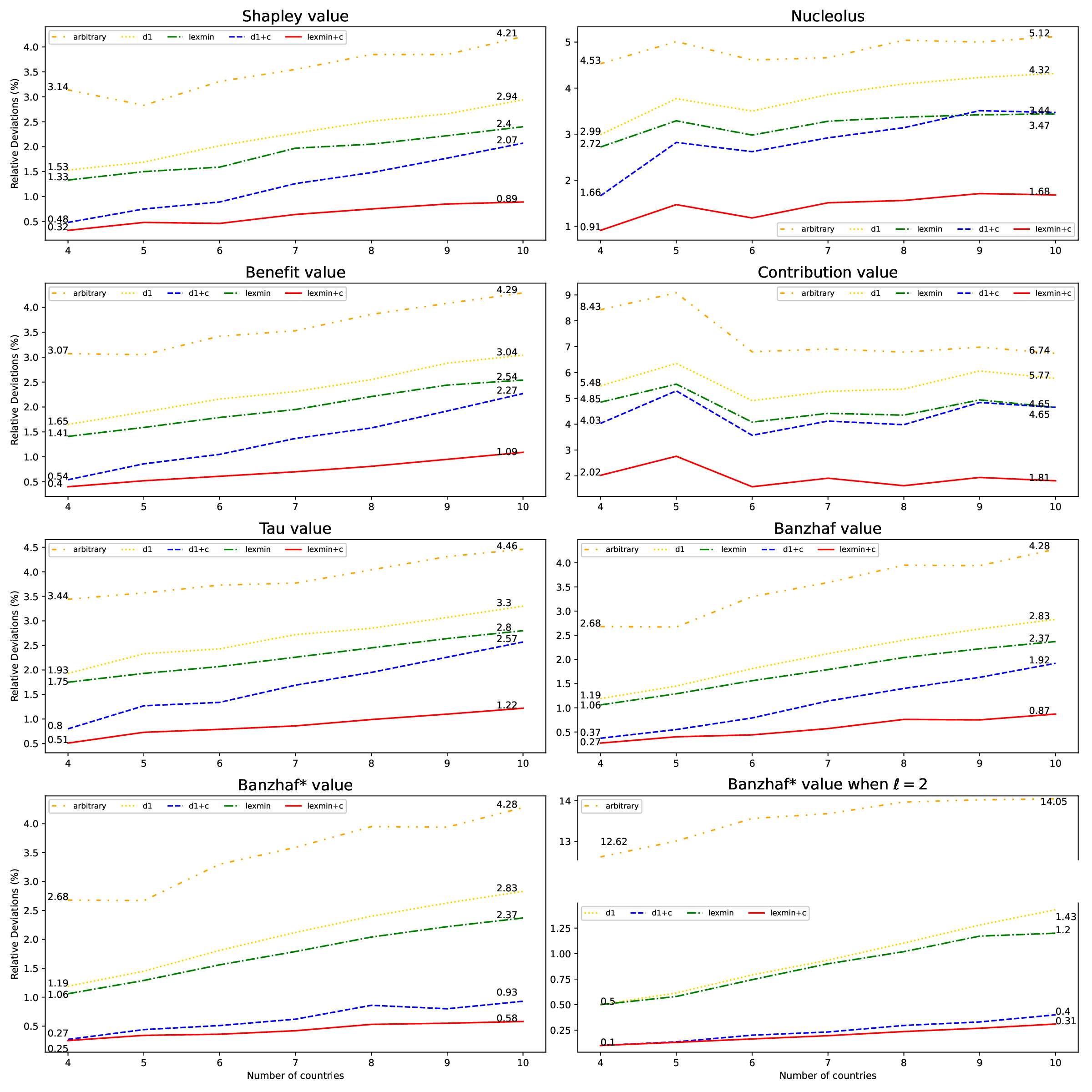}
  \caption{Average total relative deviations for each of the seven solution concepts under the five different scenarios for {\bf varying} country sizes, where the number of countries $n$ is ranging from $4$ to $10$. For comparison, the lower right figure displays a result from~\cite{BBPY24} for $\ell=2$, namely for the Banzhaf* value, which behaved best for $\ell=2$.}\label{fig:absolutedeviationsundervarying}
\end{figure}

Figure~\ref{fig:absolutedeviationsundervarying} shows the same kind of results as Figure~\ref{fig:absolutedeviations} but now for {\it varying country sizes}, that is, when we have  small, medium and large countries, divided exactly as in~\cite{BBPY24} for $\ell=2$ (see also Section~\ref{s-setup}). Note that subject to minor fluctuations we can draw the same conclusions from Figure~\ref{fig:absolutedeviationsundervarying} for varying country sizes as we did from Figure~\ref{fig:absolutedeviations} for equal country sizes.

Figures~\ref{fig:absolutedeviations} and~\ref{fig:absolutedeviationsundervarying} highlight the comparison between different scenarios. For an easier comparison between the effects of choosing different  solution concepts for prescribing the initial allocations, we grouped together all the \emph{lexmin+c} plots in of Figures~\ref{fig:absolutedeviations} and~\ref{fig:absolutedeviationsundervarying} in Figure~\ref{fig:lexmin+c} and all the \emph{d1+c} plots from Figures~\ref{fig:absolutedeviations} and~\ref{fig:absolutedeviationsundervarying} in Figure~\ref{fig:d1+c}. 
Figures~\ref{fig:lexmin+c} and~\ref{fig:d1+c} both show  
an ordering from the Banzhaf* value, which has the best performance, to the contribution value which has the worst in almost all cases for the two most important scenarios \emph{lexmin+c} and \emph{d1+c}.

We note that Figures~\ref{fig:lexmin+c} and~\ref{fig:d1+c} also show that (as expected) the effect of varying the country sizes is stronger if the number $n$ of countries is relatively small, especially when $n\in \{4,5,6\}$.

\medskip
\noindent
{\bf Our Second Evaluation Measure.}  We can draw exactly the same conclusions as above if we use the average maximum relative deviation instead of the average total relative deviation. We refer to Appendix~\ref{a-average} for the analogs of Figures~\ref{fig:absolutedeviationsundervarying}--\ref{fig:d1+c} if we use the average maximum relative deviation as our evaluation measure.

\begin{figure}
  \centering
  \includegraphics[width=1.1\linewidth]{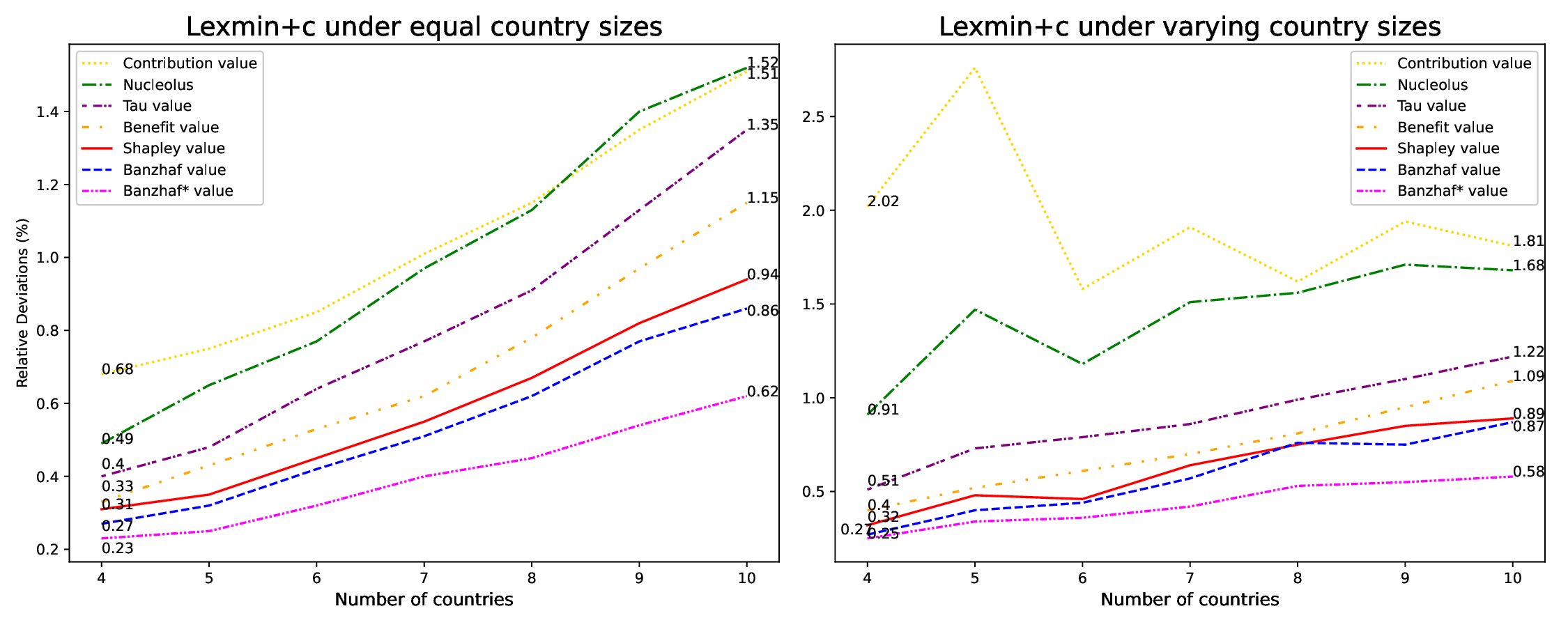}
  \caption{Average total relative deviations for all solution concepts in the \emph{lexmin+c} scenario, where the number of countries $n$ ranges from $4$ to $10$.}\label{fig:lexmin+c}
\end{figure}

\begin{figure}
  \centering
  \includegraphics[width=1.1\linewidth]{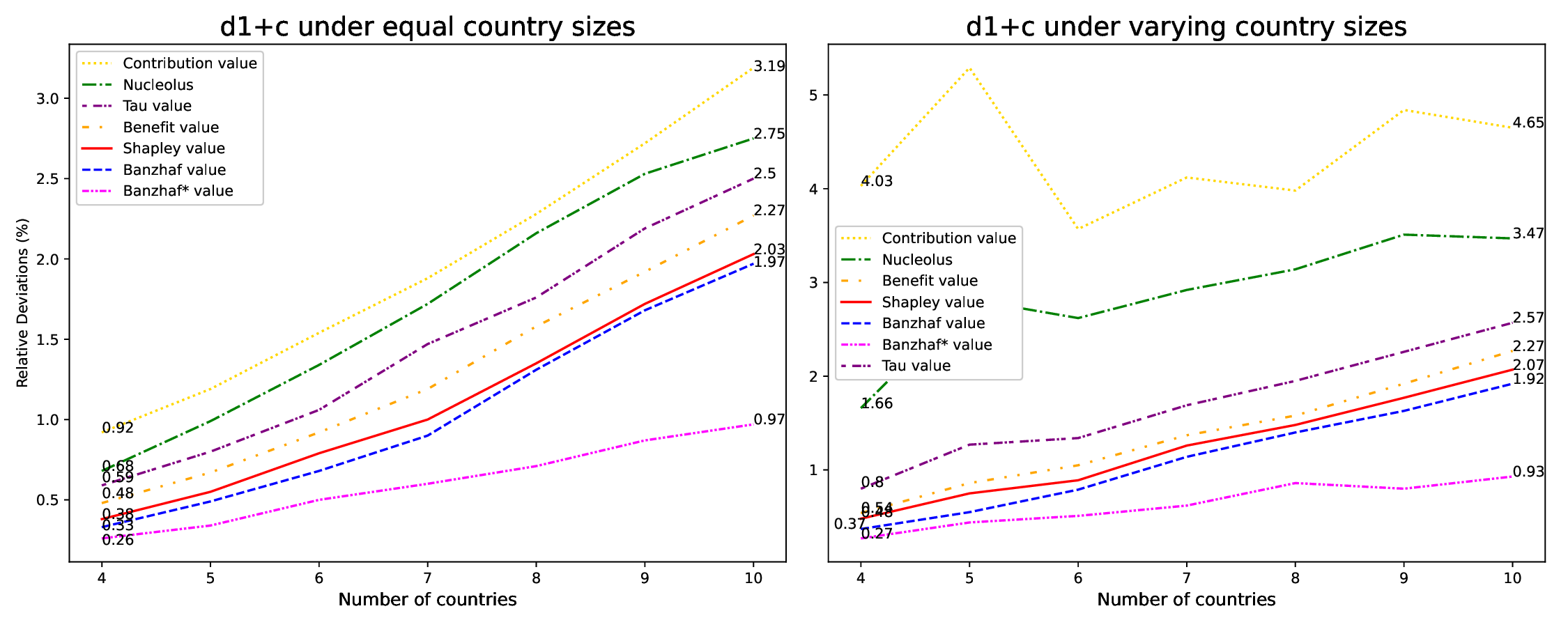}
  \caption{Average total relative deviations for all solution concepts in the \emph{d1+c} scenario, where the number of countries $n$ ranges from $4$ to $10$.}\label{fig:d1+c}
\end{figure}

\medskip
\noindent
{\bf Incomplete Instances.} In the simulations for $\ell=2$~\cite{BBPY24}, ILPs were not used, and all simulation instances were solved to optimality. However, our new simulations for $\ell=\infty$ heavily relied on ILPs, and we recall that we imposed a time limit of one hour for our ILP solver (Gurobi) to solve an ILP. Given the nature of ILPs, it is not surprising that there were several ILPs that could not be solved within the $1$-hour time limit. 
We call the corresponding simulation instances {\it incomplete}. So, for such instances, there was one ILP in some round, which our ILP solver could not handle within the $1$-hour time limit. We call
 call other simulation instances {\it complete}. 
 
All ILPs that were not finished within one hour were of the \eqref{ILPdt} type, 
almost always for $t > 1$. 
In fact, for equal country sizes, only one single incomplete simulation instance occurred at \emph{d1 / d1+c} (for nine countries, using the tau value). Tables~\ref{table:not_solved} and~\ref{table:not_solvedvaryingsizes}  summarize the distribution of incomplete simulation instances aggregated over different choices of solution concepts or scenarios. 
{See} Appendix~\ref{a-in} for  a complete breakdown of the averages in both these tables. 
{Here, we~only~observe 
from Tables~\ref{table:not_solved} and~\ref{table:not_solvedvaryingsizes}
that the average number of incomplete instances is relatively low, both for equal country sizes (at most 2\% in total) and varying country sizes (at most 0.75\% in total).}

Given that
{the average number of incomplete instances is relatively low, and moreover that}
the solutions found for incomplete simulation instances are all maximum cycle packings (which might not be weakly or strongly close), policy makers could therefore still decide to use them. For this reason, we decided to include the incomplete simulation instances in {Figures~\ref{fig:absolutedeviations}--\ref{fig:d1+c}}. We believe this was justified after doing some additional research. That is, we also constructed the same figures as Figures~\ref{fig:absolutedeviations} and~\ref{fig:absolutedeviationsundervarying} but {\it without} the incomplete simulation instances; see Appendix~\ref{a-in}. It turned out that the largest percentage points difference in Figure~\ref{fig:absolutedeviations} of the average total deviations is only 0.059\% (with an average of 0.0024\%). Hence, the quality of the ``current-best'' optimal solutions for the incomplete simulation instances 
are almost indistinguishably close to those for the complete simulation instances.

\setlength{\tabcolsep}{0.3cm}
\begin{table}[h]
\centering
\begin{tabular}{@{}lllllllll@{}}
\toprule
\ \#unsolved / n   & 4    & 5    & 6    & 7    & 8    & 9    & 10 & Total \\ \midrule
Shapley value      & 0.00\% & 0.00\% & 0.00\% & 1.50\% & 1.75\% & 6.75\% & 3.50\% & 1.93\% \\ 
Nucleolus          & 0.00\% & 0.00\% & 0.00\% & 0.50\% & 0.00\% & 0.00\% & 0.25\% & 0.11\% \\
Benefit value      & 0.25\% & 0.25\% & 1.25\% & 0.25\% & 0.75\% & 0.75\% & 0.00\% & 0.50\% \\
Contribution value & 0.75\% & 0.75\% & 1.75\% & 1.00\% & 1.50\% & 0.50\% & 2.25\% & 1.21\% \\
Banzhaf value      & 0.00\% & 0.75\% & 1.50\% & 1.75\% & 0.50\% & 4.25\% & 4.50\% & 1.89\% \\
Tau value          & 0.00\% & 0.25\% & 0.00\% & 0.00\% & 0.00\% & 1.00\% & 0.00\% & 0.18\% \\
Banzhaf* value      & 0.00\% & 0.00\% & 0.00\% & 0.00\% & 0.00\% & 0.00\% & 0.50\% & 0.07\% \\ \midrule
total: d1          & 0.00\% & 0.00\% & 0.00\% & 0.00\% & 0.00\% & 0.17\% & 0.00\% & 0.02\% \\
total: lexmin      & 0.17\% & 0.33\% & 0.50\% & 1.33\% & 0.83\% & 4.17\% & 3.33\% & 1.52\% \\
total: d1+c        & 0.00\% & 0.00\% & 0.00\% & 0.00\% & 0.00\% & 0.14\% & 0.00\% & 0.02\% \\
total: lexmin+c    & 0.43\% & 0.86\% & 2.14\% & 1.71\% & 1.86\% & 3.71\% & 3.29\% & 2.00\% \\
\bottomrule
\end{tabular}
\medskip
\caption{Average number of incomplete simulation instances for {\bf equal} country sizes.}\label{table:not_solved} 
\end{table}
\setlength{\tabcolsep}{0.3cm}
\begin{table}[h]
\vspace*{0cm}
\centering
\begin{tabular}{@{}lllllllll@{}}
\toprule
\#unfinished / n & 4    & 5    & 6    & 7    & 8    & 9    & 10   & Total
\\ \midrule
Shapley value      & 0.00\% & 0.00\% & 0.00\% & 0.50\% & 0.00\% & 0.50\% & 2.00\% & 0.43\% \\
Nucleolus          & 0.00\% & 0.00\% & 0.00\% & 0.00\% & 0.00\% & 0.75\% & 0.00\% & 0.11\% \\
Benefit value      & 0.25\% & 0.00\% & 0.75\% & 0.00\% & 0.00\% & 0.25\% & 0.75\% & 0.29\% \\
Contribution value & 0.00\% & 0.00\% & 0.50\% & 0.00\% & 0.50\% & 0.50\% & 0.25\% & 0.25\% \\
Banzhaf value      & 0.00\% & 1.00\% & 0.25\% & 0.25\% & 1.25\% & 0.50\% & 2.00\% & 0.75\% \\
Tau value          & 0.00\% & 0.00\% & 0.00\% & 0.00\% & 0.25\% & 0.00\% & 0.00\% & 0.04\% \\
Banzhaf* value     & 0.00\% & 1.00\% & 0.00\% & 0.00\% & 0.50\% & 0.50\% & 2.00\% & 0.57\% \\
 \midrule
total: d1          & 0.00\% & 0.17\% & 0.00\% & 0.00\% & 0.00\% & 0.17\% & 0.00\% & 0.05\% \\
total: lexmin      & 0.00\% & 0.17\% & 0.33\% & 0.33\% & 0.50\% & 0.50\% & 1.67\% & 0.50\% \\
total: d1+c        & 0.00\% & 0.29\% & 0.14\% & 0.00\% & 0.14\% & 0.00\% & 0.00\% & 0.08\% \\
total: lexmin+c    & 0.14\% & 0.29\% & 0.43\% & 0.14\% & 0.71\% & 1.00\% & 2.00\% & 0.67\% \\
\bottomrule
\end{tabular}
\medskip
\caption{Average number of incomplete simulation instances for {\bf varying} country sizes.}\label{table:not_solvedvaryingsizes} 
\end{table}
\noindent
{\bf Number of Kidney Transplants {and} Cycle Length.}
When $\ell=\infty$ instead of $\ell=2$, we expect more kidney transplants, {because} exchange cycles may now have any size. In {Tables~\ref{t-num=10}-\ref{t-varying-num=10} and Figure~\ref{fig:cycledistribution}} we quantified this.

{First,} from {Table~\ref{t-num=10} we see that for $n=10$ and equal country sizes 
it is possible to achieve 46\% more kidney transplants when $\ell=\infty$ instead of $\ell=2$, and that this is irrespective of the chosen scenario or chosen solution concept}, as the total number of kidney transplants are nearly identical, 
{both for $\ell=2$ and $\ell=\infty$.
Table~\ref{t-varying-num=10} shows the results for varying country sizes, for which we can draw the same conclusions. We refer to 
Tables~\ref{t-num-equal} and~\ref{t-num-varying} in Appendix~\ref{a-average2} for 
the corresponding results for $n=4,\ldots,9$, which show the same behaviour.} 

{Second, Figure~\ref{fig:cycledistribution}
shows that setting $\ell=\infty$ may indeed lead to long cycles, with even more than 400 vertices.
Figure~\ref{fig:cycledistribution} shows in particular that these long cycles all occur in the first round.
The reason is that 
we start the first round with roughly 500 vertices and introduce the remaining 1500 patient-donor pairs in the later rounds by a uniform distribution. Hence, the compatibility graphs in later rounds are all smaller.} 

\setlength{\tabcolsep}{0.25cm}
\begin{table}[h]
\centering
\begin{tabular}{cllllll}
\toprule
\multicolumn{2}{c}{Solution   concepts/scenarios}   & $\emph{arbitrary}$ & $\emph{d1}$ & $\emph{d1+c}$ & $\emph{lemxin}$ & $\emph{lexmin+c}$ \\
\midrule
\multirow{7}{*}{$\ell=\infty$} & benefit value      & 1781.58            & 1781.65     & 1781.46       & 1782.58         & 1780.31           \\

                               & contribution value & 1781.58            & 1781.52     & 1780.50       & 1782.52         & 1780.05           \\
                               & Nucleolus          & 1781.58            & 1781.80     & 1780.45       & 1781.98         & 1780.12           \\
                               & Shapley value      & 1781.58            & 1781.54     & 1780.92       & 1782.02         & 1780.57           \\
                               & Banzhaf value      & 1781.58            & 1781.83     & 1780.38       & 1782.46         & 1780.66           \\
                               & Banzhaf* value     & 1781.58            & 1781.83     & 1781.08       & 1782.46         & 1781.24           \\
                               & Tau value          & 1781.58            & 1782.22     & 1780.90       & 1782.12         & 1780.45           \\
                               \midrule
\multirow{7}{*}{$\ell=2$}      & benefit value      & 1221.38            & 1222.4      & 1223.44       & 1222.54         & 1223.6            \\
                               & contribution value & 1221.38            & 1222.58     & 1222.56       & 1223.24         & 1222.66           \\
                               & Nucleolus          & 1221.38            & 1222.28     & 1221.18       & 1221.84         & 1221.66           \\
                               & Shapley value      & 1221.38            & 1221.6      & 1221.2        & 1223.04         & 1222.32           \\
                               & Banzhaf value      & 1221.38            & 1222.62     & 1221.26       & 1220.9          & 1221.98           \\
                               & Banzhaf* value     & 1221.38            & 1222.62     & 1220.86       & 1220.9          & 1221.82           \\
                               & Tau value          & 1221.38            & 1223.96     & 1223.24       & 1224.26         & 1222.72     \\
                               \bottomrule
\end{tabular}
\medskip
\caption{{Average number of kidney transplants for {\bf equal} country sizes for $n=10$ when $\ell=2$~\cite{BBPY24} and $\ell=\infty$.}}\label{t-num=10}
\end{table}
\setlength{\tabcolsep}{0.25cm}
\begin{table}[h]
\centering
\begin{tabular}{cllllll}
\toprule
\multicolumn{2}{c}{Solution   concepts/scenarios}   & $\emph{arbitrary}$ & $\emph{d1}$ & $\emph{d1+c}$ & $\emph{lemxin}$ & $\emph{lexmin+c}$ \\
\midrule
\multirow{7}{*}{$\ell=\infty$} & benefit value      & 1763.57            & 1763.39     & 1763.37       & 1763.73         & 1762.45           \\
                               & contribution value & 1763.57            & 1763.45     & 1762.90       & 1763.80         & 1761.84           \\
                               & Nucleolus          & 1763.66            & 1764.02     & 1762.79       & 1763.74         & 1762.19           \\
                               & Shapley value      & 1764.04            & 1763.09     & 1762.99       & 1763.71         & 1762.60           \\
                               & Banzhaf value      & 1763.24            & 1763.34     & 1762.85       & 1763.47         & 1762.85           \\
                               & Banzhaf* value     & 1763.24            & 1763.34     & 1763.09       & 1763.47         & 1762.97           \\
                               & Tau value          & 1763.27            & 1763.54     & 1762.99       & 1763.62         & 1761.84           \\
                               \midrule
\multirow{7}{*}{$\ell=2$}      & benefit value      & 1203.8             & 1206.14     & 1205.92       & 1205.06         & 1205.88           \\
                               & contribution value & 1204.06            & 1205.7      & 1205.46       & 1205.36         & 1206.16           \\
                               & Nucleolus          & 1204.06            & 1205.34     & 1205.38       & 1205.52         & 1207              \\
                               & Shapley value      & 1204.06            & 1205.7      & 1206.64       & 1206.34         & 1205.48           \\
                               & Banzhaf value      & 1204.06            & 1206.14     & 1205.56       & 1205.88         & 1206.34           \\
                               & Banzhaf* value     & 1204.06            & 1206.14     & 1206.34       & 1205.88         & 1205.92           \\
                               & Tau value          & 1204.06            & 1206.86     & 1205.92       & 1207.32         & 1206.3\\
                               \bottomrule
\end{tabular}
\medskip
\caption{{Average number of kidney transplants for {\bf varying} country sizes for $n=10$ when $\ell=2$~\cite{BBPY24} and $\ell=\infty$.}}\label{t-varying-num=10}
\vspace*{-3cm}
\end{table}
\clearpage
\begin{figure}
\vspace*{0cm}
  \centering
  \hspace*{-1cm}
  \includegraphics[width=1.2\linewidth]{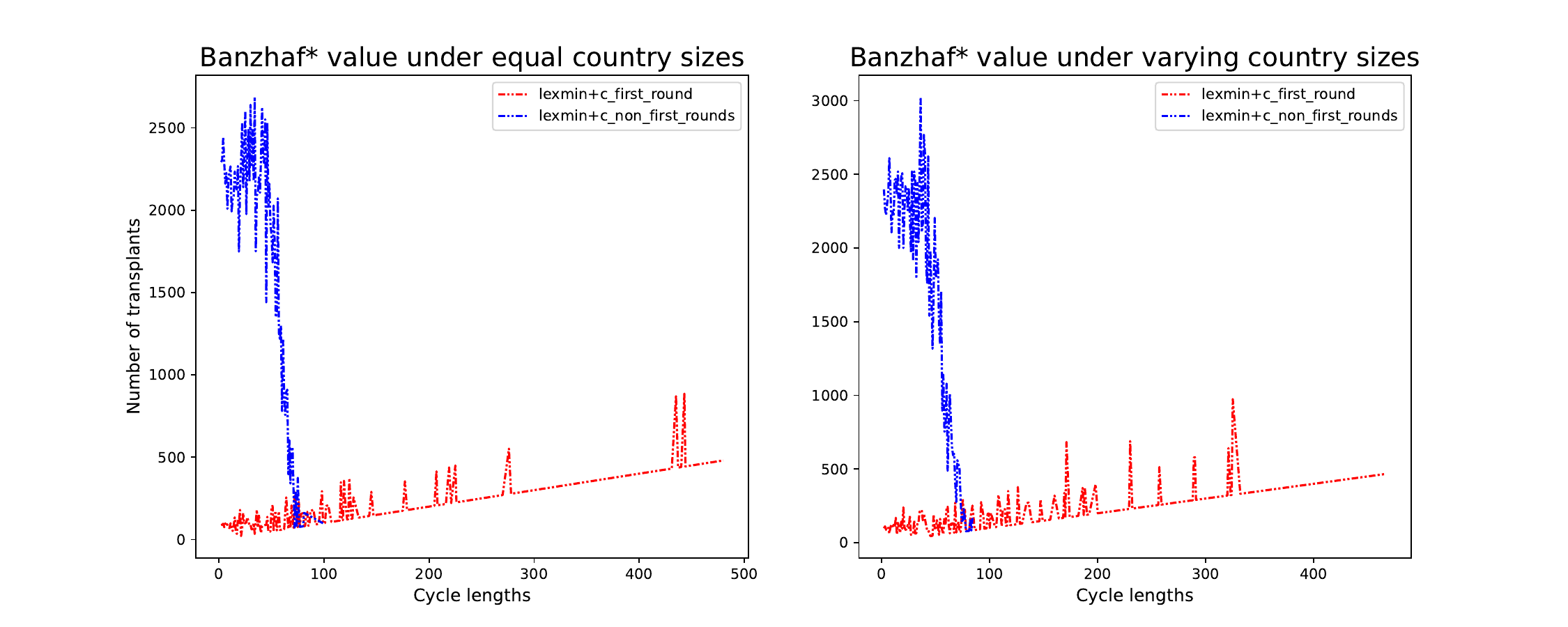}
  \caption{Cycle distribution for the Banzhaf* value in the \emph{lexmin+c} scenario, where the x-axis represents the cycle length, and the y-axis shows the number of kidney transplants involved in a cycle of that length.}\label{fig:cycledistribution}
\end{figure}

\noindent
{\bf Influence of Using Credits.}
We now discuss the power of our credit system. Does it effectively reduce deviations from the target allocation? Can it keep deviations on a consistently low level in the long-term?
Theoretically it is possible that credits keep accumulating, which would make them ineffective (see~\cite{BBPY24} for a theoretical example of this behaviour when $\ell=2$).

In Figures~\ref{fig:totaldeviations} and~\ref{fig:totaldeviationsundervarying} we show to what extent credits accumulate for equal and varying country sizes, respectively.
If credits are not incorporated, we can still compute and track them as we did in these two figures. We note that over a period of 24 rounds, credits accumulate more and more (but at different rates) under all three scenarios \emph{arbitary}, \emph{d1} and \emph{lexmin}, even though especially under \emph{d1} and \emph{lexmin} we do find solutions that are relatively close to the target allocations (as we saw from Figures~\ref{fig:absolutedeviations} and~\ref{fig:absolutedeviationsundervarying}).
However, only when we {\it also} incorporate credits we not only find solutions that are close to the target solutions but that also ensure stability.
We find that the Banzhaf* value under  \emph{lexmin+c} is also the best in maintaining consistently low levels of credits, much like in the case of $\ell=2$~\cite{BBPY24}. Again, the main difference between $\ell=\infty$ and $\ell=2$ is that the \emph{arbitrary} scenario performs much worse for $\ell=2$. 

\begin{figure}
  \centering
  \hspace*{-1cm}
  \includegraphics[scale=0.42]{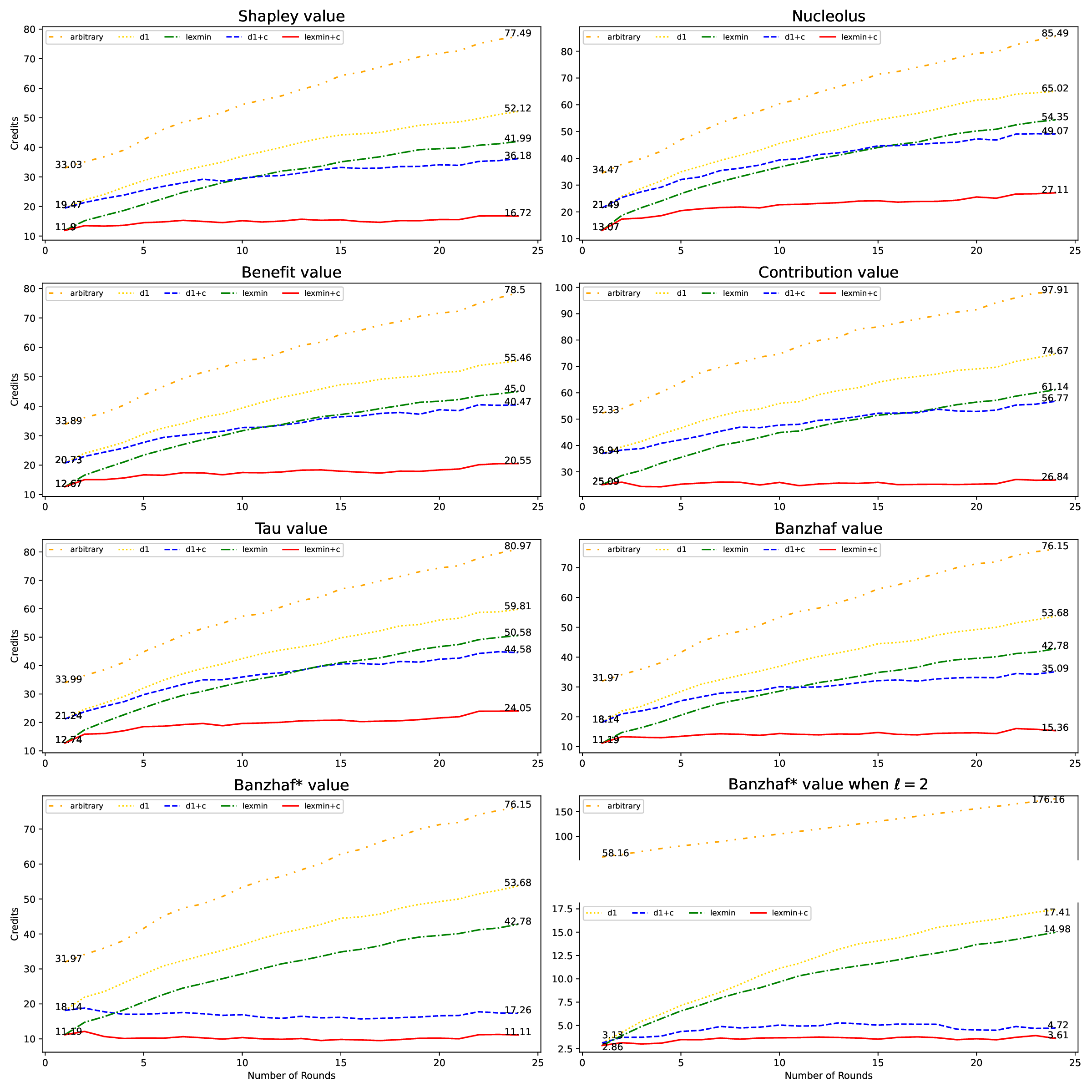}
  \caption{Average credits for each of the seven solution concepts under the five different scenarios for {\bf equal} country sizes, where the number of countries is $n=10$ and the period ranges from $1$ to $24$. For comparison, the lower right figure displays a result from~\cite{BBPY24} for $\ell=2$, namely for the Banzhaf* value, which also behaved best for $\ell=2$.}\label{fig:totaldeviations}
\end{figure}

\begin{figure}
  \centering
  \hspace*{-1cm}
  \includegraphics[scale=0.42]{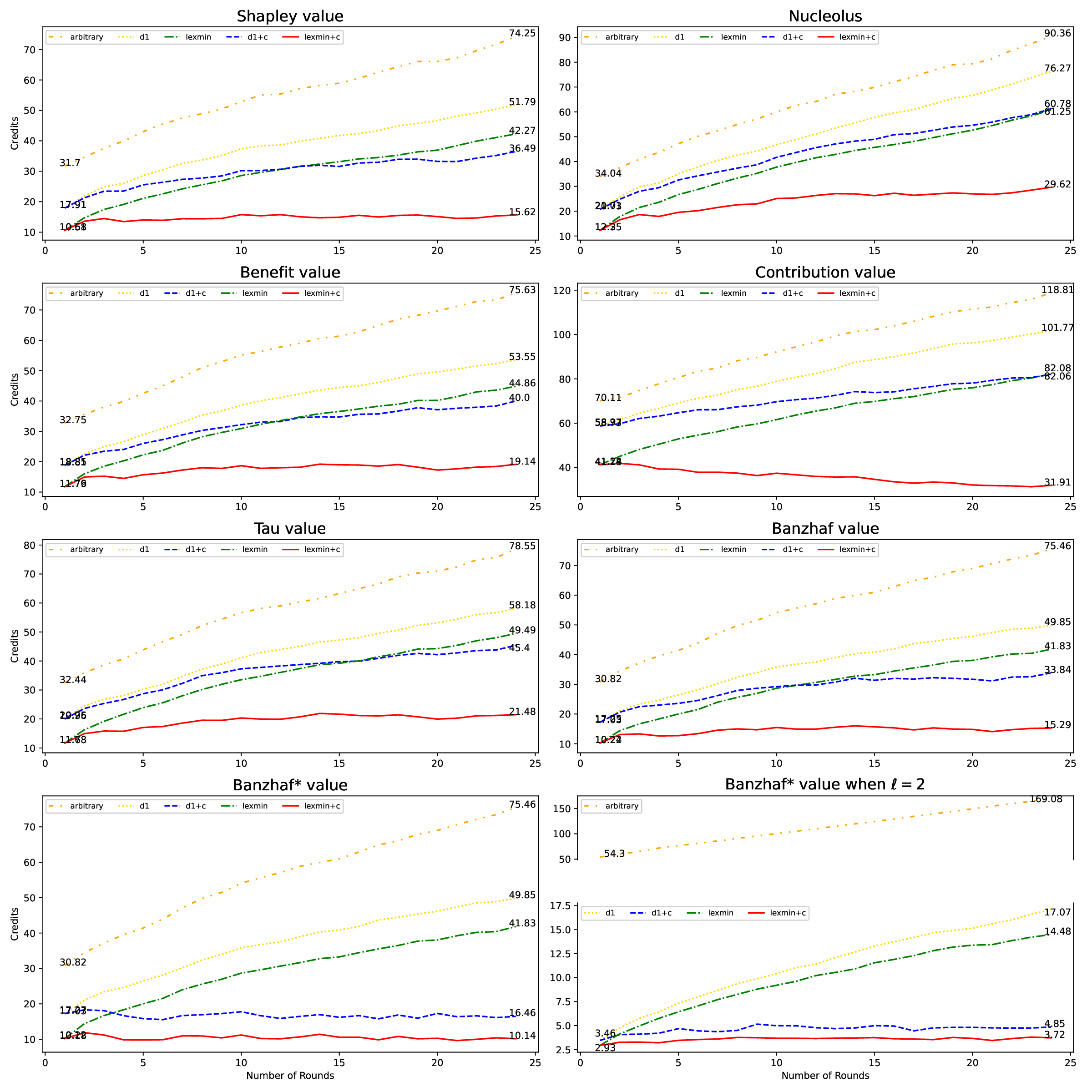}
  \caption{Average credits for each of the seven solution concepts under the five different scenarios for {\bf varying} country sizes, where the number of countries is $n=10$ and the period ranges from $1$ to $24$. For comparison, the lower right figure displays a result from~\cite{BBPY24} for $\ell=2$, namely for the Banzhaf* value, which also behaved best for $\ell=2$.}\label{fig:totaldeviationsundervarying}
\end{figure}

\medskip
\noindent
{\bf Deviations for Worst-off Countries.}
Finally, we consider the ratio of the average maximum relative deviation and the average total relative deviation as a measure for 
the proportion of the total country deviations being due to the worst-off country. 
We call this ratio the {\it relative ratio}. For stability reasons, low relative ratios are preferable. 
From Table~\ref{table:concentration} we conclude that for the most interesting scenario \emph{lexmin+c} and equal country sizes, the relative ratio decreases as the number of countries increases (as expected). In this figure we also included the relative ratios for $\ell=2$ (generated from the data used in~\cite{BBPY24}). We note that the relative ratio is slightly better for $\ell=2$. This extends to varying countries as well; see Appendix~\ref{a-conc}. In this appendix we also included the figures for the other scenarios as well. From these figures, we note that for \emph{d1+c} the relative ratio for $\ell=\infty$ is lower than for $\ell=2$. 
The \emph{arbitrary} scenario showed consistently lower deviations for $\ell=\infty$ than for $\ell=2$. However,  this is not the case for the relative ratio. For equal country sizes, $\ell=\infty$ leads to consistently higher relative ratios than $\ell=2$ does, whereas for varying country set sizes the opposite holds. For the remaining scenarios (\emph{lexmin} and \emph{d1}), there is no clear distinction.

\setlength{\tabcolsep}{0.25cm}
\begin{table}[h]
\vspace*{-0cm}
\centering
\begin{tabular}{@{}lllllllll@{}}
\toprule
 \multicolumn{2}{c}{n}   & 4    & 5    & 6    & 7    & 8    & 9    & 10   \\ \midrule
\multirow{3}{*}{\minitab[c]{Shapley\\value}}  & $\ell=\infty$ & 46.82\% & 39.70\% & 38.99\% & 34.55\% & 31.40\% & 31.41\% & 28.17\% \\ 
 & $\ell=2$ & 41.70\% & 37.92\% & 35.29\% & 30.85\% & 28.00\% & 26.71\% & 26.52\% \\ 
 & difference & 5.12\% & 1.78\% & 3.70\% & 3.70\% & 3.40\% & 4.69\% & 1.66\% \\ \midrule
\multirow{3}{*}{\minitab[c]{Banzhaf\\value}}  & $\ell=\infty$ & 45.91\% & 40.92\% & 39.50\% & 36.87\% & 30.95\% & 30.57\% & 27.61\% \\ 
 & $\ell=2$ & 43.90\% & 37.07\% & 34.88\% & 30.05\% & 29.27\% & 26.89\% & 26.26\% \\ 
 & difference & 2.01\% & 3.85\% & 4.62\% & 6.82\% & 1.68\% & 3.68\% & 1.35\% \\ \midrule
 \multirow{3}{*}{\minitab[c]{nucleolus}} & $\ell=\infty$ & 47.35\% & 39.50\% & 38.25\% & 35.21\% & 32.95\% & 29.81\% & 27.02\% \\ 
 & $\ell=2$ & 45.85\% & 40.05\% & 36.57\% & 33.68\% & 30.84\% & 29.44\% & 26.75\% \\ 
 & difference & 1.51\% & -0.54\% & 1.67\% & 1.52\% & 2.11\% & 0.36\% & 0.27\% \\ \midrule
  \multirow{3}{*}{\minitab[c]{tau\\value}}  & $\ell=\infty$ & 47.11\% & 41.16\% & 39.47\% & 34.02\% & 33.91\% & 31.57\% & 27.07\% \\ 
 & $\ell=2$ & 42.00\% & 38.53\% & 37.70\% & 32.41\% & 30.31\% & 28.68\% & 27.69\% \\ 
 & difference & 5.11\% & 2.63\% & 1.78\% & 1.62\% & 3.60\% & 2.89\% & -0.62\% \\ \midrule
   \multirow{3}{*}{\minitab[c]{benefit\\value}}  & $\ell=\infty$ & 46.24\% & 40.37\% & 39.72\% & 34.89\% & 34.08\% & 30.44\% & 27.10\% \\ 
 & $\ell=2$ & 42.91\% & 38.13\% & 36.62\% & 30.11\% & 31.76\% & 28.57\% & 26.02\% \\ 
 & difference & 3.33\% & 2.25\% & 3.11\% & 4.78\% & 2.32\% & 1.87\% & 1.08\% \\ \midrule
    \multirow{3}{*}{\minitab[c]{contribution\\value}}  & $\ell=\infty$ & 44.41\% & 37.06\% & 38.30\% & 34.77\% & 31.76\% & 28.55\% & 26.41\% \\ 
 & $\ell=2$ & 43.19\% & 40.15\% & 37.50\% & 31.94\% & 29.80\% & 28.19\% & 25.90\% \\ 
 & difference & 1.23\% & -3.09\% & 0.80\% & 2.83\% & 1.96\% & 0.36\% & 0.51\% \\ \midrule
 \multirow{3}{*}{\minitab[c]{Banzhaf*\\value}}  & $\ell=\infty$ & 45.12\% & 40.51\% & 38.89\% & 35.16\% & 31.22\% & 29.59\% & 27.22\% \\ 
 & $\ell=2$ & 43.33\% & 37.57\% & 34.47\% & 30.42\% & 28.48\% & 27.23\% & 25.33\% \\ 
 & difference & 1.78\% & 2.94\% & 4.42\% & 4.74\% & 2.74\% & 2.36\% & 1.89\% \\
\bottomrule
\end{tabular}
\medskip
\caption{Relative ratios for equal country sizes under \emph{lexmin+c}, for $\ell=\infty$, $\ell = 2$ and their difference.}\label{table:concentration} 
\end{table}

\subsection{Computation Times}

\medskip
\noindent
Table~\ref{table:computational time} summarizes the average computational time of solving a single (24-round) simulation instance. As expected, computing initial allocations using the two easy-to-compute solution concepts (benefit and contribution values) is inexpensive, especially compared to the other, more sophisticated but hard-to-compute solution concepts.

In Table~\ref{table:computational time} we also see that the computation time for the \emph{arbitrary} scenario, which does not require initial allocations,  is less time-consuming than for the other four scenarios. This is in line with Theorem~\ref{t-hard}, which gave us a polynomial time algorithm for computing an arbitrary optimal solution, whereas we needed to rely on using ILPs for computing optimal solutions that are weakly close (\emph{d1} and \emph{d1+c}) or even strongly close  (\emph{lexmin} and \emph{lexmin+c}). From Table~\ref{table:computational time}, it can be clearly seen that 
the computation of the ILP series requires the majority of the work, especially for the \emph{lexmin} and \emph{lexmin+c} scenarios: these involve longer series of ILPs than \emph{d1} and \emph{d1+c}, for which we need to solve only \eqref{ILPd1}.

On a side note, we observe from Table~\ref{table:computational time}  that when $n$ increases, the computation time increases roughly within the expected rate, apart from a few notable exceptions: Gurobi's ILP solver handles, for some reason, \eqref{ILPd1} far more easily for $n=6$ than for $n=5$ (or even $n=4$). We have no explanation for this; it may well be due to the nature of Gurobi, in which we have no insights.

\begin{table}[h]
\centering
\begin{tabular}{@{}llllllll@{}}
\toprule
CPU time   / n     & 4       & 5       & 6       & 7       & 8       & 9       & 10      \\ \midrule
data   preparation & 0.01   & 0.01    & 0.01    & 0.01    & 0.01    & 0.01    & 0.01    \\
graph building     & 0.04   & 0.04    & 0.04    & 0.04    & 0.04    & 0.04    & 0.04    \\ \midrule
Shapley value      & 12.32  & 25.28   & 48.32   & 94.70   & 188.7  & 393  & 745.2  \\
Nucleolus          & 10.75  & 22.33   & 44.41   & 86.83   & 165.5  & 342  & 707.8  \\
Benefit value      & 7.87   & 8.41    & 9.64    & 11.31   & 12.68   & 14.45   & 15.81   \\
Contribution value & 7.39   & 8.80    & 10.35   & 12.44   & 14.15   & 16.68   & 17.00   \\
Banzhaf value    & 10.56          & 21.73   & 42.06   & 89.00   & 174.75  & 343.06  & 689.62  \\
Banzhaf* value & 13.03    & 25.48   & 50.98   & 89.12    & 173.76  & 359.56  & 687.96   \\
Tau value          & 17.76  & 29.94   & 56.70   & 99.66   & 199  & 358.1  & 796.2  \\ \midrule
total:   arbitrary & 11.63  & 20.35   & 35.35   & 65.49   & 123.81  & 257.33  & 530.78  \\
total: d1          & 66.81       & 85.85   & 42.97   & 127.49  & 171.84  & 315.08  & 473.38  \\
total: lexmin      & 305.04 & 688.20  & 1044.89 & 1929.32 & 2771.54 & 4912.33 & 6737.04 \\
total: d1+c        & 93.26  & 87.09   & 43.93   & 121.11  & 166.00  & 330.56  & 504.52  \\
total: lexmin+c    & 470.29 & 1030.02 & 2010.87 & 2514.62 & 4032.30 & 5764.23 & 7990.36 \\ \bottomrule
\end{tabular}
\medskip
\caption{Average CPU time for a single 24-round simulation instance, broken down into the different computational tasks. Here, the times for data preparation and graph building are average times taken over all scenarios and solution concepts, whereas the time for each solution concept is the average time taken over all scenarios. The total times for the scenarios are average times taken over all solutions concepts, where ``total'' refers to the total computation time {(using up to $32$ threads for each simulation instance)}, which includes computing the initial allocations, and only excludes the time for data preparation and graph building.}\label{table:computational time} 
\end{table}

\section{Conclusions}\label{s-con}

We introduced the class of partitioned permutation games and
{first}
proved a number of complexity results.

{In particular, we showed that  every partitioned permutation game has a nonempty core and that we can compute a core allocation of a partitioned permutation game in polynomial time. 
We also proved that for partitioned permutation games of fixed width~$c$, the problem of deciding if an allocation is in the core is polynomial-time solvable if $c=1$ but is already co\NP-complete if $c=2$. Moreover, we proved that in this setting even the problem of finding a weakly optimal solution for a given target allocation~$x$ is \NP-hard, and we gave an exact randomized XP-time algorithm for solving this problem. Our theoretical results contrast some known results for partitioned $\ell$-permutation games where $\ell$ is some constant, as surveyed in Table~\ref{dichotomies}.
As an open theoretical problem we ask: }

\begin{open}
{Determine} the complexity of computing the nucleolus for permutation games and for partitioned permutation games.
\end{open}

\noindent
Our new results guided our simulations for IKEPs, {where we could include} 
up to ten countries, with exchange bound $\ell=\infty$. 
{Our simulations showed the clear benefits of using a credit system with strongly close optimal solutions, namely an improvement in balance up to 56\%. The exact improvement depends on the choice of solution concept, with the Banzhaf* value yielding the best results, namely on average, a deviation of at most 0.90\%. 
Our simulations showed a significant improvement (46\% on average) in the total number of kidney transplants if cycles of any length are allowed compared to the case where $\ell=2$~\cite{BBPY24}.
In our simulations, we first let all countries be of the same size. We then examined, in the same way as was done in~\cite{BBPY24}, whether our conclusions would change for varying country sizes. Just as in~\cite{BBPY24}, this turned out not to be the case.}

{For future research, it would be interesting to determine to what extent our theoretical and experimental results change if we allow non-directed (altruistic) donors or multiple donors registering for one patient. We did not consider such types of donors for this paper, as our goal was to compare our results with the previously known results for $\ell=2$, where non-directed donors and multiple donors for one patient were also excluded.}
{Moreover, we recall} that matching games and permutation games 
{may be} defined on edge-weighted graphs{, in which each edge $e\in A$ has some weight $w(e)$ reflecting the expected utility of the corresponding kidney transplant.} It would be interesting to do simulations in the presence of edge weights. 
{Furthermore, we recall that more complex hierarchical optimization criteria than only optimizing the number of transplants are used in European kidney exchange~\cite{Bi_etal2019,Bi_etal2021}.} We leave {conducting simulations with such criteria} for future research as well.

{Finally}, for future research we will consider the more realistic exchange bounds $\ell\in \{3,4,5\}$.
We note that the simulations done in~\cite{BGKPPV20,KNPV20} were only for $\ell=3$; a more limited number of solution concepts; for IKEPs with up to four countries; and only for scenarios that use weakly close optimal solutions. They also used different data sets. For our follow-up study for $\ell\in \{3,4,5\}$, we must now also overcome, as in~\cite{BGKPPV20,KNPV20},
 the additional computational obstacle of not being able to compute an optimal solution for a compatibility graph in a KEP round and the values $v(S)$ of the associated permutation game
in polynomial time (see Theorem~\ref{t-hard}).
Current techniques for $\ell\in \{3,4,5\}$ therefore involve, besides ILPs based on the edge-formulation, ILPs based on the cycle-formulation, with a variable for each cycle of length at most~$\ell$ (see, for example,~\cite{CKVR13,DGGKMP,DMPST16}). Hence, computing a single value~$v(S)$ will become significantly more expensive, and even more so for increasing $\ell$. For expensive solution concepts, such as the Shapley or nucleolus, we must compute an exponential number of values $v(S)$. Without any new methods, we expect it will not be possible to do this for simulations up to the same number of countries (ten) as we did for $\ell=\infty$.

\medskip
\noindent
{\it Acknowledgments.}
Benedek was supported by the National Research, Development and Innovation Office of Hungary (OTKA Grant No.\ K138945); Bir\'o by the Hungarian Scientific Research Fund (OTKA, Grant No.\ K143858) and the Hungarian Academy of Sciences (Momentum Grant No. LP2021-2); Cs\'aji by the Hungarian Scientific Research Fund, OTKA, Grant No. K143858, and by the Doctoral Student Scholarship Program (number C2258525) of the Cooperative Doctoral Program of the Ministry of Innovation and Technology financed by the National Research, Development, and Innovation Fund; and Paulusma was supported by the Leverhulme Trust (Grant~RF-2022-607) and EPSRC (Grant EP/X01357X/1). 
This work has made use of the Hamilton HPC Service of Durham University. In particular, we thank Rob Powell for all his help with the initial set up of our simulations.

We also thank the three anonymous reviewers of our paper for all their useful comments for improving the readability of our paper. 

\bibliographystyle{splncs04}
\bibliography{jair2}

\appendix

\section{Average Maximum Relative Deviations}\label{a-average}

In this section we consider the average maximum relative deviation as our evaluation measure.
We refer to Figures~\ref{fig:maxdeviations}--\ref{fig:max_d1+c} for similar figures as in Section~\ref{s-simul} where we used the average total relative deviation.

\begin{figure}
  \centering
    \hspace*{-1cm}
  \includegraphics[scale=0.42]{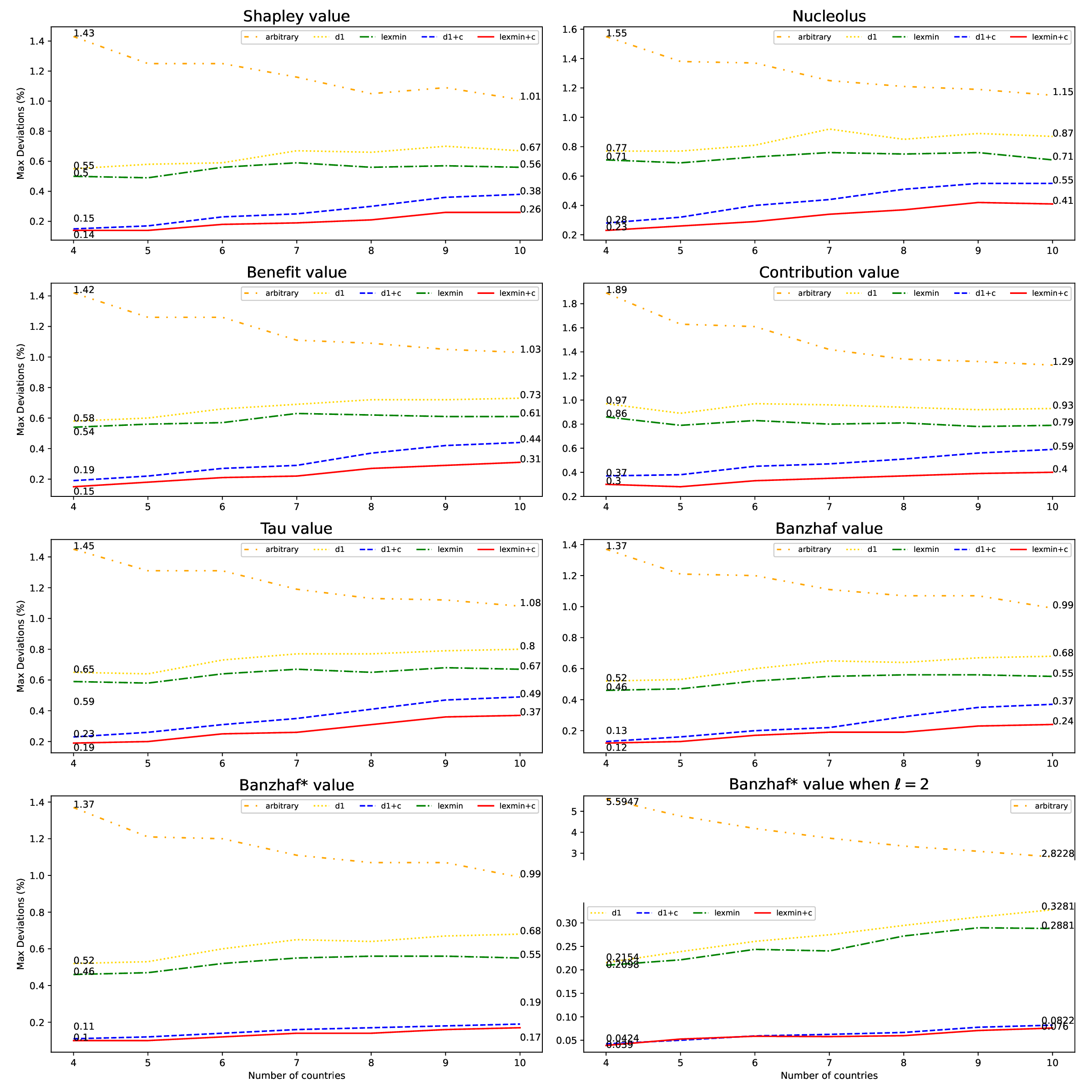}
  \caption{Average maximum relative deviations for each of the seven solution concepts under the five different scenarios for {\bf equal} country sizes, where the number of countries $n$ is ranging from $4$ to $10$. For comparison, the lower right figure displays a result from~\cite{BBPY24} for $\ell=2$, namely for the Banzhaf* value, which behaved best for $\ell=2$.}\label{fig:maxdeviations}
\end{figure}
\begin{figure}
  \centering
    \hspace*{-1cm}
  \includegraphics[scale=0.42]{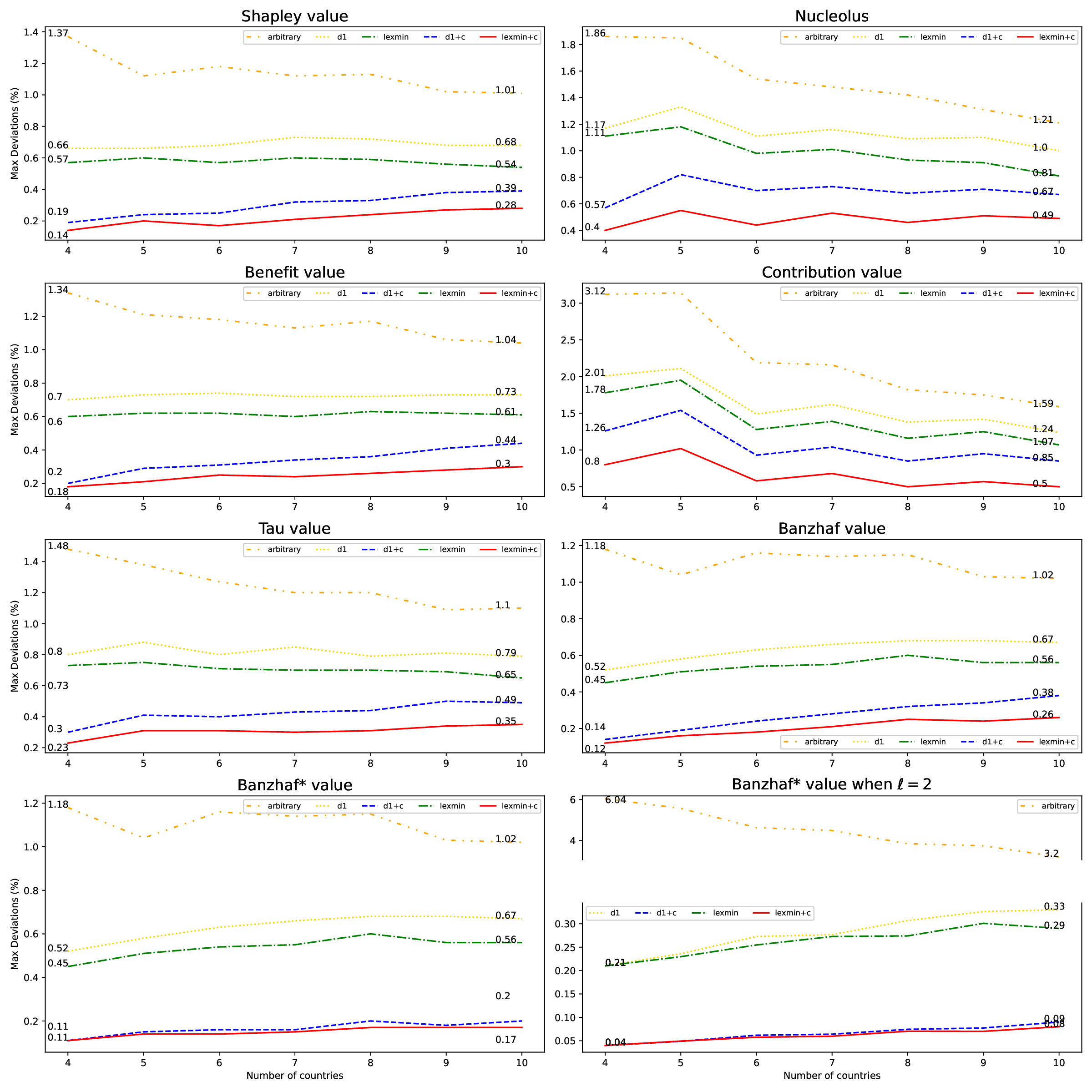}
  \caption{Average maximum relative deviations for each of the seven solution concepts under the five different scenarios for {\bf varying} country sizes, where the number of countries $n$ is ranging from $4$ to $10$. For comparison, the lower right figure displays a result from~\cite{BBPY24} for $\ell=2$, namely for the Banzhaf* value, which behaved best for $\ell=2$.}\label{fig:maxdeviationsundervarying}
\end{figure}
\begin{figure}
  \centering
  \includegraphics[width=1.1\linewidth]{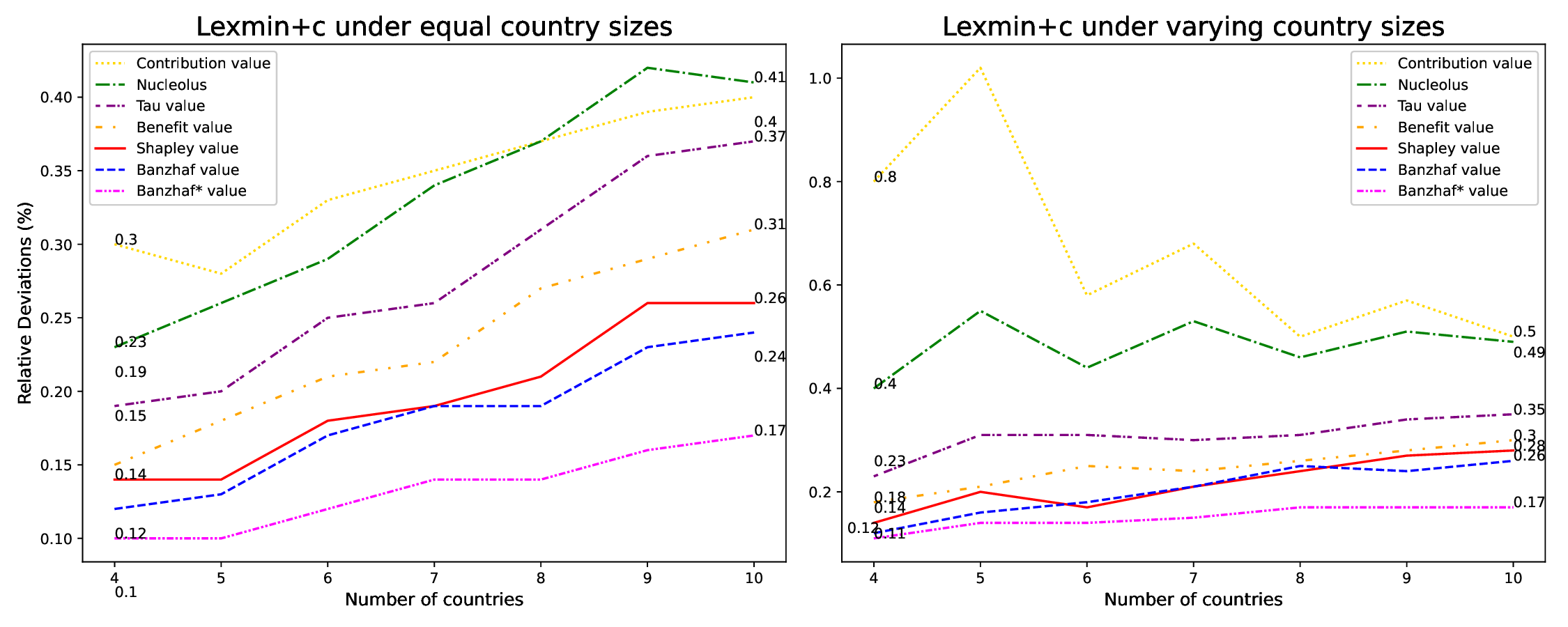}
  \caption{Average maximum relative deviations for all solution concepts in the \emph{lexmin+c} scenario, where the number of countries $n$ ranges from $4$ to $10$.}\label{fig:max_lexmin+c}
\end{figure}
\begin{figure}
  \centering
  \includegraphics[width=1.1\linewidth]{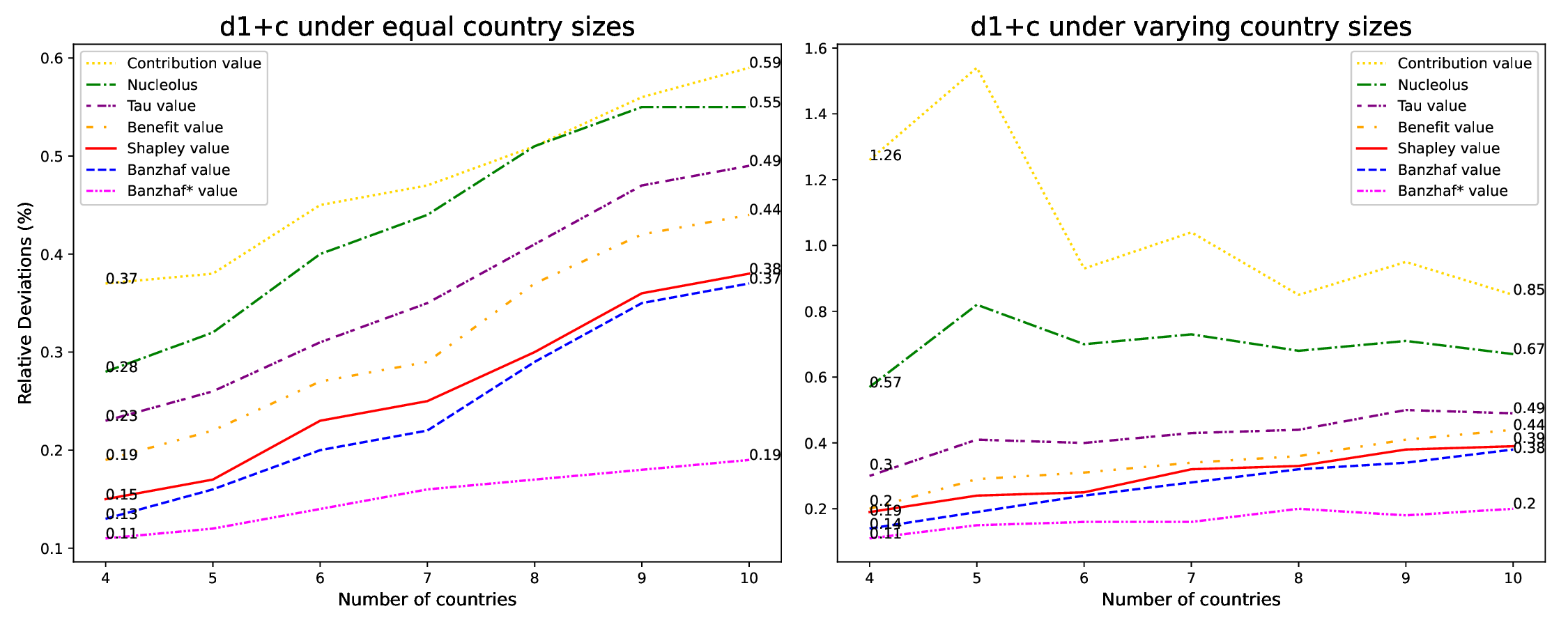}
  \caption{Average maximum relative deviations for all solution concepts in the \emph{d1+c} scenario, where the number of countries $n$ ranges from $4$ to $10$.}\label{fig:max_d1+c}
\end{figure}

\clearpage
\section{Additional Information on Incomplete Simulations}\label{a-in}

{Tables~\ref{table:breakdown_of_not_solved} and~\ref{table:breakdown_of_not_solved2} give complete breakdowns of the number of incomplete simullations instances for equal country sizes and varying country sizes, respectively.
Figures~\ref{fig:optimaltotaldeviations} and~\ref{fig:optimaltotaldeviationsvaryingsizes} are analogous to
Figure~\ref{fig:absolutedeviations}
and~\ref{fig:absolutedeviationsundervarying}, respectively; the only difference is that for 
Figures~\ref{fig:optimaltotaldeviations} and~\ref{fig:optimaltotaldeviationsvaryingsizes} we left out
the incomplete instances.}

\setlength{\tabcolsep}{0.5cm}
\begin{table}[h]
\centering
\begin{tabular}{lllllllll}
\textbf{\#unsolved / n} & \textbf{4} & \textbf{5} & \textbf{6} & \textbf{7} & \textbf{8} & \textbf{9} & \textbf{10} & \textbf{Total} \\ \toprule
\multicolumn{8}{c}{\textbf{Shapley value}}                                                             \\ \midrule
d1         & 0          & 0          & 0          & 0          & 0          & 0          & 0     & 0      \\
d1+c       & 0          & 0          & 0          & 0          & 0          & 0          & 0        & 0   \\
lexmin     & 0          & 0          & 0          & 3          & 2          & 14         & 7         & 26  \\
lexmin+c   & 0          & 0          & 0          & 3          & 5          & 13         & 7        & 28   \\ \midrule
\multicolumn{8}{c}{\textbf{Nucleolus}}                                                                 \\ \midrule
d1         & 0          & 0          & 0          & 0          & 0          & 0          & 0       & 0    \\
d1+c       & 0          & 0          & 0          & 0          & 0          & 0          & 0        & 0   \\
lexmin     & 0          & 0          & 0          & 0          & 0          & 0          & 0         & 0  \\
lexmin+c   & 0          & 0          & 0          & 2          & 0          & 0          & 1         & 3  \\ \midrule
\multicolumn{8}{c}{\textbf{Benefit value}}                                                             \\ \midrule
d1         & 0          & 0          & 0          & 0          & 0          & 0          & 0      & 0     \\
d1+c       & 0          & 0          & 0          & 0          & 0          & 0          & 0        & 0   \\
lexmin     & 0          & 0          & 0          & 0          & 1          & 0          & 0        & 1   \\
lexmin+c   & 1          & 1          & 5          & 1          & 2          & 3          & 0        & 13   \\ \midrule
\multicolumn{8}{c}{\textbf{Contribution value}}                                                        \\ \midrule
d1         & 0          & 0          & 0          & 0          & 0          & 0          & 0     & 0      \\
d1+c       & 0          & 0          & 0          & 0          & 0          & 0          & 0       & 0    \\
lexmin     & 1          & 1          & 2          & 1          & 1          & 2          & 4        & 12   \\
lexmin+c   & 2          & 2          & 5          & 3          & 5          & 0          & 5        & 22   \\ \midrule
\multicolumn{8}{c}{\textbf{Tau value}}                                                                 \\ \midrule
d1         & 0          & 0          & 0          & 0          & 0          & 1          & 0      & 1     \\
d1+c       & 0          & 0          & 0          & 0          & 0          & 1          & 0      & 1     \\
lexmin     & 0          & 0          & 0          & 0          & 0          & 1          & 0       & 1    \\
lexmin+c   & 0          & 1          & 0          & 0          & 0          & 1          & 0       & 2    \\ \midrule
\multicolumn{8}{c}{\textbf{Banzhaf value}}                                                             \\ \midrule
d1         & 0          & 0          & 0          & 0          & 0          & 0          & 0      & 0     \\
d1+c       & 0          & 0          & 0          & 0          & 0          & 0          & 0      & 0     \\
lexmin     & 0          & 1          & 1          & 4          & 1          & 8          & 9      & 24     \\
lexmin+c   & 0          & 2          & 5          & 3          & 1          & 9          & 9      & 29     \\ \midrule
\multicolumn{8}{c}{\textbf{Banzhaf* value}}                                                            \\ \midrule
d1+c       & 0          & 0          & 0          & 0          & 0          & 0          & 0        & 0   \\
lexmin+c   & 0          & 0          & 0          & 0          & 0          & 0          & 1       & 1   \\ \bottomrule
\end{tabular}
\smallskip
\caption{Complete breakdown of the number of incomplete simulation instances (out of 100) for the seven different solution concepts and the four different scenarios (excluding the arbitrary matching scenario, where no ILPs were used) for {\bf equal} country sizes.}\label{table:breakdown_of_not_solved} 
\end{table}

\setlength{\tabcolsep}{0.5cm}
\begin{table}[h]
\centering
\begin{tabular}{lllllllll}
\textbf{\#unfinished / n} & \textbf{4} & \textbf{5} & \textbf{6} & \textbf{7} & \textbf{8} & \textbf{9} & \textbf{10} & \textbf{Total} \\ \toprule
\multicolumn{8}{c}{\textbf{Shapley value}}                                                             \\ \midrule
d1                 & 0 & 0 & 0 & 0 & 0 & 0 & 0 & 0 \\
d1+c               & 0 & 0 & 0 & 0 & 0 & 0 & 0  & 0\\
lexmin             & 0 & 0 & 0 & 1 & 0 & 0 & 6  & 7\\
lexmin+c           & 0 & 0 & 0 & 1 & 0 & 2 & 2 & 5 \\ \midrule
\multicolumn{8}{c}{Nucleolus}                   \\ \midrule
d1                 & 0 & 0 & 0 & 0 & 0 & 1 & 0  & 1\\
d1+c               & 0 & 0 & 0 & 0 & 0 & 0 & 0 & 0 \\
lexmin             & 0 & 0 & 0 & 0 & 0 & 1 & 0  & 1\\
lexmin+c           & 0 & 0 & 0 & 0 & 0 & 1 & 0  & 1\\ \midrule
\multicolumn{8}{c}{Benefit value}              \\ \midrule
d1                 & 0 & 0 & 0 & 0 & 0 & 0 & 0  & 0\\
d1+c               & 0 & 0 & 1 & 0 & 0 & 0 & 0  & 1\\
lexmin             & 0 & 0 & 0 & 0 & 0 & 0 & 0  & 0\\
lexmin+c           & 1 & 0 & 2 & 0 & 0 & 1 & 3  & 7\\ \midrule
\multicolumn{8}{c}{Contribution value}          \\ \midrule
d1                 & 0 & 0 & 0 & 0 & 0 & 0 & 0  & 0\\
d1+c               & 0 & 0 & 0 & 0 & 0 & 0 & 0  & 0\\
lexmin             & 0 & 0 & 1 & 0 & 1 & 1 & 0  & 3\\
lexmin+c           & 0 & 0 & 1 & 0 & 1 & 1 & 1  & 4\\ \midrule
\multicolumn{8}{c}{Tau value}                   \\ \midrule
d1                 & 0 & 0 & 0 & 0 & 0 & 0 & 0  & 0\\
d1+c               & 0 & 0 & 0 & 0 & 1 & 0 & 0  & 1\\
lexmin             & 0 & 0 & 0 & 0 & 0 & 0 & 0  & 0\\
lexmin+c           & 0 & 0 & 0 & 0 & 0 & 0 & 0  & 0\\ \midrule
\multicolumn{8}{c}{Banzhaf value}               \\ \midrule
d1                 & 0 & 1 & 0 & 0 & 0 & 0 & 0  & 1\\
d1+c               & 0 & 1 & 0 & 0 & 0 & 0 & 0  & 1\\
lexmin             & 0 & 1 & 1 & 1 & 2 & 1 & 4  & 10\\
lexmin+c           & 0 & 1 & 0 & 0 & 3 & 1 & 4  & 9\\ \midrule
\multicolumn{8}{c}{Banzhaf* value}              \\ \midrule
d1+c               & 0 & 1 & 0 & 0 & 0 & 0 & 0 & 1 \\
lexmin+c           & 0 & 1 & 0 & 0 & 1 & 1 & 4 & 7\\ \bottomrule
\end{tabular}
\medskip
\caption{Complete breakdown of the number of incomplete simulation instances (out of 100) for the seven different solution concepts and the four different scenarios (excluding the arbitrary matching scenario, where no ILPs were used) for {\bf varying} country sizes.}\label{table:breakdown_of_not_solved2} 
\end{table}

\newpage
\begin{figure}
  \centering
  \hspace*{-1cm}
  \includegraphics[scale=0.42]{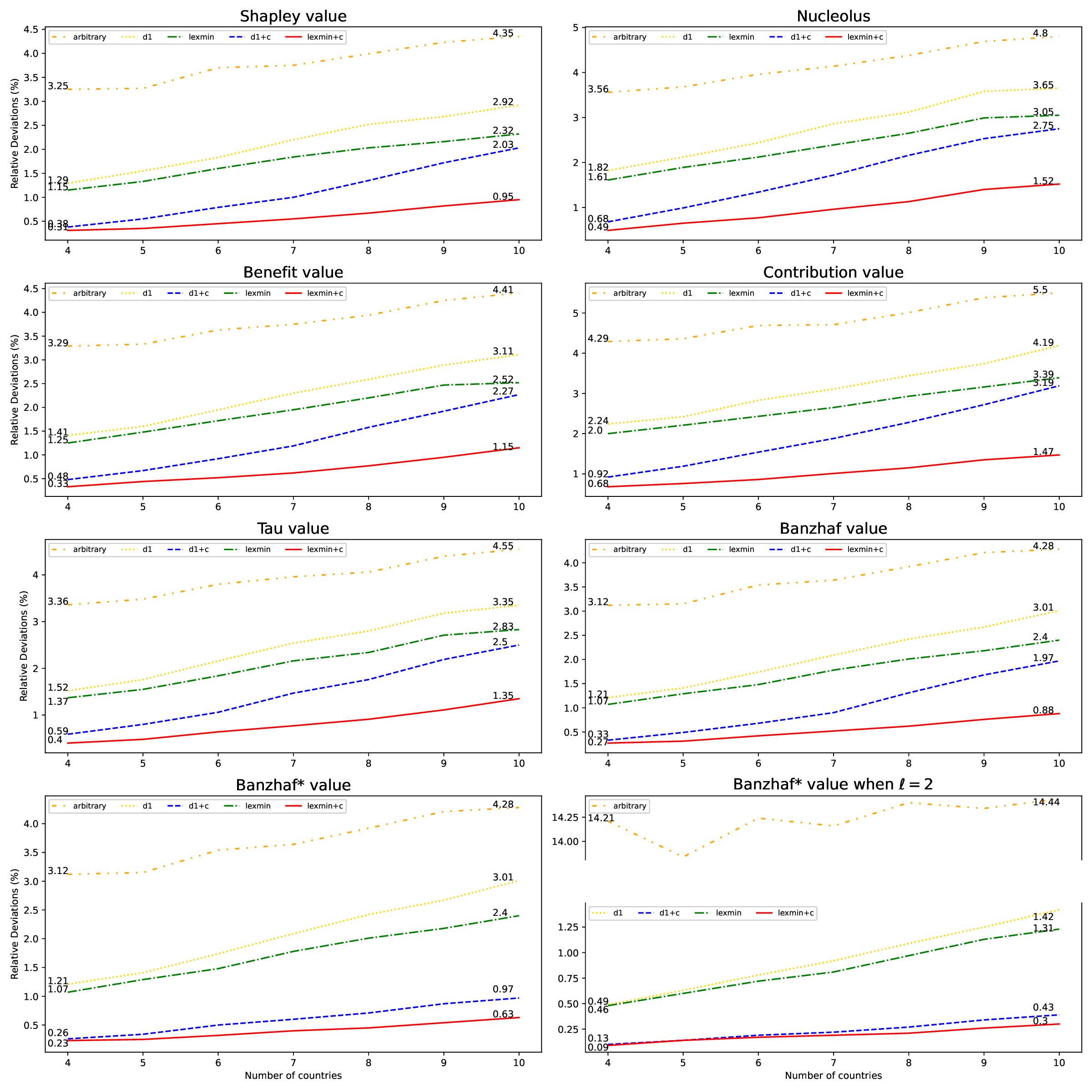}
  \caption{Average total relative deviations, leaving out the incomplete simulation instances under the five different scenarios for {\bf equal} country sizes, where the number of countries ranges from $4$ to $10$. We note that this figure is identical to Figure~\ref{fig:absolutedeviations} that includes all simulation instances.}\label{fig:optimaltotaldeviations}
\end{figure}

\begin{figure}
  \centering
  \hspace*{-1cm}
  \includegraphics[scale=0.42]{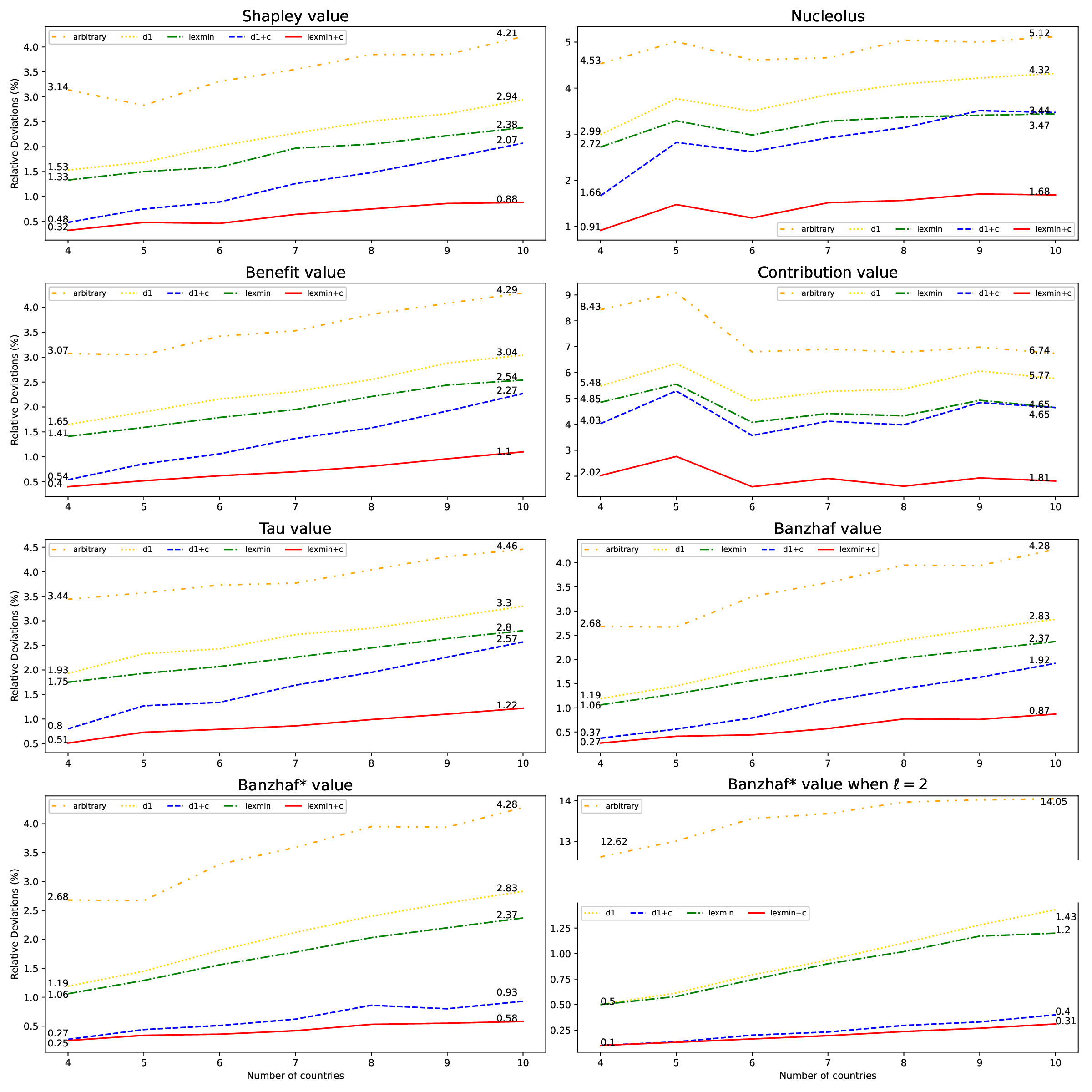}
  \caption{Average total relative deviations, leaving out the incomplete simulation instances under the five different scenarios for {\bf varying} country sizes, where the number of countries ranges from $4$ to $10$. We note that this figure is identical to Figure~\ref{fig:absolutedeviationsundervarying} that includes all simulation instances.}\label{fig:optimaltotaldeviationsvaryingsizes}
\end{figure}

\clearpage
\section{Relative Ratios}\label{a-conc}

{Recall that the relative ratio is 
the ratio of the average maximum relative deviation and the average total relative deviation.
In Table~\ref{table:concentration}
we showed the relative ratios for equal country sizes under \emph{lexmin+c} both for $\ell=\infty$
and $\ell=2$, and we also showed their difference. In Tables~\ref{table:concentration_varying}--\ref{table:concentration_arb_varying} we show the same information for the other scenarios and also for varying country sizes (where the relative ratios for $\ell=2$ where generated  from the data used in~\cite{BBPY24}).}

\setlength{\tabcolsep}{0.25cm}
\begin{table}[h]
\centering
\begin{tabular}{@{}lllllllll@{}}
\toprule
 \multicolumn{2}{c}{n}   & 4    & 5    & 6    & 7    & 8    & 9    & 10   \\ \midrule
\multirow{3}{*}{\minitab[c]{Shapley\\value}}  & $\ell=\infty$ & 44.43\% & 40.58\% & 37.71\% & 33.29\% & 32.18\% & 31.99\% & 31.07\% \\ 
 & $\ell=2$ & 42.82\% & 38.62\% & 35.31\% & 31.05\% & 27.83\% & 26.96\% & 24.81\% \\ 
 & difference & 1.61\% & 1.96\% & 2.40\% & 2.24\% & 4.35\% & 5.03\% & 6.26\% \\ \midrule
\multirow{3}{*}{\minitab[c]{Banzhaf\\value}}  & $\ell=\infty$ & 43.75\% & 39.53\% & 40.34\% & 35.92\% & 32.83\% & 31.31\% & 30.06\% \\ 
 & $\ell=2$ & 43.18\% & 37.15\% & 33.94\% & 31.19\% & 30.13\% & 27.61\% & 24.50\% \\ 
 & difference & 0.57\% & 2.38\% & 6.40\% & 4.73\% & 2.70\% & 3.70\% & 5.56\% \\ \midrule
 \multirow{3}{*}{\minitab[c]{nucleolus}} & $\ell=\infty$ & 43.99\% & 37.67\% & 37.67\% & 35.10\% & 29.67\% & 30.08\% & 29.06\% \\ 
 & $\ell=2$ & 46.06\% & 39.05\% & 35.94\% & 34.13\% & 31.35\% & 30.26\% & 27.10\% \\ 
 & difference & -2.07\% & -1.38\% & 1.73\% & 0.98\% & -1.68\% & -0.18\% & 1.96\% \\ \midrule
  \multirow{3}{*}{\minitab[c]{tau\\value}}  & $\ell=\infty$ & 45.63\% & 42.38\% & 39.61\% & 34.47\% & 31.25\% & 30.81\% & 28.82\% \\ 
 & $\ell=2$ & 41.92\% & 38.50\% & 35.39\% & 36.04\% & 31.29\% & 28.05\% & 26.06\% \\ 
 & difference & 3.72\% & 3.88\% & 4.22\% & -1.57\% & -0.04\% & 2.77\% & 2.76\% \\ \midrule
   \multirow{3}{*}{\minitab[c]{benefit\\value}}  & $\ell=\infty$ & 45.53\% & 40.96\% & 40.33\% & 34.19\% & 32.73\% & 29.63\% & 27.95\% \\ 
 & $\ell=2$ & 42.31\% & 37.14\% & 33.68\% & 31.91\% & 30.41\% & 28.76\% & 25.05\% \\ 
 & difference & 3.22\% & 3.82\% & 6.65\% & 2.28\% & 2.33\% & 0.87\% & 2.91\% \\ \midrule
    \multirow{3}{*}{\minitab[c]{contribution\\value}}  & $\ell=\infty$ & 39.67\% & 36.83\% & 36.54\% & 35.61\% & 31.01\% & 29.51\% & 27.45\% \\ 
 & $\ell=2$ & 43.57\% & 36.74\% & 36.41\% & 32.65\% & 28.38\% & 27.49\% & 25.31\% \\ 
 & difference & -3.91\% & 0.09\% & 0.13\% & 2.96\% & 2.63\% & 2.02\% & 2.15\% \\ \midrule
 \multirow{3}{*}{\minitab[c]{Banzhaf*\\value}}  & $\ell=\infty$ & 44.26\% & 40.10\% & 38.69\% & 36.31\% & 31.90\% & 31.17\% & 29.39\% \\ 
 & $\ell=2$ & 41.60\% & 35.21\% & 31.94\% & 28.43\% & 27.12\% & 22.61\% & 21.53\% \\ 
 & difference & 2.66\% & 4.89\% & 6.75\% & 7.88\% & 4.78\% & 8.56\% & 7.86\% \\
\bottomrule
\end{tabular}
\medskip
\caption{Relative ratios for varying country sizes under \emph{lexmin+c}, for $\ell=\infty$, $\ell = 2$ and their difference.}\label{table:concentration_varying} 
\end{table}

\setlength{\tabcolsep}{0.25cm}
\begin{table}[h]
\centering
\begin{tabular}{@{}lllllllll@{}}
\toprule
 \multicolumn{2}{c}{n}   & 4    & 5    & 6    & 7    & 8    & 9    & 10   \\ \midrule
\multirow{3}{*}{\minitab[c]{Shapley\\value}}  & $\ell=\infty$ & 39.35\% & 30.99\% & 29.02\% & 25.07\% & 22.48\% & 21.05\% & 18.76\% \\ 
 & $\ell=2$ & 43.38\% & 38.19\% & 33.14\% & 30.28\% & 26.48\% & 23.92\% & 22.07\% \\ 
 & difference & -4.03\% & -7.19\% & -4.12\% & -5.21\% & -4.00\% & -2.87\% & -3.31\% \\ \midrule
\multirow{3}{*}{\minitab[c]{Banzhaf\\value}}  & $\ell=\infty$ & 39.74\% & 33.02\% & 29.33\% & 24.52\% & 22.23\% & 20.96\% & 18.82\% \\ 
 & $\ell=2$ & 42.73\% & 35.31\% & 30.25\% & 27.88\% & 25.00\% & 22.51\% & 21.07\% \\ 
 & difference & -2.99\% & -2.29\% & -0.92\% & -3.36\% & -2.77\% & -1.55\% & -2.25\% \\ \midrule
 \multirow{3}{*}{\minitab[c]{nucleolus}} & $\ell=\infty$ & 40.64\% & 32.40\% & 29.62\% & 25.55\% & 23.57\% & 21.75\% & 19.80\% \\ 
 & $\ell=2$ & 42.33\% & 37.23\% & 31.97\% & 29.21\% & 25.00\% & 23.09\% & 20.56\% \\ 
 & difference & -1.69\% & -4.83\% & -2.36\% & -3.65\% & -1.43\% & -1.33\% & -0.76\% \\ \midrule
  \multirow{3}{*}{\minitab[c]{tau\\value}}  & $\ell=\infty$ & 39.97\% & 32.84\% & 29.41\% & 24.10\% & 23.24\% & 21.48\% & 19.39\% \\ 
 & $\ell=2$ & 56.85\% & 46.43\% & 37.42\% & 31.95\% & 29.29\% & 26.97\% & 24.74\% \\ 
 & difference & -16.88\% & -13.59\% & -8.01\% & -7.85\% & -6.05\% & -5.49\% & -5.35\% \\ \midrule
   \multirow{3}{*}{\minitab[c]{benefit\\value}}  & $\ell=\infty$ & 39.59\% & 32.02\% & 29.45\% & 24.35\% & 23.49\% & 21.67\% & 19.32\% \\ 
 & $\ell=2$ & 41.36\% & 38.72\% & 33.50\% & 28.94\% & 26.86\% & 23.02\% & 20.49\% \\ 
 & difference & -1.77\% & -6.71\% & -4.05\% & -4.59\% & -3.37\% & -1.35\% & -1.18\% \\ \midrule
    \multirow{3}{*}{\minitab[c]{contribution\\value}}  & $\ell=\infty$ & 39.79\% & 31.56\% & 29.06\% & 25.00\% & 22.37\% & 20.50\% & 18.45\% \\ 
 & $\ell=2$ & 42.38\% & 38.24\% & 33.41\% & 29.89\% & 26.22\% & 23.84\% & 21.14\% \\ 
 & difference & -2.59\% & -6.68\% & -4.36\% & -4.89\% & -3.85\% & -3.33\% & -2.69\% \\ \midrule
 \multirow{3}{*}{\minitab[c]{Banzhaf*\\value}}  & $\ell=\infty$ & 41.21\% & 34.16\% & 28.94\% & 26.12\% & 23.40\% & 20.99\% & 19.57\% \\ 
 & $\ell=2$ & 42.40\% & 35.79\% & 31.21\% & 28.45\% & 24.70\% & 22.88\% & 21.08\% \\ 
 & difference & -1.19\% & -1.62\% & -2.27\% & -2.33\% & -1.30\% & -1.89\% & -1.51\% \\
\bottomrule
\end{tabular}
\medskip
\caption{Relative ratios for equal country sizes under \emph{d1+c}, for $\ell=\infty$, $\ell = 2$ and their difference.}\label{table:concentration_d1c} 
\end{table}

\setlength{\tabcolsep}{0.25cm}
\begin{table}
\centering
\begin{tabular}{@{}lllllllll@{}}
\toprule
 \multicolumn{2}{c}{n}   & 4    & 5    & 6    & 7    & 8    & 9    & 10   \\ \midrule
\multirow{3}{*}{\minitab[c]{Shapley\\value}}  & $\ell=\infty$ & 39.18\% & 31.66\% & 27.75\% & 25.12\% & 22.51\% & 21.39\% & 18.83\% \\ 
 & $\ell=2$ & 42.00\% & 37.21\% & 31.11\% & 27.88\% & 25.48\% & 22.45\% & 19.78\% \\ 
 & difference & -2.82\% & -5.56\% & -3.35\% & -2.76\% & -2.98\% & -1.05\% & -0.95\% \\ \midrule
\multirow{3}{*}{\minitab[c]{Banzhaf\\value}}  & $\ell=\infty$ & 38.21\% & 33.81\% & 30.03\% & 25.03\% & 22.51\% & 20.96\% & 19.55\% \\ 
 & $\ell=2$ & 41.64\% & 37.85\% & 30.53\% & 27.92\% & 25.37\% & 22.68\% & 21.09\% \\ 
 & difference & -3.42\% & -4.04\% & -0.49\% & -2.89\% & -2.85\% & -1.72\% & -1.54\% \\ \midrule
 \multirow{3}{*}{\minitab[c]{nucleolus}} & $\ell=\infty$ & 34.29\% & 28.92\% & 26.69\% & 25.16\% & 21.51\% & 20.13\% & 19.42\% \\ 
 & $\ell=2$ & 39.79\% & 32.16\% & 26.26\% & 23.07\% & 21.66\% & 20.33\% & 18.24\% \\ 
 & difference & -5.50\% & -3.24\% & 0.43\% & 2.10\% & -0.15\% & -0.20\% & 1.17\% \\ \midrule
  \multirow{3}{*}{\minitab[c]{tau\\value}}  & $\ell=\infty$ & 37.66\% & 32.07\% & 29.60\% & 25.25\% & 22.80\% & 21.94\% & 19.10\% \\ 
 & $\ell=2$ & 46.67\% & 35.53\% & 30.91\% & 29.27\% & 24.75\% & 21.81\% & 20.12\% \\ 
 & difference & -9.01\% & -3.45\% & -1.31\% & -4.02\% & -1.96\% & 0.13\% & -1.02\% \\ \midrule
   \multirow{3}{*}{\minitab[c]{benefit\\value}}  & $\ell=\infty$ & 37.56\% & 33.48\% & 29.28\% & 25.08\% & 22.56\% & 21.38\% & 19.45\% \\ 
 & $\ell=2$ & 41.50\% & 34.41\% & 32.13\% & 26.41\% & 24.49\% & 21.74\% & 19.71\% \\ 
 & difference & -3.94\% & -0.93\% & -2.85\% & -1.32\% & -1.93\% & -0.35\% & -0.26\% \\ \midrule
    \multirow{3}{*}{\minitab[c]{contribution\\value}}  & $\ell=\infty$ & 31.23\% & 29.10\% & 26.19\% & 25.19\% & 21.45\% & 19.59\% & 18.38\% \\ 
 & $\ell=2$ & 38.51\% & 32.26\% & 28.71\% & 25.60\% & 22.88\% & 20.89\% & 19.32\% \\ 
 & difference & -7.27\% & -3.16\% & -2.53\% & -0.42\% & -1.42\% & -1.30\% & -0.94\% \\ \midrule
 \multirow{3}{*}{\minitab[c]{Banzhaf*\\value}}  & $\ell=\infty$ & 38.21\% & 33.81\% & 30.03\% & 25.03\% & 22.51\% & 20.96\% & 19.55\% \\ 
 & $\ell=2$ & 41.64\% & 37.85\% & 30.53\% & 27.92\% & 25.37\% & 22.68\% & 21.09\% \\ 
 & difference & -3.42\% & -4.04\% & -0.49\% & -2.89\% & -2.85\% & -1.72\% & -1.54\% \\
\bottomrule
\end{tabular}
\medskip
\caption{Relative ratios for varying country sizes under \emph{d1+c}, for $\ell=\infty$, $\ell = 2$ and their difference.}\label{table:concentration_d1c_varying} 
\end{table}

\setlength{\tabcolsep}{0.25cm}
\begin{table}[h]
\centering
\begin{tabular}{@{}lllllllll@{}}
\toprule
 \multicolumn{2}{c}{n}   & 4    & 5    & 6    & 7    & 8    & 9    & 10   \\ \midrule
\multirow{3}{*}{\minitab[c]{Shapley\\value}}  & $\ell=\infty$ & 43.66\% & 36.50\% & 35.01\% & 31.87\% & 27.73\% & 25.86\% & 23.59\% \\ 
 & $\ell=2$ & 43.02\% & 37.65\% & 34.21\% & 30.58\% & 27.36\% & 25.96\% & 23.11\% \\ 
 & difference & 0.64\% & -1.15\% & 0.80\% & 1.29\% & 0.37\% & -0.10\% & 0.48\% \\ \midrule
\multirow{3}{*}{\minitab[c]{Banzhaf\\value}}  & $\ell=\infty$ & 43.04\% & 36.34\% & 35.15\% & 31.14\% & 27.69\% & 25.74\% & 22.87\% \\ 
 & $\ell=2$ & 43.71\% & 36.90\% & 33.85\% & 29.67\% & 28.06\% & 25.64\% & 23.42\% \\ 
 & difference & -0.66\% & -0.56\% & 1.31\% & 1.48\% & -0.37\% & 0.10\% & -0.56\% \\ \midrule
 \multirow{3}{*}{\minitab[c]{nucleolus}} & $\ell=\infty$ & 43.97\% & 36.83\% & 34.53\% & 31.64\% & 28.27\% & 25.31\% & 23.27\% \\ 
 & $\ell=2$ & 44.57\% & 36.52\% & 34.48\% & 29.52\% & 27.85\% & 26.20\% & 24.17\% \\ 
 & difference & -0.59\% & 0.31\% & 0.06\% & 2.12\% & 0.41\% & -0.89\% & -0.90\% \\ \midrule
  \multirow{3}{*}{\minitab[c]{tau\\value}}  & $\ell=\infty$ & 43.55\% & 37.36\% & 34.61\% & 31.16\% & 27.79\% & 24.95\% & 23.68\% \\ 
 & $\ell=2$ & 44.00\% & 39.67\% & 33.87\% & 30.85\% & 26.95\% & 25.18\% & 23.49\% \\ 
 & difference & -0.45\% & -2.31\% & 0.73\% & 0.31\% & 0.84\% & -0.23\% & 0.18\% \\ \midrule
   \multirow{3}{*}{\minitab[c]{benefit\\value}}  & $\ell=\infty$ & 42.98\% & 37.59\% & 33.34\% & 32.33\% & 28.17\% & 24.73\% & 24.22\% \\ 
 & $\ell=2$ & 42.98\% & 39.41\% & 32.90\% & 30.28\% & 26.75\% & 25.73\% & 23.74\% \\ 
 & difference & 0.00\% & -1.83\% & 0.44\% & 2.05\% & 1.42\% & -1.00\% & 0.47\% \\ \midrule
    \multirow{3}{*}{\minitab[c]{contribution\\value}}  & $\ell=\infty$ & 42.92\% & 35.82\% & 34.11\% & 30.26\% & 27.68\% & 24.79\% & 23.13\% \\ 
 & $\ell=2$ & 44.52\% & 37.22\% & 33.77\% & 31.10\% & 28.12\% & 26.39\% & 22.57\% \\ 
 & difference & -1.60\% & -1.40\% & 0.33\% & -0.84\% & -0.44\% & -1.60\% & 0.56\% \\ \midrule
\multirow{3}{*}{\minitab[c]{Banzhaf*\\value}}  & $\ell=\infty$ & 43.04\% & 36.34\% & 35.15\% & 31.14\% & 27.69\% & 25.74\% & 22.87\% \\ 
 & $\ell=2$ & 43.71\% & 36.90\% & 33.85\% & 29.67\% & 28.06\% & 25.64\% & 23.42\% \\ 
 & difference & -0.66\% & -0.56\% & 1.31\% & 1.48\% & -0.37\% & 0.10\% & -0.56\% \\
\bottomrule
\end{tabular}
\medskip
\caption{Relative ratios for equal country sizes under \emph{lexmin}, for $\ell=\infty$, $\ell = 2$ and their difference.}\label{table:concentration_lexmin} 
\end{table}

\setlength{\tabcolsep}{0.25cm}
\begin{table}[h]
\centering
\begin{tabular}{@{}lllllllll@{}}
\toprule
 \multicolumn{2}{c}{n}   & 4    & 5    & 6    & 7    & 8    & 9    & 10   \\ \midrule
\multirow{3}{*}{\minitab[c]{Shapley\\value}}  & $\ell=\infty$ & 42.41\% & 39.85\% & 35.90\% & 30.36\% & 28.87\% & 25.47\% & 22.75\% \\ 
 & $\ell=2$ & 45.17\% & 37.62\% & 33.77\% & 29.46\% & 27.20\% & 25.13\% & 23.66\% \\ 
 & difference & -2.76\% & 2.23\% & 2.14\% & 0.90\% & 1.67\% & 0.34\% & -0.90\% \\ \midrule
\multirow{3}{*}{\minitab[c]{Banzhaf\\value}}  & $\ell=\infty$ & 42.44\% & 39.50\% & 34.98\% & 30.66\% & 29.63\% & 25.36\% & 23.76\% \\ 
 & $\ell=2$ & 42.08\% & 39.60\% & 34.42\% & 30.31\% & 26.87\% & 25.73\% & 24.03\% \\ 
 & difference & 0.36\% & -0.11\% & 0.56\% & 0.35\% & 2.75\% & -0.37\% & -0.26\% \\ \midrule
 \multirow{3}{*}{\minitab[c]{nucleolus}} & $\ell=\infty$ & 40.74\% & 35.71\% & 33.04\% & 30.84\% & 27.59\% & 26.51\% & 23.65\% \\ 
 & $\ell=2$ & 42.53\% & 38.17\% & 34.08\% & 29.74\% & 28.51\% & 26.64\% & 25.27\% \\ 
 & difference & -1.80\% & -2.46\% & -1.05\% & 1.10\% & -0.92\% & -0.14\% & -1.62\% \\ \midrule
  \multirow{3}{*}{\minitab[c]{tau\\value}}  & $\ell=\infty$ & 41.81\% & 38.72\% & 34.16\% & 30.77\% & 28.55\% & 26.31\% & 23.12\% \\ 
 & $\ell=2$ & 45.29\% & 37.62\% & 34.39\% & 29.38\% & 27.49\% & 25.13\% & 23.88\% \\ 
 & difference & -3.48\% & 1.11\% & -0.22\% & 1.38\% & 1.06\% & 1.18\% & -0.76\% \\ \midrule
   \multirow{3}{*}{\minitab[c]{benefit\\value}}  & $\ell=\infty$ & 42.52\% & 39.28\% & 34.44\% & 30.99\% & 28.39\% & 25.26\% & 23.82\% \\ 
 & $\ell=2$ & 43.70\% & 37.10\% & 33.92\% & 30.51\% & 27.40\% & 25.11\% & 23.98\% \\ 
 & difference & -1.18\% & 2.18\% & 0.51\% & 0.47\% & 0.99\% & 0.15\% & -0.16\% \\ \midrule
    \multirow{3}{*}{\minitab[c]{contribution\\value}}  & $\ell=\infty$ & 36.57\% & 35.10\% & 31.46\% & 31.39\% & 26.63\% & 25.40\% & 23.02\% \\ 
 & $\ell=2$ & 37.19\% & 33.40\% & 32.36\% & 30.46\% & 28.79\% & 27.69\% & 24.80\% \\ 
 & difference & -0.63\% & 1.70\% & -0.89\% & 0.93\% & -2.16\% & -2.29\% & -1.78\% \\ \midrule
\multirow{3}{*}{\minitab[c]{Banzhaf*\\value}}  & $\ell=\infty$ & 42.44\% & 39.50\% & 34.98\% & 30.66\% & 29.63\% & 25.36\% & 23.76\% \\ 
 & $\ell=2$ & 42.08\% & 39.60\% & 34.42\% & 30.31\% & 26.87\% & 25.73\% & 24.03\% \\ 
 & difference & 0.36\% & -0.11\% & 0.56\% & 0.35\% & 2.75\% & -0.37\% & -0.26\% \\
\bottomrule
\end{tabular}
\medskip
\caption{Relative ratios for varying country sizes under \emph{lexmin}, for $\ell=\infty$, $\ell = 2$ and their difference.}\label{table:concentration_lexmin_varying} 
\end{table}

\setlength{\tabcolsep}{0.25cm}
\begin{table}[h]
\centering
\begin{tabular}{@{}lllllllll@{}}
\toprule
 \multicolumn{2}{c}{n}   & 4    & 5    & 6    & 7    & 8    & 9    & 10   \\ \midrule
\multirow{3}{*}{\minitab[c]{Shapley\\value}}  & $\ell=\infty$ & 42.66\% & 37.54\% & 32.40\% & 30.34\% & 26.29\% & 26.12\% & 22.79\% \\ 
 & $\ell=2$ & 42.57\% & 37.54\% & 34.20\% & 30.24\% & 27.99\% & 24.22\% & 22.87\% \\ 
 & difference & 0.09\% & 0.00\% & -1.81\% & 0.09\% & -1.70\% & 1.90\% & -0.08\% \\ \midrule
\multirow{3}{*}{\minitab[c]{Banzhaf\\value}}  & $\ell=\infty$ & 43.36\% & 37.84\% & 34.55\% & 31.11\% & 26.35\% & 24.97\% & 22.70\% \\ 
 & $\ell=2$ & 43.96\% & 37.92\% & 33.42\% & 29.87\% & 27.04\% & 25.00\% & 23.11\% \\ 
 & difference & -0.60\% & -0.08\% & 1.13\% & 1.24\% & -0.69\% & -0.03\% & -0.40\% \\ \midrule
 \multirow{3}{*}{\minitab[c]{nucleolus}} & $\ell=\infty$ & 42.18\% & 36.51\% & 33.15\% & 31.99\% & 27.25\% & 24.81\% & 23.85\% \\ 
 & $\ell=2$ & 45.00\% & 37.98\% & 32.49\% & 29.57\% & 27.37\% & 25.42\% & 22.73\% \\ 
 & difference & -2.82\% & -1.47\% & 0.66\% & 2.42\% & -0.13\% & -0.61\% & 1.12\% \\ \midrule
  \multirow{3}{*}{\minitab[c]{tau\\value}}  & $\ell=\infty$ & 42.54\% & 36.46\% & 33.55\% & 30.46\% & 27.41\% & 24.95\% & 23.71\% \\ 
 & $\ell=2$ & 42.68\% & 37.32\% & 33.21\% & 29.89\% & 28.12\% & 25.07\% & 23.42\% \\ 
 & difference & -0.13\% & -0.86\% & 0.34\% & 0.57\% & -0.71\% & -0.12\% & 0.28\% \\ \midrule
   \multirow{3}{*}{\minitab[c]{benefit\\value}}  & $\ell=\infty$ & 41.25\% & 37.80\% & 33.96\% & 30.01\% & 27.99\% & 25.01\% & 23.54\% \\ 
 & $\ell=2$ & 43.67\% & 37.81\% & 33.68\% & 29.98\% & 27.14\% & 24.64\% & 23.92\% \\ 
 & difference & -2.42\% & -0.01\% & 0.28\% & 0.02\% & 0.85\% & 0.36\% & -0.39\% \\ \midrule
    \multirow{3}{*}{\minitab[c]{contribution\\value}}  & $\ell=\infty$ & 43.32\% & 36.90\% & 34.20\% & 30.21\% & 27.44\% & 24.70\% & 22.31\% \\ 
 & $\ell=2$ & 42.91\% & 37.76\% & 32.70\% & 30.33\% & 27.63\% & 25.55\% & 23.26\% \\ 
 & difference & 0.40\% & -0.86\% & 1.49\% & -0.11\% & -0.19\% & -0.85\% & -0.96\% \\ \midrule
\multirow{3}{*}{\minitab[c]{Banzhaf*\\value}}  & $\ell=\infty$ & 43.36\% & 37.84\% & 34.55\% & 31.11\% & 26.35\% & 24.97\% & 22.70\% \\ 
 & $\ell=2$ & 43.96\% & 37.92\% & 33.42\% & 29.87\% & 27.04\% & 25.00\% & 23.11\% \\ 
 & difference & -0.60\% & -0.08\% & 1.13\% & 1.24\% & -0.69\% & -0.03\% & -0.40\% \\
\bottomrule
\end{tabular}
\medskip
\caption{Relative ratios for equal country sizes under \emph{d1}, for $\ell=\infty$, $\ell = 2$ and their difference.}\label{table:concentration_d1} 
\end{table}

\setlength{\tabcolsep}{0.25cm}
\begin{table}[h]
\centering
\begin{tabular}{@{}lllllllll@{}}
\toprule
 \multicolumn{2}{c}{n}   & 4    & 5    & 6    & 7    & 8    & 9    & 10   \\ \midrule
\multirow{3}{*}{\minitab[c]{Shapley\\value}}  & $\ell=\infty$ & 42.74\% & 38.66\% & 33.65\% & 31.96\% & 28.69\% & 25.47\% & 23.17\% \\ 
 & $\ell=2$ & 44.61\% & 38.56\% & 33.07\% & 29.70\% & 27.57\% & 24.85\% & 22.91\% \\ 
 & difference & -1.87\% & 0.10\% & 0.58\% & 2.26\% & 1.13\% & 0.62\% & 0.26\% \\ \midrule
\multirow{3}{*}{\minitab[c]{Banzhaf\\value}}  & $\ell=\infty$ & 44.07\% & 39.81\% & 34.52\% & 31.04\% & 28.49\% & 25.85\% & 23.58\% \\ 
 & $\ell=2$ & 42.96\% & 38.70\% & 34.52\% & 29.45\% & 27.90\% & 25.48\% & 23.01\% \\ 
 & difference & 1.11\% & 1.11\% & 0.00\% & 1.60\% & 0.59\% & 0.38\% & 0.57\% \\ \midrule
 \multirow{3}{*}{\minitab[c]{nucleolus}} & $\ell=\infty$ & 38.99\% & 35.17\% & 31.74\% & 30.03\% & 26.66\% & 25.90\% & 23.07\% \\ 
 & $\ell=2$ & 41.91\% & 36.93\% & 33.97\% & 29.82\% & 27.56\% & 25.70\% & 23.59\% \\ 
 & difference & -2.93\% & -1.75\% & -2.23\% & 0.21\% & -0.89\% & 0.20\% & -0.52\% \\ \midrule
  \multirow{3}{*}{\minitab[c]{tau\\value}}  & $\ell=\infty$ & 41.58\% & 37.88\% & 33.09\% & 31.44\% & 27.77\% & 26.33\% & 23.96\% \\ 
 & $\ell=2$ & 48.20\% & 38.31\% & 32.66\% & 30.14\% & 26.17\% & 25.75\% & 23.44\% \\ 
 & difference & -6.62\% & -0.43\% & 0.43\% & 1.30\% & 1.60\% & 0.59\% & 0.52\% \\ \midrule
   \multirow{3}{*}{\minitab[c]{benefit\\value}}  & $\ell=\infty$ & 42.62\% & 38.50\% & 34.37\% & 31.24\% & 28.34\% & 25.48\% & 24.08\% \\ 
 & $\ell=2$ & 41.88\% & 36.39\% & 34.85\% & 30.41\% & 28.05\% & 25.58\% & 22.31\% \\ 
 & difference & 0.75\% & 2.11\% & -0.48\% & 0.84\% & 0.29\% & -0.10\% & 1.76\% \\ \midrule
    \multirow{3}{*}{\minitab[c]{contribution\\value}}  & $\ell=\infty$ & 36.67\% & 33.28\% & 30.37\% & 30.77\% & 25.79\% & 23.44\% & 21.56\% \\ 
 & $\ell=2$ & 38.40\% & 33.79\% & 31.53\% & 30.64\% & 27.91\% & 25.48\% & 24.00\% \\ 
 & difference & -1.73\% & -0.51\% & -1.16\% & 0.13\% & -2.12\% & -2.04\% & -2.44\% \\ \midrule
\multirow{3}{*}{\minitab[c]{Banzhaf*\\value}}  & $\ell=\infty$ & 44.07\% & 39.81\% & 34.52\% & 31.04\% & 28.49\% & 25.85\% & 23.58\% \\ 
 & $\ell=2$ & 42.96\% & 38.70\% & 34.52\% & 29.45\% & 27.90\% & 25.48\% & 23.01\% \\ 
 & difference & 1.11\% & 1.11\% & 0.00\% & 1.60\% & 0.59\% & 0.38\% & 0.57\% \\
\bottomrule
\end{tabular}
\medskip
\caption{Relative ratios for varying country sizes under \emph{d1}, for $\ell=\infty$, $\ell = 2$ and their difference.}\label{table:concentration_d1_varying} 
\end{table}

\setlength{\tabcolsep}{0.25cm}
\begin{table}[h]
\centering
\begin{tabular}{@{}lllllllll@{}}
\toprule
 \multicolumn{2}{c}{n}   & 4    & 5    & 6    & 7    & 8    & 9    & 10   \\ \midrule
\multirow{3}{*}{\minitab[c]{Shapley\\value}}  & $\ell=\infty$ & 43.98\% & 38.25\% & 33.81\% & 30.90\% & 26.32\% & 25.66\% & 23.25\% \\ 
 & $\ell=2$ & 39.36\% & 34.50\% & 29.35\% & 26.18\% & 23.18\% & 21.57\% & 19.49\% \\ 
 & difference & 4.62\% & 3.75\% & 4.46\% & 4.72\% & 3.14\% & 4.09\% & 3.75\% \\ \midrule
\multirow{3}{*}{\minitab[c]{Banzhaf\\value}}  & $\ell=\infty$ & 44.06\% & 38.36\% & 33.94\% & 30.56\% & 27.16\% & 25.49\% & 23.25\% \\ 
 & $\ell=2$ & 39.37\% & 34.46\% & 29.35\% & 26.26\% & 23.19\% & 21.57\% & 19.55\% \\ 
 & difference & 4.68\% & 3.90\% & 4.59\% & 4.31\% & 3.97\% & 3.92\% & 3.70\% \\ \midrule
 \multirow{3}{*}{\minitab[c]{nucleolus}} & $\ell=\infty$ & 43.55\% & 37.44\% & 34.60\% & 30.29\% & 27.71\% & 25.34\% & 24.00\% \\ 
 & $\ell=2$ & 39.83\% & 35.14\% & 29.60\% & 26.69\% & 23.49\% & 22.09\% & 20.07\% \\ 
 & difference & 3.71\% & 2.29\% & 5.00\% & 3.59\% & 4.22\% & 3.26\% & 3.93\% \\ \midrule
  \multirow{3}{*}{\minitab[c]{tau\\value}}  & $\ell=\infty$ & 43.20\% & 37.75\% & 34.42\% & 29.98\% & 27.80\% & 25.45\% & 23.76\% \\ 
 & $\ell=2$ & 39.50\% & 34.80\% & 29.60\% & 26.25\% & 23.27\% & 21.74\% & 19.72\% \\ 
 & difference & 3.70\% & 2.95\% & 4.82\% & 3.73\% & 4.53\% & 3.71\% & 4.04\% \\ \midrule
   \multirow{3}{*}{\minitab[c]{benefit\\value}}  & $\ell=\infty$ & 43.12\% & 37.85\% & 34.67\% & 29.64\% & 27.81\% & 24.61\% & 23.38\% \\ 
 & $\ell=2$ & 39.39\% & 34.84\% & 29.48\% & 26.04\% & 23.25\% & 21.76\% & 19.59\% \\ 
 & difference & 3.73\% & 3.01\% & 5.19\% & 3.60\% & 4.57\% & 2.85\% & 3.78\% \\ \midrule
    \multirow{3}{*}{\minitab[c]{contribution\\value}}  & $\ell=\infty$ & 44.02\% & 37.40\% & 34.34\% & 30.24\% & 26.69\% & 24.59\% & 23.38\% \\ 
 & $\ell=2$ & 39.56\% & 34.43\% & 29.99\% & 26.29\% & 23.49\% & 22.17\% & 20.12\% \\ 
 & difference & 4.46\% & 2.97\% & 4.35\% & 3.95\% & 3.21\% & 2.43\% & 3.27\% \\ \midrule
\multirow{3}{*}{\minitab[c]{Banzhaf*\\value}}  & $\ell=\infty$ & 44.06\% & 38.36\% & 33.94\% & 30.56\% & 27.16\% & 25.49\% & 23.25\% \\ 
 & $\ell=2$ & 39.37\% & 34.46\% & 29.35\% & 26.26\% & 23.19\% & 21.57\% & 19.55\% \\ 
 & difference & 4.68\% & 3.90\% & 4.59\% & 4.31\% & 3.97\% & 3.92\% & 3.70\% \\
\bottomrule
\end{tabular}
\medskip
\caption{Relative ratios for equal country sizes under \emph{arbitrary}, for $\ell=\infty$, $\ell = 2$ and their difference.}\label{table:concentration_arb} 
\end{table}

\setlength{\tabcolsep}{0.25cm}
\begin{table}[h]
\centering
\begin{tabular}{@{}lllllllll@{}}
\toprule
 \multicolumn{2}{c}{n}   & 4    & 5    & 6    & 7    & 8    & 9    & 10   \\ \midrule
\multirow{3}{*}{\minitab[c]{Shapley\\value}}  & $\ell=\infty$ & 43.61\% & 39.47\% & 35.60\% & 31.60\% & 29.39\% & 26.50\% & 23.95\% \\ 
 & $\ell=2$ & 48.45\% & 44.14\% & 34.83\% & 33.25\% & 28.02\% & 27.37\% & 23.15\% \\ 
 & difference & -4.84\% & -4.67\% & 0.77\% & -1.65\% & 1.38\% & -0.87\% & 0.79\% \\ \midrule
\multirow{3}{*}{\minitab[c]{Banzhaf\\value}}  & $\ell=\infty$ & 43.90\% & 39.04\% & 35.01\% & 31.64\% & 29.10\% & 26.15\% & 23.91\% \\ 
 & $\ell=2$ & 47.85\% & 42.82\% & 34.13\% & 32.82\% & 27.54\% & 26.71\% & 22.78\% \\ 
 & difference & -3.96\% & -3.78\% & 0.88\% & -1.18\% & 1.56\% & -0.56\% & 1.13\% \\ \midrule
 \multirow{3}{*}{\minitab[c]{nucleolus}} & $\ell=\infty$ & 41.05\% & 36.85\% & 33.31\% & 31.74\% & 28.16\% & 26.13\% & 23.71\% \\ 
 & $\ell=2$ & 49.57\% & 47.31\% & 37.30\% & 35.07\% & 29.78\% & 29.64\% & 24.81\% \\ 
 & difference & -8.52\% & -10.46\% & -4.00\% & -3.33\% & -1.62\% & -3.50\% & -1.10\% \\ \midrule
  \multirow{3}{*}{\minitab[c]{tau\\value}}  & $\ell=\infty$ & 43.03\% & 38.65\% & 34.12\% & 31.94\% & 29.69\% & 25.28\% & 24.72\% \\ 
 & $\ell=2$ & 54.97\% & 45.25\% & 35.75\% & 33.56\% & 28.53\% & 28.18\% & 23.67\% \\ 
 & difference & -11.94\% & -6.59\% & -1.63\% & -1.62\% & 1.17\% & -2.90\% & 1.04\% \\ \midrule
   \multirow{3}{*}{\minitab[c]{benefit\\value}}  & $\ell=\infty$ & 43.53\% & 39.76\% & 34.39\% & 32.04\% & 30.22\% & 25.84\% & 24.33\% \\ 
 & $\ell=2$ & 48.63\% & 45.02\% & 35.56\% & 33.39\% & 28.49\% & 27.87\% & 23.42\% \\ 
 & difference & -5.10\% & -5.25\% & -1.16\% & -1.35\% & 1.73\% & -2.03\% & 0.91\% \\ \midrule
    \multirow{3}{*}{\minitab[c]{contribution\\value}}  & $\ell=\infty$ & 37.03\% & 34.58\% & 32.17\% & 31.31\% & 26.83\% & 25.12\% & 23.61\% \\ 
 & $\ell=2$ & 48.74\% & 47.81\% & 37.86\% & 36.47\% & 30.58\% & 31.08\% & 25.45\% \\ 
 & difference & -11.71\% & -13.24\% & -5.69\% & -5.16\% & -3.74\% & -5.96\% & -1.84\% \\ \midrule
\multirow{3}{*}{\minitab[c]{Banzhaf*\\value}}  & $\ell=\infty$ & 43.90\% & 39.04\% & 35.01\% & 31.64\% & 29.10\% & 26.15\% & 23.91\% \\ 
 & $\ell=2$ & 47.85\% & 42.82\% & 34.13\% & 32.82\% & 27.54\% & 26.71\% & 22.78\% \\ 
 & difference & -3.96\% & -3.78\% & 0.88\% & -1.18\% & 1.56\% & -0.56\% & 1.13\% \\
\bottomrule
\end{tabular}
\medskip
\caption{Relative ratios for varying country sizes under \emph{arbitrary}, for $\ell=\infty$, $\ell = 2$ and their difference.}\label{table:concentration_arb_varying} 
\end{table}

\clearpage
\section{Full Results on Number of Kidney Transplants}\label{a-average2}

{Tables~\ref{t-num=10} and \ref{t-varying-num=10} only showed results for the average number of kidney transplants for $n=10$. We included the information from these two tables for $\ell=\infty$ in Tables~\ref{t-num-equal}
and~\ref{t-num-varying}, together with the results for $n=4,\ldots,9$, for the case where $\ell=\infty$; we refer to~\cite{BBPY24} for the results for $n=4,\ldots,9$ for the case where $\ell = 2$.}

\setlength{\tabcolsep}{0.20cm}
\begin{table}[h]
\begin{tabular}{cllllllll}
\toprule
\multicolumn{2}{c}{Solution   concepts/n}        & 4       & 5       & 6       & 7       & 8       & 9       & 10      \\
\midrule
\multirow{5}{*}{benefit value}      & \emph{arbitrary} & 1781.89 & 1780.60 & 1775.74 & 1771.30 & 1767.59 & 1763.44 & 1781.58 \\
                                    & \emph{d1}        & 1782.15 & 1781.01 & 1775.31 & 1771.21 & 1767.38 & 1763.04 & 1781.65 \\
                                    & \emph{d1+c}      & 1781.52 & 1780.01 & 1774.51 & 1770.22 & 1766.34 & 1762.37 & 1781.46 \\
                                    & \emph{lexmin}    & 1782.25 & 1781.25 & 1775.47 & 1771.56 & 1767.87 & 1763.19 & 1782.58 \\
                                    & \emph{lexmin+c}  & 1781.69 & 1780.12 & 1774.99 & 1770.71 & 1766.30 & 1761.76 & 1780.31 \\
                                    \midrule
\multirow{5}{*}{contribution value} & \emph{arbitrary} & 1781.89 & 1780.60 & 1775.74 & 1771.30 & 1767.59 & 1763.44 & 1781.58 \\
                                    & \emph{d1}        & 1782.70 & 1781.30 & 1775.88 & 1772.07 & 1768.31 & 1763.31 & 1781.52 \\
                                    & \emph{d1+c}      & 1781.43 & 1780.41 & 1775.43 & 1770.90 & 1766.73 & 1762.68 & 1780.50 \\
                                    & \emph{lexmin}    & 1782.20 & 1781.42 & 1775.86 & 1771.93 & 1768.04 & 1763.31 & 1782.52 \\
                                    & \emph{lexmin+c}  & 1781.80 & 1780.58 & 1774.87 & 1770.21 & 1766.24 & 1762.30 & 1780.05 \\
                                    \midrule
\multirow{5}{*}{Nucleolus}          & \emph{arbitrary} & 1781.89 & 1780.60 & 1775.67 & 1771.43 & 1767.64 & 1763.17 & 1781.58 \\
                                    & \emph{d1}        & 1781.79 & 1780.99 & 1775.76 & 1771.53 & 1767.26 & 1762.97 & 1781.80 \\
                                    & \emph{d1+c}      & 1781.24 & 1779.79 & 1775.30 & 1770.68 & 1767.00 & 1762.30 & 1780.45 \\
                                    & \emph{lexmin}    & 1782.02 & 1781.11 & 1775.30 & 1771.99 & 1767.42 & 1763.19 & 1781.98 \\
                                    & \emph{lexmin+c}  & 1781.43 & 1779.96 & 1774.86 & 1770.74 & 1765.81 & 1761.97 & 1780.12 \\
                                    \midrule
\multirow{5}{*}{Shapley value}      & \emph{arbitrary} & 1781.89 & 1780.60 & 1775.78 & 1771.38 & 1767.57 & 1763.14 & 1781.58 \\
                                    & \emph{d1}        & 1782.25 & 1781.22 & 1775.81 & 1771.75 & 1767.06 & 1763.34 & 1781.54 \\
                                    & \emph{d1+c}      & 1781.67 & 1780.88 & 1775.45 & 1771.11 & 1766.68 & 1762.22 & 1780.92 \\
                                    & \emph{lexmin}    & 1782.43 & 1781.90 & 1775.38 & 1771.73 & 1767.90 & 1763.37 & 1782.02 \\
                                    & \emph{lexmin+c}  & 1781.61 & 1780.79 & 1775.50 & 1771.13 & 1767.13 & 1762.65 & 1780.57 \\
                                    \midrule
\multirow{5}{*}{Banzhaf value}      & \emph{arbitrary} & 1781.89 & 1780.60 & 1775.71 & 1771.46 & 1767.50 & 1763.31 & 1781.58 \\
                                    & \emph{d1}        & 1782.14 & 1781.14 & 1775.34 & 1771.91 & 1767.57 & 1763.15 & 1781.83 \\
                                    & \emph{d1+c}      & 1781.67 & 1780.33 & 1774.87 & 1771.37 & 1766.51 & 1762.34 & 1780.38 \\
                                    & \emph{lexmin}    & 1782.15 & 1781.15 & 1775.62 & 1771.97 & 1767.99 & 1763.16 & 1782.46 \\
                                    & \emph{lexmin+c}  & 1781.96 & 1780.58 & 1774.87 & 1771.08 & 1766.94 & 1762.79 & 1780.66 \\
                                    \midrule
\multirow{5}{*}{Banzhaf* value}     & \emph{arbitrary} & 1781.89 & 1780.60 & 1775.71 & 1771.46 & 1767.50 & 1763.31 & 1781.58 \\
                                    & \emph{d1}        & 1782.14 & 1781.14 & 1775.34 & 1771.91 & 1767.57 & 1763.15 & 1781.83 \\
                                    & \emph{d1+c}      & 1781.77 & 1780.03 & 1774.83 & 1771.38 & 1767.00 & 1762.71 & 1781.08 \\
                                    & \emph{lexmin}    & 1782.15 & 1781.15 & 1775.62 & 1771.97 & 1767.99 & 1763.16 & 1782.46 \\
                                    & \emph{lexmin+c}  & 1782.04 & 1781.03 & 1774.78 & 1771.32 & 1767.24 & 1762.70 & 1781.24 \\
                                    \midrule
\multirow{5}{*}{Tau value}          & \emph{arbitrary} & 1781.89 & 1780.60 & 1775.65 & 1771.36 & 1767.34 & 1763.62 & 1781.58 \\
                                    & \emph{d1}        & 1781.87 & 1780.67 & 1775.14 & 1771.22 & 1767.35 & 1762.70 & 1782.22 \\
                                    & \emph{d1+c}      & 1781.80 & 1780.08 & 1774.28 & 1770.34 & 1766.54 & 1761.85 & 1780.90 \\
                                    & \emph{lexmin}    & 1782.41 & 1780.63 & 1775.40 & 1771.60 & 1767.60 & 1763.42 & 1782.12 \\
                                    & \emph{lexmin+c}  & 1781.13 & 1780.01 & 1774.42 & 1770.97 & 1766.31 & 1761.81 & 1780.45\\
                                    \bottomrule
\end{tabular}
\smallskip
\caption{{Average number of kidney transplants for {\bf equal} country sizes for $n\in \{4,\ldots,10\}$ when $\ell=\infty$.}} \label{t-num-equal}
\end{table}

\setlength{\tabcolsep}{0.20cm}
\begin{table}[h]
\begin{tabular}{cllllllll}
\toprule
\multicolumn{2}{c}{Solution   concepts/n}                 & 4       & 5       & 6       & 7       & 8       & 9       & 10      \\
\midrule
\multirow{5}{*}{benefit value}      & $\emph{arbitrary}$ & 1775.56 & 1780.86 & 1765.16 & 1745.57 & 1754.08 & 1731.64 & 1763.57 \\
                                    & $\emph{d1}$        & 1775.07 & 1781.09 & 1765.15 & 1745.78 & 1753.74 & 1731.19 & 1763.39 \\
                                    & $\emph{d1+c}$      & 1774.24 & 1780.35 & 1763.89 & 1745.75 & 1753.46 & 1731.12 & 1763.37 \\
                                    & $\emph{lexmin}$    & 1775.51 & 1781.39 & 1764.70 & 1745.66 & 1753.83 & 1731.51 & 1763.73 \\
                                    & $\emph{lexmin+c}$  & 1774.38 & 1780.40 & 1763.87 & 1744.95 & 1752.83 & 1730.72 & 1762.45 \\
                                    \midrule
\multirow{5}{*}{contribution value} & $\emph{arbitrary}$ & 1775.56 & 1780.86 & 1765.16 & 1745.57 & 1754.08 & 1731.64 & 1763.57 \\
                                    & $\emph{d1}$        & 1775.17 & 1780.39 & 1765.10 & 1746.20 & 1753.91 & 1731.31 & 1763.45 \\
                                    & $\emph{d1+c}$      & 1774.09 & 1780.19 & 1763.82 & 1745.31 & 1753.04 & 1730.92 & 1762.90 \\
                                    & $\emph{lexmin}$    & 1775.52 & 1780.88 & 1765.52 & 1746.61 & 1754.36 & 1732.38 & 1763.80 \\
                                    & $\emph{lexmin+c}$  & 1773.54 & 1778.83 & 1763.59 & 1744.50 & 1752.82 & 1730.44 & 1761.84 \\
                                    \midrule
\multirow{5}{*}{Nucleolus}          & $\emph{arbitrary}$ & 1774.12 & 1781.25 & 1764.39 & 1746.27 & 1753.55 & 1731.72 & 1763.66 \\
                                    & $\emph{d1}$        & 1774.18 & 1780.85 & 1764.08 & 1746.64 & 1753.38 & 1731.13 & 1764.02 \\
                                    & $\emph{d1+c}$      & 1773.08 & 1779.98 & 1763.28 & 1745.61 & 1752.75 & 1730.71 & 1762.79 \\
                                    & $\emph{lexmin}$    & 1774.23 & 1781.34 & 1764.32 & 1746.04 & 1753.41 & 1731.24 & 1763.74 \\
                                    & $\emph{lexmin+c}$  & 1772.92 & 1780.06 & 1762.92 & 1745.32 & 1751.82 & 1730.22 & 1762.19 \\
                                    \midrule
\multirow{5}{*}{Shapley value}      & $\emph{arbitrary}$ & 1774.02 & 1781.25 & 1764.42 & 1746.17 & 1753.44 & 1731.64 & 1764.04 \\
                                    & $\emph{d1}$        & 1774.38 & 1781.17 & 1763.92 & 1746.27 & 1753.60 & 1731.43 & 1763.09 \\
                                    & $\emph{d1+c}$      & 1773.47 & 1779.94 & 1763.70 & 1745.30 & 1753.14 & 1731.43 & 1762.99 \\
                                    & $\emph{lexmin}$    & 1774.24 & 1781.50 & 1764.67 & 1745.89 & 1753.84 & 1731.73 & 1763.71 \\
                                    & $\emph{lexmin+c}$  & 1773.51 & 1780.54 & 1764.09 & 1745.53 & 1752.84 & 1730.54 & 1762.60 \\
                                    \midrule
\multirow{5}{*}{Banzhaf value}      & $\emph{arbitrary}$ & 1775.90 & 1780.86 & 1765.00 & 1745.64 & 1754.25 & 1731.39 & 1763.24 \\
                                    & $\emph{d1}$        & 1775.17 & 1781.00 & 1764.38 & 1746.30 & 1754.43 & 1731.38 & 1763.34 \\
                                    & $\emph{d1+c}$      & 1775.25 & 1780.41 & 1764.60 & 1745.27 & 1753.50 & 1731.05 & 1762.85 \\
                                    & $\emph{lexmin}$    & 1775.43 & 1781.02 & 1765.33 & 1745.74 & 1754.33 & 1732.17 & 1763.47 \\
                                    & $\emph{lexmin+c}$  & 1774.90 & 1780.01 & 1764.23 & 1745.56 & 1753.40 & 1731.02 & 1762.85 \\
                                    \midrule
\multirow{5}{*}{Banzhaf* value}     & $\emph{arbitrary}$ & 1775.90 & 1780.86 & 1765.00 & 1745.64 & 1754.25 & 1731.39 & 1763.24 \\
                                    & $\emph{d1}$        & 1775.17 & 1781.00 & 1764.38 & 1746.30 & 1754.43 & 1731.38 & 1763.34 \\
                                    & $\emph{d1+c}$      & 1774.62 & 1780.69 & 1764.32 & 1745.91 & 1753.87 & 1731.27 & 1763.09 \\
                                    & $\emph{lexmin}$    & 1775.43 & 1781.02 & 1765.33 & 1745.74 & 1754.33 & 1732.17 & 1763.47 \\
                                    & $\emph{lexmin+c}$  & 1775.02 & 1780.43 & 1764.81 & 1745.63 & 1753.82 & 1731.49 & 1762.97 \\
                                    \midrule
\multirow{5}{*}{Tau value}          & $\emph{arbitrary}$ & 1775.46 & 1780.86 & 1765.12 & 1745.75 & 1754.50 & 1731.55 & 1763.27 \\
                                    & $\emph{d1}$        & 1774.72 & 1780.85 & 1764.86 & 1745.48 & 1753.78 & 1731.42 & 1763.54 \\
                                    & $\emph{d1+c}$      & 1774.50 & 1779.87 & 1763.82 & 1746.00 & 1753.06 & 1731.11 & 1762.99 \\
                                    & $\emph{lexmin}$    & 1774.90 & 1781.11 & 1764.47 & 1745.77 & 1754.02 & 1731.48 & 1763.62 \\
                                    & $\emph{lexmin+c}$  & 1774.46 & 1779.77 & 1763.52 & 1745.32 & 1753.29 & 1730.23 & 1761.84
                                    \\
                                    \bottomrule
\end{tabular}
\smallskip
\caption{{Average number of kidney transplants for {\bf varying} country sizes for $n\in \{4,\ldots,10\}$ when $\ell=\infty$.}}\label{t-num-varying}
\end{table}
\end{document}